\documentclass[11pt,a4paper,final]{article}
\usepackage[utf8]{inputenc}
\usepackage[english]{babel}
\usepackage{stackengine}
\usepackage{scalerel}
\usepackage{amsmath}
\usepackage{amsfonts}
\usepackage{amssymb}
\usepackage{appendix}
\usepackage{faktor}
\usepackage[hidelinks]{hyperref}
\hypersetup{
	colorlinks=true,
	urlcolor=blue,
	linkcolor=blue,
	citecolor=red,
}
\usepackage{amsthm}
\usepackage{bbold}
\usepackage{mathrsfs}
\usepackage{wasysym}
\usepackage{mathtools}
\usepackage{bm}
\usepackage{cancel}
\usepackage{color}
\usepackage{enumitem,multicol}
\usepackage{multicol}
\usepackage{tikz}
\usepackage{tikz-cd}
\usepackage[Glenn]{fncychap}
\usepackage{subfig}
\usepackage{accents}
\usepackage{slashed}
\usepackage{psfrag}
\usepackage{stackengine}

\usetikzlibrary{babel}
\makeatletter
\newcommand*\owedge{\mathpalette\@owedge\relax}
\newcommand*\@owedge[1]{%
	\mathbin{%
		\ooalign{%
			$#1\m@th\bigcirc$\cr
			\hidewidth$#1\m@th\wedge$\hidewidth\cr
		}%
	}%
}
\makeatother

\usepackage{cjhebrew}

\newcounter{mnotecount}

\newcommand{\mnote}[1]
{\protect{\stepcounter{mnotecount}}$^{\mbox{\tiny
			$\,\bullet$\themnotecount}}$ \marginpar{
		\raggedright\tiny\em
		$\,\bullet$\themnotecount: #1} }

\newtheorem{teo}{Theorem}[section]
\newtheorem{cor}[teo]{Corollary}
\newtheorem{prop}[teo]{Proposition}
\newtheorem{lema}[teo]{Lemma}
\newtheorem{defi}[teo]{Definition}

\newtheorem{rmk}[teo]{Remark}
\newtheorem{nota}[teo]{Notation}
\newtheorem{con}[teo]{Conjecture}

\usepackage{makeidx}
\usepackage{graphicx}
\usepackage[left=2.60cm, right=2.60cm, top=2.45cm, bottom=2.80cm]{geometry}

\usepackage{scalerel}
\usepackage{stackengine,wasysym}

\makeatletter
\newcommand{\douwidehat}[2]{%
	\sbox0{$\m@th#1\widehat{\hphantom{#2}}$}%
	\sbox2{$\m@th#1x$}
	\sbox4{$\m@th#1#2$}
	\dimen0=\ht0
	\advance\dimen0 -.8\ht2
	\dimen2=\dp4
	\rlap{%
		\raisebox{\dimexpr\dimen0-\dimen2}{%
			\scalebox{1}[-1]{\box0}%
		}%
	}%
	{#2}%
}
\makeatother

\usepackage{color}
\makeatletter
\renewcommand\part{%
	\if@openright
	\cleardoublepage
	\else
	\clearpage
	\fi
	\thispagestyle{empty}%
	\if@twocolumn
	\onecolumn
	\@tempswatrue
	\else
	\@tempswafalse
	\fi
	\null\vfil
	\secdef\@part\@spart}
\makeatother
\renewcommand{\H}{\mc H}
\newcommand{\uH}{\mc{\ul H}}
\newcommand{\pH}{\partial\mc H}
\newcommand{\puH}{\partial\mc{\ul H}}
\newcommand{\kt}{\mf{K}}
\newcommand{\ktt}{\mathbb{k}}
\renewcommand{\k}{\mc{K}}

\newcommand{\mS}{\mathbb{\Sigma}}

\newcommand{\wh}{\widehat}

\newcommand{\ul}{\underline}

\newcommand{\elltwo}{\ell^{(2)}}
\newcommand{\ntwo}{n^{(2)}}
\newcommand{\nablacero}{\accentset{\circ}{\nabla}}

\newcommand{\bY}{\textup{\textbf Y}}
\newcommand{\Y}{\textup{Y}}
\newcommand{\bT}{\textup{\textbf T}}
\newcommand{\T}{\textup{T}}
\newcommand{\bZ}{\textup{\textbf Z}}
\newcommand{\bz}{\textup{\textbf z}}
\newcommand{\z}{\textup{z}}
\newcommand{\Z}{\textup{Z}}
\newcommand{\bU}{\textup{\textbf U}}
\newcommand{\U}{\textup{U}}

\newcommand{\bF}{\textup{\textbf F}}
\newcommand{\F}{\textup{F}}
\newcommand{\br}{\textup{\textbf r}}
\renewcommand{\r}{\textup{r}}

\newcommand{\bs}{\textup{\textbf s}}
\newcommand{\s}{\textup{s}}
\newcommand{\bK}{\textup{\textbf K}}
\newcommand{\K}{\textup{K}}
\newcommand{\bPi}{\bm\Pi}
\newcommand{\bg}{\bm{\gamma}}

\newcommand{\mf}{\mathfrak}
\newcommand{\real}{\mathbb R}
\newcommand{\tr}{\operatorname{tr}}

\newcommand{\R}{\bm{\mc R}}
\newcommand{\st}{\stackrel}
\newcommand{\hess}{\operatorname{Hess }}

\newcommand{\para}{\parallel}
\newcommand{\Rcero}{\accentset{\circ}{R}}

\newcommand{\X}{\mathfrak{X}}
\renewcommand{\d}{\coloneqq}

\newcommand{\riem}{\operatorname{Riem}}
\newcommand{\ric}{\operatorname{Ric}}
\newcommand{\ein}{\operatorname{Ein}}
\newcommand{\lie}{\mathcal{L}}
\newcommand{\mc}{\mathcal}

\renewcommand{\to}{\longrightarrow}

\title{\vspace*{-1.35cm}\textbf{Killing and homothetic initial data for general hypersurfaces}}
\author{Marc Mars\footnote{\href{mailto:marc@usal.es}{marc@usal.es}}\,\, and Gabriel Sánchez-Pérez\footnote{Corresponding author, \href{mailto:gasape21@usal.es}{gasape21@usal.es}} \\
	Departamento de Física Fundamental, Universidad de Salamanca\\
	Plaza de la Merced s/n, 37008 Salamanca, Spain}
\date{\today}

\begin{document}
	\maketitle
	\begin{abstract}
In this paper we present a collection of general identities relating the deformation tensor $\k=\lie_{\eta}g$ of an arbitrary vector field $\eta$ with the tensor $\Sigma=\lie_{\eta}\nabla$ on an abstract hypersurface $\mc H$ of any causal character. As an application we establish necessary conditions on $\mc H$ for the existence of a homothetic Killing vector on the spacetime where $\mc H$ is embedded. The sufficiency of these conditions is then analysed in three specific settings. For spacelike hypersurfaces, we recover the well-known homothetic KID equations \cite{coll1977evolution,garcia2019conformal} in the language of hypersurface data. For two intersecting null hypersurfaces, we generalize a previous result \cite{chrusciel2013kids}, valid for Killings, to the homothetic case and, moreover, demonstrate that the equations can be formulated solely in terms of the initial data for the characteristic Cauchy problem, i.e., without involving a priori spacetime quantities. This puts the characteristic KID problem on equal footing with the spacelike KID problem. Furthermore, we highlight the versatility of the formalism by addressing the homothetic KID problem for smooth spacelike-characteristic initial data. Other initial value problems, such as the spacelike-characteristic with corners, can be approached similarly.
	\end{abstract}
	\section{Introduction}
	
Killing and homothetic vector fields are of an utmost importance in general relativity, as they encode the symmetries of the spacetime, providing crucial insights into its physical and geometrical properties. Killing vectors are associated with isometries of the spacetime, leading to conserved quantities and corresponding, for instance, to equilibrium states of the gravitational field. Homothetic vectors describe scale invariances, allowing the study of self-similar solutions of the Einstein equations \cite{rodnianski}, which arise, for instance, in critical gravitational collapse \cite{gundlach2007critical}. Homotheties also play a prominent role in the Fefferman-Graham approach to conformal geo\-metry by means of the ambient space construction \cite{fefferman2012ambient}, which is also of significant importance in several physically motivated research areas such as holography (see \cite{parisini2024ambient} and references therein). One particularly interesting problem is how to encode the existence of a symmetry at the initial data level. This allows for instance to transform a problem from a Lorentzian setting in dimension $\mf n+1$ into a simpler problem in dimension $\mf n$. This translation is often the right setup to tackle several interesting problems, e.g. the existence and uniqueness of black holes in equilibrium (see \cite{robinson2009four,sansom2023dain}).  \\ 
	

There are various ways to prescribe initial data for the Einstein equations. One common \mbox{approach} involves specifying an abstract $\mf n$-dimensional Riemannian manifold $(\mc H,h)$ along with a symmetric, two covariant tensor field $K$, subject to the constraint equations. The foundational result of Choquet-Bruhat \cite{Choquet,choquet1969global} establishes that such initial data gives rise to an \mbox{$(\mf n+1)$-dimensional} $\lambda$-vacuum spacetime where the data is embedded (see also \cite{ringstrom}). \mbox{Another} approach is to prescribe initial data on a pair of characteristic hypersurfaces. As shown by Rendall \cite{Rendall} it suffices to specify the spacetime metric (again subject to suitable constraint equations) on the hypersurfaces and the torsion one-form on their intersection to construct a $\lambda$-vacuum spacetime in a future neighbourhood of the intersection. This result has been extended in subsequent works to cover the entire neighbourhood of the two hypersurfaces \cite{Luk,chrusciel2023neighborhood}. In a recent work \cite{Mio1,Mio2} we have developed a formalism to define two intersecting null \mbox{hypersurfaces} from a detached viewpoint via the notion of \textit{double null data}, thereby placing the characteristic Cauchy problem on the same footing as the spacelike one. The initial value problem is also known to be well-posed with initial data on a light-cone \cite{cone}. \\


The connection between initial value problems in general relativity and the existence of (homothetic) Killing vector fields in the bulk spacetime leads to the so-called homothetic Killing initial data (KID) problem. This question has received considerable attention over the years. Here we only make a brief summary of results, and direct the reader to the cited works for details and additional references. For spacelike initial data, the classic strategy \cite{moncrief1975spacetime,coll1977evolution,beig1997killing} begins by prescribing the initial value of the candidate to Killing vector $\eta|_{\mc H}$ and imposing the so-called KID equations, which involve solely the abstract data $(h,K)$ as well as the va\-lues of the tangential component of $\eta|_{\H}$ (a vector) and its transversal component (a scalar). These equations ensure that, once the spacetime is constructed, the conditions $\Phi^{\star}\lie_{\eta}g=0$ and $\Phi^{\star}\big(\lie_{\nu}\lie_{\eta}g\big)=0$ hold, where $\nu$ is a normal of $\mc H$. The next step is to specify $\lie_{\nu} \eta|_{\mc H}$ such that the remaining components of $\lie_{\eta} g$ vanish at $\mc H$. The vector field $\eta$ is then extended off $\mc H$ by solving a homogeneous wave-type equation with the prescribed initial conditions $(\eta|_{\mc H}, \lie_{\nu} \eta|_{\mc H})$. The wave-type equation must of course be compatible with $\eta$ being a Killing vector and must have the essential properties of (i) make sure that the remaining components of $\lie_{\nu}\lie_{\eta}g$ at $\mc H$ vanish identically and (ii) that $\lie_{\eta}g$ also satisfies a homogeneous wave-type equation. Combining both things one concludes that $\eta$ is a Killing vector of the Cauchy development of $(\mc H,h,K)$. This strategy also applies to homothetic and conformal Killing vectors \cite{garcia2019conformal}. \\

In the null case, the KID problem has been analysed for initial data posed on light-cones and pairs of intersecting characteristic hypersurfaces, both for the Einstein equations \cite{chrusciel2013kids} and for the conformal Einstein equations \cite{paetz2014kids} (restricted to four dimensions if one of the intersec\-ting characteristic hypersurfaces is null infinity). The strategy put forward in \cite{chrusciel2013kids} was to consider two characteristic hypersurfaces embedded in a $\lambda$-vacuum spacetime and adapt spacetime \mbox{coordinates} to the initial characteristic hypersurfaces. The initial value of $\eta$ are the components of this vector in this adapted coordinate system. The extension of $\eta$ is also achieved using a wave-type equation, and sufficient conditions are derived to ensure that $\lie_{\eta} g$ vanishes on the initial hypersurfaces. Since $\lie_{\eta} g$ satisfies a wave-type equation as well, this implies that $\eta$ is a Killing of the spacetime metric $g$. \\

A key difference between the spacelike and null KID problems is that, in the former, the KID equations are formulated entirely in terms of the abstract data, enabling to decide whether an initial data set gives rise to a spacetime admitting a Killing vector \textit{before} solving the Einstein equations. In the null case, however, the KID equations depend on the ambient Levi-Civita connection and Riemann tensor, making it necessary to first solve the Einstein equations and subsequently verify whether the KID equations are satisfied or not.\\

Motivated by this issue, in this paper we find a collection of completely general identities relating the deformation tensor $\k[\eta]=\lie_{\eta}g$ of an arbitrary vector field $\eta$ to the associated tensor $\Sigma[\eta]=\lie_{\eta}\nabla$ on an abstract hypersurface of any causal character. These identities are fully gauge and diffeomorphism covariant and can be applied regardless the ambient field equations and the properties of $\eta$. As an application, we derive the KID equations for characteristic initial data and show that they are at the level of the initial data for the Einstein equations. This result puts the KID problem for two null hypersurfaces on equal footing with the standard KID problem for spacelike hypersurfaces. Although this application is our main interest in this paper, we emphasize that the identities we derive are obtained in full generality and their range of applicability extends far beyond this concrete example. In fact, had our interest been just working with homotheties (or Killings), several of the computations could have been made simpler. We believe that this extra effort is worth doing in view of other potential applications of the identities, e.g. to deal with conformal Killing or more generally in any situation where one has a priori information on the metric Lie derivative tensor $\k[\eta]$ or the  connection Lie derivative tensor $\Sigma[\eta]$ of a vector field $\eta$.\\

In order to deal with abstract hypersurfaces it is convenient to employ the so-called hypersurface data formalism, as it provides a detached framework to analyse hypersurfaces of arbitrary causal character. Given a smooth manifold $\mc H$, hypersurface data consists of a 2-covariant, symmetric tensor field $\bg$, a one-form $\bm\ell$, a scalar function $\elltwo$ and a 2-covariant, symmetric tensor $\bY$. When $\mc H$ happens to be embedded in an ambient manifold $(\mc M,g)$, the tensor $\bg$ corresponds to its first fundamental form, $\bm\ell$ and $\elltwo$ capture the transverse-tangent and transverse-transverse information of $g$ at $\mc H$, while $\bY$ coincides with the first derivative of $g$ along a transverse vector $\xi$ (also known as \textit{rigging}) at $\mc H$. From the set $\{\bg,\bm\ell,\elltwo\}$ one can construct a torsion-free connection $\nablacero$ and the dual set $\{P,n,\ntwo\}$ consisting of a 2-contravariant, symmetric tensor $P$, a vector $n$ and a scalar $\ntwo$. Hypersurface data is equipped with an internal gauge freedom that encodes, at the abstract level, all possible riggings of $\mc H$. Hypersurface data satisfying $\ntwo=0$ corresponds to an abstract null hypersurface with $n$ as its null generator.\\

Given a semi-Riemannian manifold $(\mc M,g)$ and any vector field $\eta$ it is possible to relate the deformation tensor of $\eta$, $\mc K[\eta]_{\mu\nu}=\lie_{\eta}g_{\mu\nu}$, with the so-called Lie derivative of the connection $\Sigma[\eta]=\lie_{\eta}\nabla$, and various others geometric quantities associated to $g$. By pulling back these identities to a hypersurface $\mc H$, we obtain explicit relations connecting geometric quantities associated to $\k[\eta]_{\mu\nu}$ at $\mc H$ with the fully tangential components of $\Sigma[\eta]$ at $\mc H$, as well as with its one transverse-two tangent components (denoted by $\mS_{ab}$). These identities are fully gauge and diffeomorphism covariant, do not assume any field equations for $g$ and hold for hypersurfaces of any causal character, as well as for any vector field $\eta$ (not necessarily a homothety or a Killing vector). Applying them to spacelike hypersurfaces, we recover the homothetic KID problem for spacelike initial data in the language of hypersurface data (Theorem \ref{teo_espacial}). A key advantage of this approach is its gauge freedom, allowing one to adapt the rigging vector to specific requirements (in particular the rigging need not be normal to $\mc H$). \\



Inspired by the ideas in \cite{chrusciel2013kids}, for null hypersurfaces we establish that (i) the tangential components of $\k[\eta]$ at $\mc H$ are expressible in terms of hypersurface data, and (ii) the transverse-tangent ($\kt_a$) and transverse-transverse ($\ktt$) components satisfy a hierarchical, inhomogeneous system of transport equations along the null generator $n$, with a source determined by $\mS$ (see Lemma \ref{lemaagrupado}). Furthermore, Corollary \ref{cor_junto} demonstrates that specific contractions of $\mS$, namely $P^{ab}\mS_{ab}$ and $\mS_{ab}n^b$ (but not $\mS(n,n)$), also satisfy transport equations. All in all, this results in a closed hierarchical system for $\{\kt,\ktt,\mS\}$. We prove that this system becomes homogeneous provided $\k_{ab}$ is pure trace (i.e., proportional to $\gamma_{ab}$) and $\mS(n,n)=0$. These two conditions are the natural replacement for the KID equations in the null case. We use this fact in Theorem \ref{teorema_null} to propagate $\{\kt,\ktt,P^{ab}\mS_{ab},\mS_{ab}n^b\}$ from initial conditions at a cross-section under the assumption that $\mc H$ admits a product topology.\\

One direct application of this formalism is the characteristic KID problem. For two intersecting null hypersurfaces $\H$ and $\uH$ embedded in a $\lambda$-vacuum spacetime, the intersection surface serves as initial cross-section to propagate $\{\kt,\ktt,\mS\}$ using Theorem \ref{teorema_null}. It turns out that only by requiring $\mS(n,n)=0$ and $\k_{ab}$ to be pure trace on the two null hypersurfaces, as well as making sure that the scalar $\kt(n)$ and the one-form $\mS_{ab}n^b$ take suitable values at the intersection, the deformation tensor of $\eta$ satisfies $\lie_{\eta}g=\mu g$ at $\mc H\cup\mc{\ul H}$, and hence everywhere. The crucial point (and the key difference with \cite{chrusciel2013kids} besides the fact that we deal with homotheties and not just Killings) is that all these conditions are written solely in terms of initial data (double null data), and hence they can be ascertained, or imposed, \textit{a priori} (i.e., before solving the Einstein equations). This result puts the characteristic KID problem on the same footing as the standard KID problem. \\

Other initial value problems with data posed on hypersurfaces that are a mixture of the pre\-vious ones can be treated similarly. As an example, we examine the homothetic KID problem for smooth spacelike-characteristic initial data (assuming that this is well-posed, see Conjectures \ref{conj0}-\ref{conj}). By imposing the standard KID equations on the spacelike region and the null KID equations on the characteristic region (both written in terms of detached data), we show that the Cauchy development of the initial data admits a homothetic Killing field (see Theorem \ref{teo_smooth}). A similar conclusion holds when initial data is given on a spacelike-characteristic hypersurface with corner \cite{chrusciel2015characteristic,czimek2022spacelike}. In this case, no additional conditions on the ``corner'' are required. For the sake of brevity, this result is not discussed and proved explicitly in the text and its validity is only mentioned (without proof) in the conclusions.\\

This paper is organized as follows. Section \ref{sec_hypersurfacedata} provides an overview of the fundamental aspects of hypersurface data. In Section \ref{sec_general}, we review several identities relating $\k[\eta]$ and $\Sigma[\eta]$ for an arbitrary vector field $\eta$. Section \ref{sec_hip} analyses these identities in the context of embedded hypersurface data of arbitrary causal character. Focusing on the non-null case, in Section \ref{sec_nonnull} we revisit the spacelike homothetic KID problem in the language of hypersurface data. Section \ref{sec_null} focuses on the null case, proving that the identities of Section \ref{sec_hip} lead to transport equations for specific components of $\k[\eta]$ and $\Sigma[\eta]$ along the null generator of the hypersurface. After reviewing the essential elements of double null data in Section \ref{sec_DND}, Section \ref{sec_charact} presents the homothetic KID problem for two characteristic null hypersurfaces. Finally, in Section \ref{sec_sp-ch} we illustrate another application of the formalism by addressing the smooth spacelike-characteristic problem. Additionally, three appendices complement the main text: Appendix \ref{app} provides the computation of pullbacks of ambient tensor fields to arbitrary hypersurfaces; Appendix \ref{app_proof} contains the (somewhat long) proof of Proposition \ref{prop_lienSigma}; and Appendix \ref{app_conn} details several contractions used throughout the paper.

	\section*{Notation and conventions}
	
All manifolds in this paper are assumed to be connected and smooth and, depending on convenience, both index-free and abstract index notation are used to denote tensorial operations. Spacetime indices are denoted with Greek letters, indices on a hypersurface are written in lowercase Latin letters, and indices at cross-sections of a hypersurface are expressed in uppercase Latin letters. As usual, square brackets enclosing indices denote antisymmetrization and parenthesis are for symmetrization. The symmetrized tensor product is denoted with $\otimes_s$. By $\mc F(\mc M)$, $\X(\mc M)$ and $\X^{\star}(\mc M)$ we denote, respectively, the set of smooth functions, vector fields and one-forms on a manifold $\mc M$. The subset $\mc F^{\star}(\mc M)\subset\mc F(\mc M)$ consists of the nowhere vanishing functions on $\mc M$. The pullback of a function $f$ via a diffeomorphism $\Phi$ will be denoted by $\Phi^{\star}f$ or simply by $f$ if the specific meaning can be inferred from the context. A $(p,q)$-tensor refers to a tensor field $p$ times contravariant and $q$ times covariant. Given any pair of $(2,0)$ and $(0,2)$ tensors $A^{ab}$ and $B_{cd}$ we denote $\tr_A \bm B \d A^{ab}B_{ab}$. We employ the symbol $\nabla$ for the Levi-Civita connection of $g$. A condition in parenthesis right next to an identity means that the identity has been obtained under the assumption that the condition holds. Finally, we say that an $(\mf n+1)$-dimensional manifold $(\mc M,g)$ with $\mf n\ge 2$ satisfies the $\lambda$-vacuum equations provided that 
\begin{equation}
	\label{notacion}
\ric[g]=\lambda g, \qquad \mbox{ or equivalently, }\qquad  \ein[g] = \frac{1-\mf{n}}{2}\lambda g.
\end{equation}

	\section{Review of hypersurface data formalism}
	\label{sec_hypersurfacedata}
	
In this section we summarize the basic notions of the \textit{hypersurface data formalism}. Further details can be found in \cite{Marc1,Marc2,Marc3}. Let us start by introducing the necessary objects for this paper.
	
	\begin{defi}
		\label{def_hypersurfacedata}
		Let $\mc H$ be an $\mf n$-dimensional manifold, $\bg$ a symmetric (0,2)-tensor field, $\bm\ell$ a one-form and $\elltwo$ a scalar function on $\mc H$. We say that $\{\mc H,\bg,\bm\ell,\elltwo\}$ defines \textbf{metric hypersurface data} provided that the (0,2) symmetric tensor $\bm{\mc A}|_p$ on $T_p\mc H\times\real$ defined by 
		\begin{equation}
			\label{def_A}
			\mc A|_p\left((W,a),(V,b)\right) \d \bg|_p (W,V) + a\bm\ell|_p(V)+b\bm\ell|_p(W)+ab\ell^{(2)}|_p
		\end{equation}
		is non-degenerate at every $p\in\mc H$. A five-tuple $\{\mc H,\bg,\bm\ell,\elltwo,\bY\}$, where $\bY$ is a (0,2) symmetric tensor field on $\mc H$, is called \textbf{hypersurface data}.
	\end{defi}
	The non-degeneracy of $\bm{\mc A}$ allows us to introduce its ``inverse'' $\bm{\mc A}^{\sharp}$ by $\bm{\mc A}^{\sharp}\big(\bm{\mc A}((V,a),\cdot),\cdot\big)=(V,a)$ for every $(V,a)\in\X(\mc H)\otimes\mc{F}(\mc H)$. From $\bm{\mc A}^{\sharp}$ one can define a $(2,0)$ symmetric tensor field $P$, a vector $n$ and a scalar $n^{(2)}$ on $\mc H$ by the decomposition 
	\begin{equation}
		\label{def_Asharp}
		\mc A^{\sharp}\left((\bm\alpha,a),(\bm\beta,b)\right) = P (\bm\alpha,\bm\beta)+a n(\bm\beta)+bn(\bm\alpha)+ab n^{(2)} \quad\, \forall\, (\bm\alpha,a),(\bm\beta,b)\in\X^{\star}(\mc H)\times F(\mc H).
	\end{equation} 
	Equivalently, $P$, $n$ and $n^{(2)}$ can be defined by
	\begin{multicols}{2}
		\noindent
		\begin{align}
			\gamma_{ab}n^b + n^{(2)}\ell_a&=0,\label{gamman}\\
			\ell_an^a+n^{(2)}\ell^{(2)}&=1,\label{ell(n)}
		\end{align}
		\begin{align}
			P^{ab}\ell_b+\ell^{(2)} n^a&=0,\label{Pell}\\
			P^{ac}\gamma_{cb} + \ell_b n^a &=\delta^a_b.\label{Pgamma}
		\end{align}
	\end{multicols}
The main idea behind Definition \ref{def_hypersurfacedata} is to encode the geometric properties of hypersurfaces detached from the ambient manifold. The connection between this and the standard notion of a hypersurface is as follows.
	\begin{defi}
		\label{defi_embedded}
		Metric hypersurface data $\{\mc H,\bg,\bm\ell,\ell^{(2)}\}$ is $\bm{(\Phi,\xi)}$\textbf{-embedded} in a semi-Rie\-mannian manifold $(\mc M,g)$ if there exists an embedding $\Phi:\mc H\hookrightarrow\mc M$ and a vector field $\xi$ along $\Phi(\mc H)$ everywhere transversal to $\Phi(\mc H)$, called rigging, such that
		\begin{equation}
			\label{embedded_equations}
			\Phi^{\star}(g)=\bg, \hspace{0.5cm} \Phi^{\star}\left(g(\xi,\cdot)\right) = \bm\ell,\hspace{0.5cm} \Phi^{\star}\left(g(\xi,\xi)\right) = \ell^{(2)}.
		\end{equation}
		Hypersurface data $\{\mc H,\bg,\bm\ell,\ell^{(2)},\bY\}$ is $(\Phi,\xi)$-embedded provided that, in addition, 
		\begin{equation}
			\label{Yembedded}
			\dfrac{1}{2}\Phi^{\star}\left(\lie_{\xi} g\right) = \bY.
		\end{equation}
	\end{defi}
	In the context of embedded hypersurface data, $\bg$ being degenerate is equivalent to $\Phi(\mc H)$ being an embedded null hypersurface. It is easy to show \cite{Marc2} that the degeneracy of $\bg$ is equivalent to $n^{(2)}=0$. Hence, hypersurface data satisfying $n^{(2)}=0$ are called \textbf{null hypersurface data}. In such case, a cross-section (or simply a section) $\mc S$ of $\mc H$ is an embedded hypersurface $\mc S\hookrightarrow\mc H$ with the property that every integral curve of $n$ crosses $\mc S$ exactly once.\\
	

	Given hypersurface data one can define the symmetric (0,2) tensor fields\\
	
	\begin{minipage}{0.5\textwidth}
	\begin{equation}
		\label{defU}
		\bU\d \dfrac{1}{2}\lie_n\bg + \bm\ell\otimes_s d\ntwo,
	\end{equation} 
\end{minipage}
\begin{minipage}{0.5\textwidth}
	\begin{equation}
		\label{defK}
	\bK\d \bU+\ntwo\bY.	
	\end{equation}
\end{minipage}
\vspace{0.2cm}

	When the data is embedded $\bK$ corresponds to the second fundamental form of $\Phi(\mc H)$ w.r.t the unique normal one-form $\bm{\nu}$ satisfying $\bm{\nu}(\xi)=1$ \cite{Marc1}. It is also convenient to introduce the $(0,2)$ tensors\\
	
	\begin{minipage}{0.5\textwidth}
		\begin{equation}
			\label{def_F}
			\bF\d \dfrac{1}{2}d\bm\ell,
		\end{equation}
	\end{minipage}
	\begin{minipage}{0.5\textwidth}
		\begin{equation}
			\label{defPi}
			\bPi \d \bY+\bF.
		\end{equation}
	\end{minipage}
	\vspace{0.2cm}
	
	Let $\{\mc H,\bg,\bm\ell,\elltwo\}$ be metric hypersurface data $(\Phi,\xi)$-embedded in $(\mc M,g)$. By the transversality of $\xi$, given a (local) basis $\{e_a\}$ on $\mc H$, the set $\{\wh{e}_a\d\Phi_{\star}e_a,\xi\}$ is a (local) basis of $T_p\mc M$ for all $p\in\Phi(\mc H)$. We write the dual basis as $\{\bm\theta^a,\bm\nu\}$ and note that by construction $\bm\nu$ agrees with the normal mentioned before. Raising indices we can introduce $\nu\d g^{\sharp}(\bm\nu,\cdot)$ and $\theta^a \d g^{\sharp}(\bm\theta^a,\cdot)$, which as a consequence of \eqref{gamman}-\eqref{Pgamma} are given in terms of $\{\xi, \wh{e}_a\}$ by 
	\begin{equation}
		\label{nuthetaa}
		\nu =  \ntwo\xi + n^a \wh{e}_a,\qquad \theta^a = P^{ab}\wh{e}_b + n^a\xi.
	\end{equation}
For the sake of simplicity we will often abuse the notation and denote $\wh{e}_a$ by $e_a$. From \eqref{def_Asharp} the inverse metric $g^{\alpha\beta}$ at points on $\mc H$ can be written in the basis $\{\xi, {e}_a\}$ as
	\begin{equation}
		\label{inversemetric}
		g^{\alpha\beta}\st{\mc H}{=}	P^{ab}e_a^{\alpha} e_b^{\beta} + n^{a} e_a^{\alpha}\xi^{\beta} + n^{b} e_b^{\beta}\xi^{\alpha} +\ntwo \xi^{\alpha}\xi^{\beta} .
	\end{equation}
The rigging vector is not unique, since given a rigging $\xi$ any other vector of the form $\xi' = z(\xi+\Phi_{\star}V)$ with $(z,V)\in\mc{F}^{\star}(\mc H)\times\X(\mc H)$ is also transverse to $\Phi(\mc H)$. This is translated into an abstract gauge group that acts on hypersurface data by

\begin{minipage}{0.4\textwidth}
\begin{align}
\mc{G}_{(z,V)}\left(\bg \right)&\d \bg,\label{transgamma}\\
\mc{G}_{(z,V)}\left( \bm{\ell}\right)&\d z\left(\bm{\ell}+\bg(V,\cdot)\right),\label{tranfell}
\end{align}
\end{minipage}
\begin{minipage}{0.6\textwidth}
	\begin{align}
\mc{G}_{(z,V)}\big( \ell^{(2)} \big)&\d z^2\big(\ell^{(2)}+2\bm\ell(V)+\bg(V,V)\big),\label{transell2}\\
\mc{G}_{(z,V)}\left( \bY\right)&\d z \bY + \bm\ell\otimes_s d z +\dfrac{1}{2}\lie_{zV}\bg.\label{transY}
	\end{align}
\end{minipage}
\vspace{0.3cm}

Applying the Cartan identity $\lie_n = \iota_n \circ d + d\circ \iota_n$ to the one-form $\bm\ell$, taking into account $2\bF = d\bm\ell$ and using \eqref{ell(n)} one has 
\begin{equation}
	\label{lienell}
	\lie_n\bm\ell = 2\bs - d\big(\ntwo\elltwo\big),
\end{equation}
where $\bs\d\bF(n,\cdot)$. From this and \eqref{gamman}-\eqref{ell(n)} the contraction of $\bU$ with $n$ reads
	\begin{equation}
		\label{Un}
\bU(n,\cdot)=\ntwo\left(-\bs+\dfrac{1}{2}\ntwo d\elltwo\right)+\dfrac{1}{2}d\ntwo,
	\end{equation} 
 Metric hypersurface data $\{\mc H,\bg,\bm\ell,\elltwo\}$ is endowed with a torsion-free connection $\nablacero$ defined by \cite{Marc2}

	\begin{multicols}{2}
		\noindent
		\begin{equation}
			\label{nablagamma}
			\nablacero_a\gamma_{bc} = -\ell_c\U_{ab} - \ell_b\U_{ac},
		\end{equation}
		\begin{equation}
			\label{nablaell}
			\nablacero_a\ell_b  = \F_{ab} - \elltwo\U_{ab}.
		\end{equation}
	\end{multicols}
	The following immediate consequence of \eqref{nablagamma} will be used several times below
\begin{equation}
	\label{combina}
2\nablacero_{(a}\gamma_{b)c}-\nablacero_c\gamma_{ab} = -2\ell_c \U_{ab}.
\end{equation}
The action of $\nablacero$ on the contravariant data $\{P,n\}$ is given by \cite{Marc2}
\begin{align}
	\nablacero_a n^b & =\s_a n^b + P^{bc}\U_{ca}-\ntwo\big(n^b (d\elltwo)_a +P^{bc}\F_{ac}\big),\label{derivadan}\\
	\nablacero_a P^{bc} & = -\big(n^bP^{cd}+n^cP^{bd}\big) \F_{ad} - n^bn^c (d\elltwo)_a.\label{derivadaP}
\end{align}

The curvature tensor of $\nablacero$ is denoted by $\Rcero^a{}_{bcd}$ and the Ricci tensor $\Rcero_{bc}$ is the contraction $\Rcero_{bc}=\Rcero^a{}_{bac}$. 
When the data is $(\Phi,\xi)$-embedded in $(\mc M,g)$, $\nablacero$ is related with the Levi-Civita connection $\nabla$ of $g$ by 
	\begin{equation}
		\label{connections}
		\nabla_{\Phi_{\star}X}\Phi_{\star}Y \st{\mc H}{=} \Phi_{\star}\nablacero_X Y - \bY(X,Y)\Phi_{\star}n - \big(\bU+\ntwo\bY\big)(X,Y)\xi \qquad \forall X,Y\in\X(\mc H).
	\end{equation}
		Unless otherwise indicated, scalar functions related by $\Phi^{\star}$ are denoted with the same symbol. Later we will also need the $\nabla$-derivative of $\xi$ along tangent directions to $\mc H$ \cite{Marc1}. We write it in two different forms
	\begin{align}
	{e}_a^{\mu}\nabla_{\mu}\xi^{\beta}&\st{\mc H}{=} (\r-\s)_a\xi^{\beta} + P^{bc}\big(\Y_{ac}+\F_{ac}\big){e}^{\beta}_b + \dfrac{1}{2}\nu^{\beta}\nablacero_a\elltwo ,\label{nablaxi1}\\
{e}_a^{\mu}\nabla_{\mu}\xi^{\beta}&\st{\mc H}{=} \left(\r_a-\s_a +\dfrac{1}{2}\ntwo \nablacero_a\elltwo\right)\xi^{\beta} + V^b{}_a e_b^{\beta},\label{nablaxi}
	\end{align}
	where $V^b{}_a\d P^{bc}\big(\Y_{ac}+\F_{ac}\big) + \dfrac{1}{2}n^b\nablacero_a\elltwo$ and we have introduced $\br\d\bY(n,\cdot)$. The following contractions of the tensor $V$ will also be needed. They can be easily proven using \eqref{gamman}-\eqref{Pgamma},
	\begin{align}
\ell_b V^b{}_a &= -\elltwo\left(\r_a-\s_a+\dfrac{1}{2}\ntwo\nablacero_a\elltwo\right) + \dfrac{1}{2}\nablacero_a\elltwo,\label{ellV}\\
\gamma_{bc}V^b{}_a &= \Y_{ac}+\F_{ac} -\left(\r_a-\s_a + \dfrac{1}{2}\ntwo \nablacero_a\elltwo\right)\ell_c,\label{gammaV}\\
V^b{}_a n^a &= P^{bc}\big(\r_c+\s_c\big)+\dfrac{1}{2}n(\elltwo) n^b.\label{Vn}
	\end{align}
In particular, from \eqref{nablaxi} and the decomposition of $\nu$ in \eqref{nuthetaa},
\begin{align*}
\nu^{\mu}\nabla_{\mu}\xi^{\beta} &\st{\mc H}{=}	\left(\dfrac{1}{2}\ntwo n(\elltwo)-\kappa_n\right) \xi^{\beta}  + V^b{}_an^a e_b^{\beta}+\ntwo \xi^{\mu}\nabla_{\mu}\xi,
\end{align*}
where $\kappa_n\d - \bm{r}(n)$, and hence in the null case 
\begin{flalign}
\label{nablanuxiconV}
&& \nu^{\mu}\nabla_{\mu}\xi^{\beta} \st{\mc H}{=}	-\kappa_n \xi^{\beta}  + V^b{}_an^a e_b^{\beta}. &&  (\ntwo=0)
\end{flalign}
Generally speaking we will substitute \eqref{ellV}-\eqref{Vn} whenever a contraction of $V$ with $\bg$, $\bm\ell$ or $n$ appears, and leave $V$ unsubstituted otherwise. We only skip this rule if a subsequent simplification becomes more evident.\\

	As proven in \cite{Marc1} (see also \cite{Marc3}), the completely tangential components of the ambient Riemann tensor, as well as its 3-tangential, 1-transverse components can be written in terms of hypersurface data as
	\begin{equation}
		\label{ABembedded}
		R_{\alpha\beta\mu\nu}\xi^{\alpha}e^{\beta}_be^{\mu}_ce^{\nu}_d \st{\mc H}{=} A_{bcd}, \hspace{1cm} R_{\alpha\beta\mu\nu} e^{\alpha}_ae^{\beta}_be^{\mu}_ce^{\nu}_d \st{\mc H}{=} B_{abcd},
	\end{equation}
	where $A$ and $B$ are the tensors on $\mc H$ defined by 
	\begin{align}
	A_{bcd}&\d 2\nablacero_{[d}\F_{c]b} + 2\nablacero_{[d}\Y_{c]b} + \U_{b[d}\nablacero_{c]}\elltwo + 2\Y_{b[d}(\r-\s)_{c]} + \ntwo\Y_{b[d}\nablacero_{c]}\elltwo,\label{A} \\
	B_{abcd}&\d \gamma_{af}\Rcero^f{}_{bcd} + 2\ell_a\nablacero_{[d}\U_{c]b}+2\U_{a[d}\Y_{c]b} + 2\U_{b[c}\Pi_{d]a} + 2\ntwo \Y_{b[c}\Y_{d]a}.\label{B}
	\end{align}
As a consequence, the contraction of the ambient Einstein tensor with the normal $\nu$ can be computed in terms of hypersurface data as follows
	\begin{align}
G^{\alpha}{}_{\beta}\nu_{\alpha}\xi^{\beta} &\st{\mc H}{=}	P^{ac}A_{abc}n^b + \dfrac{1}{2}P^{ac}P^{bd}B_{abcd} \label{Einxi},\\
G^{\alpha}{}_{\beta}\nu_{\alpha} e_c^{\beta} &\st{\mc H}{=}	\ntwo P^{bd}A_{bcd} - A_{bcd}n^b n^d + P^{bd}B_{abcd}n^a \label{Eine}.
\end{align}		
In the null case ($\ntwo=0$), \eqref{ABembedded} and \eqref{inversemetric} imply that all the tangential components of the ambient Ricci tensor can be written in terms of hypersurface data
\begin{flalign*}
&& g^{\mu\nu} R_{\alpha \mu \beta \nu} e_a^{\alpha} e_b^{\beta}\st{\mc H}{=}B_{acbd}P^{cd}- (A_{bca}+A_{acb})n^c. && (\ntwo=0)
\end{flalign*}

For non-null hypersurfaces it is not possible to write down the tangential components of the ambient Ricci tensor solely in terms of hypersurface data due to the presence of the $\ntwo$-term in \eqref{inversemetric} and the fact that the contraction of the ambient Riemann tensor with two riggings is not expressible in terms of hypersurface data. This is a well-known fact and follows e.g. from the following general identity that relates such Riemann tensor components with second order derivatives of the metric along $\xi$ \cite{Mio3}.
\begin{lema}
	\label{propriemxixi1}
	Let $(\mc M,g)$ be a semi-Riemannian manifold and $\xi,X,Y\in\X(\mc M)$. Then,
	\begin{equation}
		\label{riemxixi}
		\riem(X,\xi,Y,\xi) = -\dfrac{1}{2}\left(\lie_{\xi}^{(2)}g - \lie_{\nabla_{\xi}\xi} g\right)(X,Y) + g(\nabla_X\xi,\nabla_Y\xi).
	\end{equation}
\end{lema}
An immediate consequence of this expression is that on any embedded hypersurface $\Phi: \mc H \hookrightarrow \mc M$ the tensor constructed from second transverse derivatives by means of $$\Phi^{\star} \big(\lie^{(2)}_{\xi}  g - \lie_{\nabla_{\xi} \xi} g \big)$$ is independent of the extension of $\xi$ off $\Phi(\mc H)$. We can now write down the pullback of $\riem(\cdot,\xi,\cdot,\xi)$ in terms of this tensor and hypersurface data. 
\begin{prop}
	\label{propriemxixi}
	Let $\{\mc H,\bg,\bm\ell,\elltwo,\bY\}$ be hypersurface data $(\Phi,\xi)$-embedded in $(\mc M,g)$. Then, for any extension of $\xi$ off $\Phi(\mc H)$,
	\begin{equation}
		\label{riemxixiH}
		\Phi^{\star}\left(\riem(\cdot,\xi,\cdot,\xi)\right) = -\dfrac{1}{2}\Phi^{\star}\big(\lie_{\xi}^{(2)}g-\lie_{\nabla_{\xi}\xi} g\big) + \bPi\cdot\bPi +\left(\br-\bs+\dfrac{1}{4}\ntwo d\elltwo\right)\otimes_s d\elltwo,
	\end{equation}	
	where $\big(\bPi\cdot\bPi\big)_{ab}\d P^{cd}\Pi_{ac}\Pi_{bd}$.
	\begin{proof}
From \eqref{nablaxi} together with \eqref{embedded_equations} it follows
		\begin{align*}
			g_{\alpha\beta} \nabla_{{e}_a}\xi^{\alpha}\nabla_{{e}_b}\xi^{\beta} & = \elltwo\left(\r_a-\s_a + \dfrac{1}{2}\ntwo\nablacero_a\elltwo\right)\left(\r_b-\s_b + \dfrac{1}{2}\ntwo\nablacero_b\elltwo\right)\\
			&\quad\,  +2\left(\r_{(a}-\s_{(a} + \dfrac{1}{2}\ntwo\nablacero_{(a}\elltwo\right)\ell_c V^c{}_{b)} + \gamma_{cd}V^c{}_a V^d{}_b \\
			&= P^{cd}\Pi_{ac}\Pi_{bd} + \left(\r_{(a}-\s_{(a} + \dfrac{1}{4}\ntwo\nablacero_{(a}\elltwo\right)\nablacero_{b)}\elltwo ,
		\end{align*}
		where in the second equality we used \eqref{ellV}-\eqref{gammaV}. Equation \eqref{riemxixiH} follows from this and \eqref{riemxixi}.
	\end{proof}
\end{prop}
This result suggests the following detached definition.
\begin{defi}
	A set $\{\mc H,\bg,\bm\ell,\elltwo,\bY,\bZ^{(2)}\}$ defines \textbf{extended hypersurface data} provided the set $\{\mc H,\bg,\bm\ell,\elltwo,\bY\}$ is hypersurface data and $\bZ^{(2)}$ is a (0,2) symmetric tensor field on $\mc H$. Moreover, an extended hypersurface data set is said to be $(\Phi,\xi)$-embedded in a semi-Riemannian manifold $(\mc M,g)$ provided $\{\mc H,\bg,\bm\ell,\elltwo,\bY\}$ is $(\Phi,\xi)$-embedded in the sense of Definition \ref{defi_embedded} and, in addition, 
	\begin{equation}
		\label{Z2embedded}
		\bZ^{(2)} = \dfrac{1}{2}\Phi^{\star}\big(\lie_{\xi}^{(2)}g-\lie_{\nabla_{\xi}\xi}g\big),
	\end{equation} 
	where $\xi$ is any extension of the rigging off $\Phi(\mc H)$. 
\end{defi}
The fully tangential components of the ambient Ricci tensor are then computable in terms of extended hypersurface data,
		\begin{align}
			R_{\alpha\beta}e^{\alpha}_a e^{\beta}_b=R_{ab} & = \left(P^{cd}e_c^{\mu} e_d^{\nu} + n^c e_c^{\mu}\xi^{\nu} + n^d e_d^{\nu}\xi^{\mu}+\ntwo\xi^{\mu}\xi^{\nu}\right)e_a^{\alpha}e_b^{\beta} R_{\alpha\mu\beta\nu}\nonumber\\
			&= P^{cd}B_{cadb} + A_{bca}n^c + A_{acb}n^c + \ntwo D_{ab},\label{ricci}
		\end{align}
		where (cf. \eqref{riemxixiH})
		\begin{align}
			D_{ab}&\d -\Z^{(2)}_{ab} + P^{cd}\Pi_{ac}\Pi_{bd} + \left(\r_{(a}-\s_{(a}+\dfrac{1}{4}\ntwo\nablacero_{(a}\elltwo\right)\nablacero_{b)}\elltwo .
		\end{align}
Solving for $\ntwo\bZ^{(2)}$ in \eqref{ricci} one has $$\ntwo \Z^{(2)}_{ab} = P^{cd}B_{cadb} + 2 A_{(a|c|b)}n^c - {R}_{ab} + \ntwo\left(P^{cd}\Pi_{ac}\Pi_{bd} + \left(\r_{(a}-\s_{(a}+\dfrac{1}{4}\ntwo\nablacero_{(a}\elltwo\right)\nablacero_{b)}\elltwo \right).$$ We now recall the following general identity\footnote{This is obtained from \eqref{A}-\eqref{B} when as many derivatives of $\bY$ as possible are expressed in terms of Lie derivatives along the direction $n$.} \cite{tesismiguel}
\begin{align*}
	P^{cd}B_{cadb} + 2A_{(a|c|b)}n^c  &= \Rcero_{(ab)} -2\lie_n\Y_{ab} - \left(2\kappa_n+\tr_P\bU-\ntwo\big(n(\elltwo)-\tr_P\bY\big)\right)\Y_{ab} \\
	&\quad\, + \nablacero_{(a}\big(\s+2\r\big)_{b)} - 2\r_a\r_b +4\r_{(a}\s_{b)} - \s_a\s_b -\big(\tr_P\bY\big) \U_{ab} \\
	&\quad\,+ 2P^{cd}\U_{c(a}\big(2\Y+\F\big)_{b)d} + \ntwo\left(\big(\s-3\r\big)_{(a}\nablacero_{b)}\elltwo + P^{cd}\Pi_{ca}\Pi_{db}\right),
\end{align*}
and arrive at 
\begin{align}
\ntwo \Z^{(2)}_{ab}& = \Rcero_{(ab)}- {R}_{ab} -2\lie_n\Y_{ab} - \left(2\kappa_n+\tr_P\bU-\ntwo\big(n(\elltwo)-\tr_P\bY\big)\right)\Y_{ab} \nonumber\\
	&\quad\, + \nablacero_{(a}\big(\s+2\r\big)_{b)} - 2\r_a\r_b +4\r_{(a}\s_{b)} - \s_a\s_b -\big(\tr_P\bY\big) \U_{ab} + 2P^{cd}\U_{c(a}\big(2\Y+\F\big)_{b)d}\nonumber\\
	&\quad\, + \ntwo\left(-2\r_{(a}\nablacero_{b)}\elltwo + P^{cd}\big(\Pi_{ca}\Pi_{db}+\Pi_{ac}\Pi_{bd}\big)+ \dfrac{1}{4}\ntwo\nablacero_{(a}\elltwo\nablacero_{b)}\elltwo\right) . \label{Z2}
\end{align} 
This expression shows, in particular, that given hypersurface data with $\ntwo\neq 0$ the second transversal derivative $\bZ^{(2)}$ is uniquely determined in terms of the tangential components of the ambient Ricci tensor. This is of course a well-known result that dates back at least to the classic ADM equations in general relativity \cite{ADM}. Expression \eqref{Z2} is interesting because it provides second order derivatives of the metric along an arbitrary transversal direction $\xi$ in terms of hypersurface data (and the ambient Ricci tensor) in a fully gauge covariant manner.\\

 When $\ntwo=0$ we see that $R_{ab}$ can be determined using just hypersurface data. This suggests defining, for null hypersurface data $\{\mc H,\bg,\bm\ell,\elltwo,\bY\}$, the so-called \textit{constraint tensor} $\R$ by \cite{tesismiguel,miguel3}
		\begin{equation}
				\label{constraint}
				\begin{aligned}
						\mc R_{ab}& \d \accentset{\circ}{R}_{(ab)} -2\lie_n \Y_{ab} - (2\kappa_n+\tr_P\bU)\Y_{ab} + \nablacero_{(a}\left(\s_{b)}+2\r_{b)}\right)\\
						&\quad -2\r_a\r_b + 4\r_{(a}\s_{b)} - \s_a\s_b - (\tr_P\bY)\U_{ab} + 2P^{cd}\U_{d(a}\left(2\Y_{b)c}+\F_{b)c}\right).
					\end{aligned}
			\end{equation}
Therefore, in the null case, the only components of the ambient Ricci tensor that may depend on second transverse derivatives of the metric are the transverse-tangent and transverse-transverse ones. Indeed, using \eqref{inversemetric}, \eqref{ABembedded} and \eqref{riemxixiH} gives \cite{Mio3}
\begin{flalign}
R_{\alpha\beta}\xi^{\alpha}\xi^{\beta}	 &= -\tr_P\bZ^{(2)} +P^{ab}P^{cd}\Pi_{ac}\Pi_{bd}+ P^{ab}(\r-\s)_a\nablacero_b\elltwo,  &&(\ntwo=0)\label{trPZ2}\\
R_{\alpha\beta}\xi^{\alpha}e_b^{\beta}	& =  \z^{(2)}_b - P^{ac}A_{acb} - P^{ac}\big(\r+\s)_a\Pi_{bc} +\dfrac{1}{2}\kappa_n \nablacero_b\elltwo- \dfrac{1}{2} n(\elltwo)(\r-\s)_{b},&&(\ntwo=0)\label{z2}
\end{flalign}
where $\bz^{(2)}\d \bZ^{(2)}(n,\cdot)$. Conversely, this also shows that in the null case (in contrast to the non-null one) only the trace of $\bZ^{(2)}$ w.r.t $P$ and the one-form $\bz^{(2)}$ are determinable in terms of the ambient Ricci tensor evaluated on the hypersurface.

	\section{General identities for the deformation tensor}
	\label{sec_general}

Let $(\mc M,g)$ be a semi-Riemannian manifold, $\eta\in\X(\mc M)$ any vector field and $\nabla$ the Levi-Civita connection. We introduce the tensors $\mc K[\eta]$ and $\Sigma[\eta]$ by $\mc K[\eta] \d \lie_{\eta} g$ and $\Sigma[\eta]\d\lie_{\eta}\nabla$, or more specifically $$\Sigma[\eta](V,W) \d \lie_{\eta}\nabla_V W - \nabla_{\lie_{\eta}V}W- \nabla_V\lie_{\eta}W ,\qquad V,W\in\X(\mc M).$$ The former is called ``deformation tensor'' of $\eta$. The latter can be expressed in terms of $\k[\eta]$ by means of \cite{yano2020theory}
\begin{equation}
	\label{Sigmaindex}
	\Sigma[\eta]^{\mu}{}_{\nu\beta} = \dfrac{1}{2} \left(\nabla_{\nu}\mc{K}[\eta]^{\mu}{}_{\beta}+\nabla_{\beta}\mc{K}[\eta]^{\mu}{}_{\nu}-\nabla^{\mu}\mc{K}[\eta]_{\nu\beta}\right),
	 \end{equation} 
or lowering the index,
\begin{equation}
	\label{Sigmaindexdown}
\Sigma[\eta]_{\mu\nu\beta} = \dfrac{1}{2} \left(\nabla_{\nu}\mc{K}[\eta]_{\mu\beta}+\nabla_{\beta}\mc{K}[\eta]_{\mu\nu}-\nabla_{\mu}\mc{K}[\eta]_{\nu\beta}\right),
\end{equation}
which immediately implies the symmetries
\begin{multicols}{2}
	\noindent
	\begin{equation}
		\label{S2}
		2\Sigma[\eta]_{(\mu\nu)\beta} = \nabla_{\beta}\mc{K}[\eta]_{\mu\nu},
	\end{equation}
	\begin{equation}
		\label{S3}
		\Sigma[\eta]_{\mu[\nu\beta]} = 0.
	\end{equation}
\end{multicols}
The tensor $\Sigma[\eta]$ determines the commutation between covariant and Lie derivatives. Specifically, for any $(q,p)$ tensor field $A$ it holds \cite{yano2020theory}
\begin{equation}
	\label{yano}
	\begin{aligned}
\lie_{\eta}\nabla_{\beta}A^{\alpha_1\cdots\alpha_q}{}_{\mu_1\cdots\mu_p} &= \nabla_{\beta}\lie_{\eta}A^{\alpha_1\cdots\alpha_q}{}_{\mu_1\cdots\mu_p} + \sum_{j=1}^{q}A^{\alpha_1\cdots\alpha_{j-1}\sigma\alpha_{j+1}\cdots\alpha_q}{}_{\mu_1\cdots\mu_p}\Sigma[\eta]^{\alpha_j}{}_{\beta\sigma}\\&\quad\, -  \sum_{i=1}^{p}A^{\alpha_1\cdots\alpha_q}{}_{\mu_1\cdots\mu_{i-1}\sigma\mu_{i+1}\cdots\mu_p}\Sigma[\eta]^{\sigma}{}_{\beta\mu_i}.
	\end{aligned}
\end{equation}
The tensors $\mc{K}[\eta]$ and $\Sigma[\eta]$ admit the following alternative expressions \cite{yano2020theory}
\begin{multicols}{2}
	\noindent
\begin{equation}
	\label{def_K}
\mc{K}[\eta]_{\mu\nu} = 2\nabla_{(\mu}\eta_{\nu)} ,
\end{equation}
\begin{equation}
	\label{def_S}
\Sigma[\eta]_{\mu\alpha\beta} = \nabla_{\alpha}\nabla_{\beta}\eta_{\mu} + R_{\alpha\nu\beta\mu}\eta^{\nu}.
\end{equation}
\end{multicols}
A key object in this paper is the vector field $\mathcal{Q}[\eta]$ defined by 
\begin{equation}
	\label{defQ}
\mathcal{Q}[\eta]^{\mu} \d g^{\alpha\beta}\Sigma[\eta]^{\mu}{}_{\alpha\beta}=\square_g \eta^{\mu} + R^{\mu}{}_{\nu}\eta^{\nu},
\end{equation}
the second equality being a consequence of \eqref{def_S}. By \eqref{def_K} we note that $\mc{Q}[\eta]$ can also be written as
\begin{equation}
	\label{traza}
	\mathcal{Q}[\eta]_{\nu} = \nabla_{\mu}\mc{K}[\eta]^{\mu}{}_{\nu} - \dfrac{1}{2}\nabla_{\nu}\tr_g\mc{K}[\eta].
\end{equation}
 The tensors $\Sigma[\eta]$, $\mc{K}[\eta]$ and $\mc{Q}[\eta]$ are related by the following known identities \cite{yano2020theory}. We add a proof for completeness.
\begin{lema}
Let $(\mc M,g)$ be a semi-Riemannian manifold, $\eta\in\X(\mc M)$ and $\nabla$ the Levi-Civita connection. Then, the following identities hold
\begin{align}
\nabla_{\mu}\Sigma[\eta]^{\mu}{}_{\beta\nu} = \lie_{\eta} R_{\beta\nu} +\dfrac{1}{2}\nabla_{\beta}\nabla_{\nu}\tr_g\mc{K}[\eta],\label{id2}\\
\square_g \mc{K}[\eta]_{\beta\nu} - \mc{K}[\mc{Q}[\eta]]_{\beta\nu} - 2R_{\sigma(\beta}\mc{K}[\eta]^{\sigma}{}_{\nu)} + 2R_{\sigma\beta\mu\nu}\mc{K}[\eta]^{\sigma\mu} + 2\lie_{\eta}R_{\beta\nu} = 0.\label{id1}
\end{align}
\begin{proof}
We start by computing the trace and the divergence of $\Sigma[\eta]$. Using \eqref{Sigmaindex},
\begin{align}
\Sigma[\eta]^{\mu}{}_{\mu\beta} &= \dfrac{1}{2} \nabla_{\beta}\k[\eta]^{\mu}{}_{\mu}= \dfrac{1}{2}\nabla_{\beta} \tr_g\k[\eta],\label{idm0}\\
\nabla_{\mu}\Sigma[\eta]^{\mu}{}_{\beta\nu} &= -\dfrac{1}{2} \nabla_{\mu}\nabla^{\mu}\k_{\beta\nu}+\nabla_{\mu}\nabla_{(\beta}\k[\eta]^{\mu}{}_{\nu)}\nonumber\\
&= -\dfrac{1}{2} \nabla_{\mu}\nabla^{\mu}\k_{\beta\nu} + \nabla_{(\beta|}\nabla_{\mu}\k[\eta]^{\mu}{}_{|\nu)}+R^{\mu}{}_{\sigma\mu(\beta}\k[\eta]^{\sigma}{}_{\nu)}-R^{\sigma}{}_{(\nu|\mu|\beta)}\k[\eta]^{\mu}{}_{\sigma}.\label{idm1}
\end{align}	
In order to prove \eqref{id2} we use the well-known identity \cite{yano2020theory}
\begin{equation}
	\label{Yano}
	\lie_{\eta} R_{\beta\nu} = \nabla_{\mu}\Sigma[\eta]^{\mu}{}_{\beta\nu} - \nabla_{\nu}\Sigma[\eta]^{\mu}{}_{\mu\beta},
\end{equation}
which combined with \eqref{idm0} gives \eqref{id2}. To prove \eqref{id1} we insert \eqref{traza} into \eqref{idm1}, 
\begin{equation*}
\nabla_{\mu}\Sigma[\eta]^{\mu}{}_{\beta\nu} = -\dfrac{1}{2}\nabla_{\mu}\nabla^{\mu}\k[\eta]_{\beta\nu} + \nabla_{(\beta}\mc{Q}[\eta]_{\nu)} + \dfrac{1}{2}\nabla_{\beta}\nabla_{\nu}\tr_g\k + R_{\sigma(\beta}\k[\eta]^{\sigma}{}_{\nu)} - R_{\sigma\beta\mu\nu}\k[\eta]^{\sigma\mu}.
\end{equation*}
Using \eqref{id2} and taking into account $\k[\mc{Q}[\eta]]_{\beta\nu} = 2\nabla_{(\beta}\mc{Q}[\eta]_{\nu)}$, \eqref{id1} follows at once.
\end{proof}
\end{lema}
For any $\mu\in\real$, identity \eqref{id1} can be rewritten as
\begin{equation*}
	\square_g \big(\mc{K}[\eta]_{\beta\nu}-\mu g_{\beta\nu}\big) - \mc{K}[\mc{Q}[\eta]]_{\beta\nu} - 2R_{\sigma(\beta}\big(\mc{K}[\eta]^{\sigma}{}_{\nu)} - \mu\delta^{\sigma}_{\nu)}\big) + 2R_{\sigma\beta\mu\nu}\big(\mc{K}[\eta]^{\sigma\mu}-\mu g^{\sigma\mu}\big) + 2\lie_{\eta}R_{\beta\nu} = 0.
\end{equation*}
Whenever $R_{\mu\nu}=\lambda g_{\mu\nu}$ this identity becomes
\begin{equation}
	\label{equationk}
	\begin{aligned}
	\square_g \big(\mc{K}[\eta]_{\beta\nu}-\mu g_{\beta\nu}\big) - \mc{K}[\mc{Q}[\eta]]_{\beta\nu} - 2R_{\sigma(\beta}\big(\mc{K}[\eta]^{\sigma}{}_{\nu)} - \mu\delta^{\sigma}_{\nu)}\big)&\\
	 + 2R_{\sigma\beta\mu\nu}\big(\mc{K}[\eta]^{\sigma\mu}-\mu g^{\sigma\mu}\big) + 2\lambda\big(\mc{K}_{\beta\nu}-\mu g_{\beta\nu}\big) +2\lambda\mu g_{\beta\nu} &= 0.
	\end{aligned}
\end{equation}
When the vector $\eta$ satisfies $\mc{Q}[\eta]=0$ this becomes a homogeneous PDE for $\mc{K}[\eta]_{\beta\nu}-\mu g_{\beta\nu}$ provided $\mu\lambda=0$. This gives a propagation equation that can be exploited to ensure that a vector field is a (proper) homothety in a Ricci flat manifold, or a Killing vector in a $\lambda$-vacuum manifold. Note that a homothety $\eta$ satisfying $\lie_{\eta} g = \mu g$ with $\mu \neq 0$ cannot exist in an Einstein space with $\lambda \neq 0$, because homotheties satisfy $\Sigma[\eta] =0$ and hence $0 =\lie_{\eta}\ric = \lambda\mu g$.\\

To make the computations below as tractable as possible we shall make use of the following property of the tensor $\Sigma[\eta]$.
\begin{lema}
	\label{lema_Marc}
	Let $(\mc M,g)$ be a semi-Riemannian manifold, $\eta,\zeta\in\X(\mc M)$ arbitrary vector fields, $\k[\eta]\d\lie_{\eta}g$, $\Sigma[\eta]\d\lie_{\eta}\nabla$ and define $\bm\zeta \d g(\zeta,\cdot)$. Then, the following two identities hold 
	\begin{align}
	\zeta^{\alpha}\Sigma[\eta]_{\alpha\beta\mu} &= \nabla_{(\beta}(\lie_{\eta}\bm\zeta)_{\mu)} - \dfrac{1}{2}\lie_{\eta}\lie_{\zeta}g_{\beta\mu},\label{zetaSigma1}\\
\zeta^{\alpha}\Sigma[\eta]_{\alpha\beta\mu}	& = \dfrac{1}{2}\lie_{k^{(\zeta)}} g_{\beta\mu}-\dfrac{1}{2}\lie_{\zeta}\k[\eta]_{\beta\mu},\label{zetaSigma2}
	\end{align} 
where $(k^{(\zeta)})^{\mu} \d \k[\eta]^{\mu}{}_{\beta}\zeta^{\beta}$. 
	\begin{proof}
From \eqref{yano} applied to $A=\bm\zeta$ it follows $$\lie_{\eta}\nabla_{\beta}\bm\zeta_{\mu} = \nabla_{\beta}(\lie_{\eta}\bm\zeta)_{\mu} - \Sigma[\eta]^{\alpha}{}_{\beta\mu}\zeta_{\alpha}\qquad\Longrightarrow\qquad \big(\lie_{\eta}\lie_{\zeta}g\big)_{\beta\mu} = 2\nabla_{(\beta}(\lie_{\eta}\bm\zeta)_{\mu)}-2\Sigma[\eta]^{\alpha}{}_{\beta\mu}\zeta_{\alpha},$$ which proves \eqref{zetaSigma1}. The second identity follows from the first after commuting $\lie_{\eta}\lie_{\zeta}=\lie_{\zeta}\lie_{\eta} + \lie_{[\eta,\zeta]}$ and using $\nabla_{(\beta}\lie_{\eta}\big(g_{\mu)\alpha}\zeta^{\alpha}\big) = \nabla_{(\beta}\big(\k[\eta]_{\mu)\alpha}\zeta^{\alpha}\big) + \nabla_{(\beta}([\eta,\zeta])_{\mu)}=\dfrac{1}{2}\lie_{k^{(\zeta)}}g_{\beta\mu}+\dfrac{1}{2}\lie_{[\eta,\zeta]}g_{\beta\mu}$.
	\end{proof}
\end{lema}

\section{Identities for $\k[\eta]$ and $\Sigma[\eta]$ on abstract hypersurfaces}

\label{sec_hip}

In this section we analyse the consequences of the identities in Section \ref{sec_general} on embedded hypersurface data $\{\mc H,\bg,\bm\ell,\elltwo,\bY\}$. The computations rely on general expressions for the pullback of ambient tensor fields into an arbitrary hypersurface. They are obtained in Appendix \ref{app}. In particular we refer to this appendix for the notation ${}^{(i)}T$, ${}^{(i,j)}T$, ${}^{(i)}\nabla T$ and so on. For simplicity from now on we drop the label ``$[\eta]$'' in the tensors $\k$, $\Sigma$ and $\mc{Q}$. The structure of this section is as follows. First we compute the fully tangential components of $\k$ in terms of hypersurface data and $\eta|_{\mc H}$. Then, we relate the transverse-tangent and transverse-transverse components of $\k|_{\mc H}$ with the first transverse derivative of $\eta$ at $\mc H$. We then move on to the tensor $\Sigma$ and write down its fully tangential components as well as the pullbacks of $\Sigma$ contracted with one or several transverse vectors $\xi$. We conclude the section by expressing the tangential and transversal components of $\mc{Q}[\eta]$ in terms of objects defined on the hypersurface.\\

Given $\eta\in\X(\mc M)$ we define the vector $\bar\eta\in\X(\mc H)$ and the scalar $C\in\mc{F}(\mc H)$ by means of $\eta|_{\Phi(\mc H)} = C\xi + \Phi_{\star}\bar\eta$. Using $\bm\eta_a$ to denote the pullback on $\mc H$ of $\bm\eta\d g(\eta,\cdot)$ one easily has
  \begin{equation}
  	\label{etatangxi}
  	\bm\eta_a =C\ell_a+ \gamma_{ab}\bar\eta^b ,\qquad \bm\eta(\xi) = C\elltwo+\bm\ell(\bar\eta).
  \end{equation}  

  In the next proposition we compute the fully tangent components of the tensor $\mc{K}$ and we show that they can be written solely in terms of hypersurface data and the pair $(C,\bar\eta)$.
\begin{prop}
Let $\{\mc H,\bg,\bm\ell,\elltwo,\bY\}$ be hypersurface data $(\Phi,\xi)$-embedded in $(\mc M,g)$, $\eta\in\X(\mc M)$ and $\mc{K}\d\lie_{\eta} g$. Define $\eta\st{\Phi(\mc H)}{=}C\xi+\Phi_{\star}\bar\eta$. Then,
\begin{equation}
	\label{Adown}
	\mc{K}_{ab} = 2C\Y_{ab}  + 2\ell_{(a}\nablacero_{b)}C +\lie_{\bar\eta}\gamma_{ab}.
\end{equation}
\begin{proof}
Using Proposition \ref{propliezetaT} with $T=g$ and $\zeta=\eta$ and taking into account $\k = \lie_{\eta}g$, $\Phi^{\star}\lie_{\xi}g = 2\bY$ and $\Phi^{\star}g=\bg$ (see Def. \ref{defi_embedded}) the result follows at once.
\end{proof}
\end{prop}

\begin{rmk}
When $C=1$ and $\bar\eta=0$ (i.e. when $\eta|_{\mc{H}}=\xi$) one gets $\mc{K}_{ab} = 2\Y_{ab}$, as required from the definition of embedded data regarding the tensor $\bY$. Similarly, when $C=\ntwo$ and $\bar\eta = n$ (and thus $\eta|_{\mc{H}}=\nu$) we obtain $\mc{K}_{ab}=2\ntwo\Y_{ab}+2\U_{ab}=2\K_{ab}$, as required from the geometric interpretation of $\bK$ as the second fundamental form of $\Phi(\mc H)$ along the normal $\nu$.
\end{rmk}

Unlike $\k$, the fully tangential components of $\Sigma$ cannot be written solely in terms of hypersurface data and $(C,\bar\eta)$. Indeed, after contracting \eqref{Sigmaindexdown} (or \eqref{def_S}) with three tangent vectors and using Proposition \ref{proppullback} one cannot avoid the appearance of the transverse components of the deformation tensor $\k$ that we introduce next.

\begin{nota}
	Given $\mc{K}\d \lie_{\eta} g$ and an embedded hypersurface $\Phi:\mc H\hookrightarrow\mc M$ with rigging $\xi$, we introduce the one-form $\kt\d\Phi^{\star}\big(\mc{K}(\xi,\cdot)\big)$ and the scalar $\ktt\d \Phi^{\star}\big(\mc{K}(\xi,\xi)\big)$.
\end{nota}

Up to terms involving hypersurface data and $(C,\bar\eta)$, the tensors $\kt_a$, $\ktt$ are in one-to-one corres\-pondence with the transverse derivative of the one-form $\bm\eta$. Indeed,
\begin{align}
\kt_a = 2e_a^{\alpha}\xi^{\beta}\nabla_{(\alpha}\bm\eta_{\beta)} & = e_a^{\alpha}\xi^{\beta}\nabla_{\alpha}\bm\eta_{\beta} + e_a^{\alpha}\xi^{\beta}\nabla_{\beta}\bm\eta_{\alpha}\nonumber\\
&=\nablacero_a\big(\bm\eta(\xi)\big) - 2 e_a^{\alpha} \bm\eta_{\beta}\nabla_{\alpha}\xi^{\beta} + e_a^{\alpha}\big(\lie_{\xi}\bm\eta\big)_{\alpha}\nonumber\\
&\hspace{-0.1cm}\st{\eqref{nablaxi}}{=}\nablacero_a \big(\bm\eta(\xi)\big) - 2\left(\r_a-\s_a+\dfrac{1}{2}\ntwo\nablacero_a\elltwo\right)\bm\eta(\xi) - 2V^b{}_a \eta_b + \big(\lie_{\xi}\bm\eta\big)_a\nonumber\\
&= \elltwo\nablacero_a C + \nablacero_a\big(\bm\ell(\bar\eta)\big) +  \big(\lie_{\xi}\bm\eta\big)_a - 2\big(\Y_{ac}+\F_{ac}\big)\bar\eta^c,\label{liet}
\end{align}
where in the last line we inserted \eqref{etatangxi} and simplified after using \eqref{ellV}-\eqref{gammaV}. Similarly, 
\begin{align}
\ktt =2\xi^{\alpha}\xi^{\beta}\nabla_{\alpha}\bm\eta_{\beta} &= 2\xi^{\beta}\left(\lie_{\xi}\bm\eta_{\beta}-\eta_{\alpha}\nabla_{\beta}\xi^{\alpha}\right)=2\big(\lie_{\xi}\bm\eta\big)(\xi) -2 g(\eta,\nabla_{\xi}\xi).\label{lietr}
\end{align}

Next we write down the pullback of the tensor $\Sigma$ in terms of hypersurface data and the values of $\k$ on the hypersurface.

\begin{prop}
Let $\{\mc H,\bg,\bm\ell,\elltwo,\bY\}$ be hypersurface data $(\Phi,\xi)$-embedded in $(\mc M,g)$, $\eta\in\X(\mc M)$, $\mc{K}\d \lie_{\eta}g$ and $\Sigma\d \lie_{\eta}\nabla$. Then,
\begin{equation}
	\label{Sigmattt}
\Sigma_{abc} = \nablacero_{(b}\mc{K}_{c)a} - \dfrac{1}{2}\nablacero_a\mc{K}_{bc} +\Y_{bc}\mc{K}_{ad}n^d + \big(\U_{bc}+\ntwo\Y_{bc}\big)\mf{K}_a.
\end{equation}
\begin{proof}
Contracting \eqref{Sigmaindexdown} with three tangent vectors, applying \eqref{identity1} in Proposition \ref{proppullback} to $T=\k$ and replacing \eqref{defK}, \eqref{Sigmattt} follows at once. 
\end{proof}
\end{prop}

\begin{rmk}
	\label{rmk1}
We observe that if $\k=\mu\bg$ with $\mu\in\real$ then $\Sigma_{abc}$ vanishes identically at the points where $\kt=\mu\bm\ell$. To show this we insert \eqref{nablaell} and \eqref{gamman} into \eqref{Sigmattt} and get
	\begin{align*}
\Sigma_{abc}&= -\mu \ell_{a}\U_{bc}-\mu\ntwo\ell_a \Y_{bc} +\mu\ell_a \big(\U_{bc}+\ntwo\Y_{bc}\big)=0.
	\end{align*}
\end{rmk}

Another key object in this work is the pullback tensor $\Phi^{\star}\big(\Sigma(\xi,\cdot,\cdot)\big)\eqqcolon\mathbb{\Sigma}$, i.e. the one-transverse--two-tangent components of $\Sigma$. This quantity is important because it arises in several other expressions involving the tensor $\nabla\k$ contracted with one transversal direction and two tangential ones. Let us first find these relationships. We start with $\Phi^{\star}\big({}^{(2)}\nabla\k\big)$, for which we use \eqref{identity3} in Appendix \ref{app} applied to $T=\k$ and $j=1$ to get
\begin{equation}
	\label{2K}
\big({}^{(2)}\nabla\k\big)_{ab} = \nablacero_a\kt_b + \kt(n)\Y_{ab} + \mathbb{k}\big(\U_{ab}+\ntwo\Y_{ab}\big) - \left(\r_{a}-\s_{a}+\dfrac{1}{2}\ntwo\nablacero_{a}\elltwo\right)\kt_b - V^c{}_a \k_{bc}.
\end{equation}
Concerning $\Phi^{\star}\big({}^{(1)}\nabla\k\big)$ we use \eqref{Sigmaindexdown}, namely $\mS_{ab}=	\xi^{\mu} e^{\nu}_a e^{\beta}_b \Sigma_{\mu\nu\beta} =  \big({}^{(2)}\nabla\k\big)_{(ab)} - \dfrac{1}{2}\big({}^{(1)}\nabla\k\big)_{ab}$, which combined with \eqref{2K} gives
\begin{equation}
	\label{1K}
	\begin{aligned}
\hspace{-0.2cm}\big({}^{(1)}\nabla\k\big)_{ab}  &= 2\nablacero_{(a}\kt_{b)} + 2\kt(n)\Y_{ab} + 2\mathbb{k}\big(\U_{ab}+\ntwo\Y_{ab}\big) - 2\left(\r_{(a}-\s_{(a}+\dfrac{1}{2}\ntwo\nablacero_{(a}\elltwo\right)\kt_{b)} \\
&\quad\, - 2V^c{}_{(a} \k_{b)c} - 2\mS_{ab}.
\end{aligned}
\end{equation}
For later use we shall need the contraction of $\big({}^{(2)}\nabla\k\big)_{ab}$ and $\big({}^{(1)}\nabla\k\big)_{ab}$ with $n^b$ in the null case. The former follows after using \eqref{contra1} in Lemma \ref{lemmannabla} with $\bm\theta=\kt$ and $\U_{ab}n^b=0$,
\begin{flalign}
	\label{2Kn}
&&\big({}^{(2)}\nabla\k\big)_{ab}n^b = \nablacero_a\big(\kt(n)\big)  -P^{bc}\U_{ac}\kt_b-V^c{}_a\k_{bc}n^b, &&(\ntwo=0)
\end{flalign}
while the latter follows from \eqref{nnablacerotheta} also in Lemma \ref{lemmannabla} with $\bm\theta=\kt$,
\begin{flalign}
\big({}^{(1)}\nabla\k\big)_{ab}n^b &=	\lie_n\kt_a + \nablacero_a\big(\kt(n)\big) + \kt(n)\big(\r_a-\s_a\big)-2P^{bc}\U_{ac}\kt_b + \kappa_n\kt_a - V^c{}_bn^b\k_{ac} \nonumber\\
		&\quad\,  - V^c{}_a \k_{bc}n^b- 2\mS_{ab}n^b. && \hspace{-1cm}(\ntwo=0)\label{1Kn}
\end{flalign}
It is also useful to write down the pullback of the Lie derivative of $\k$ along $\xi$. Lemma \ref{lema_Marc} [Eq. \eqref{zetaSigma2}] with $\zeta=\xi$ gives $\lie_{\xi}\k_{\alpha\beta} = 2\nabla_{(\alpha}k^{(\xi)}_{\beta)} - \xi^{\mu}\Sigma_{\mu\alpha\beta}$. Pulling this back to $\mc H$ and using \eqref{identity1} in Proposition \ref{proppullback} applied to $T_{\mu} = \k_{\mu\nu}\xi^{\nu}$ we arrive at
\begin{equation}
	\label{aux}
	\big(\lie_{\xi}\mc{K}\big)_{ab} =	2\nablacero_{(a}\mf{K}_{b)} + 2\kt(n)\Y_{ab} +2\mathbb{k}\big(\U_{ab}+\ntwo\Y_{ab}\big)-2\mS_{ab}.
\end{equation}


\begin{rmk}
	\label{rmkmS}
	We note that if $\kt=\mu\bm\ell$, $\ktt=\mu\elltwo$ and $\big(\lie_{\xi}\mc{K}\big)_{ab}=2\mu\Y_{ab}$, then $\mS_{ab}$ vanishes identically. The computation relies on \eqref{aux} and uses \eqref{ell(n)} and \eqref{nablaell} as follows
	\begin{align*}
2\mS_{ab} & = 2\mu\nablacero_{(a}\ell_{b)} +2\mu\bm\ell(n)\Y_{ab} +2\mu\elltwo\big(\U_{ab}+\ntwo\Y_{ab}\big) -2\mu\Y_{ab}\\
&= - 2\mu\elltwo\U_{ab} +2\mu\big(1-\ntwo\elltwo\big)\Y_{ab}+2\mu\elltwo \big(\U_{ab}+\ntwo\Y_{ab}\big)-2\mu\Y_{ab}=0.
	\end{align*}
\end{rmk}
Note that \eqref{2K} together with \eqref{S2}-\eqref{S3} provide the contractions of $\Sigma_{\alpha\beta\mu}$ with one rigging and two tangential directions (in any order) in terms of hypersurface data and $\mS_{ab}$. Specifically, one has 
\begin{align}
\hspace{-0.5cm} {}^{(2)}\Sigma_{ab} = {}^{(3)}\Sigma_{ab} &= -\mathbb{\Sigma}_{ab} + \big({}^{(2)}\nabla\k\big)_{ba}\nonumber\\
&=-\mathbb{\Sigma}_{ab} + \nablacero_b\mf{K}_a +\mf{K}(n)\Y_{ba} + \mathbb{k}\big(\U_{ba}+\ntwo\Y_{ba}\big) \nonumber\\
&\quad\, - \left(\r_b-\s_b + \dfrac{1}{2}\ntwo \nablacero_b\elltwo\right)\mf{K}_a - V^c{}_b \mc{K}_{ac}.\label{symmetries}
\end{align}
In particular, after using \eqref{Un}, \eqref{Vn} as well as \eqref{contra1}-\eqref{contra2} in Lemma \ref{lemmannabla} for $\theta_c = \kt_c$ 
we find that the contractions of ${}^{(2)}\Sigma_{ab}$ with $n^a$ and $n^b$ are (note that ${}^{(2)}\Sigma_{ab}$ is not symmetric)
\begin{equation}
	\label{symmetriesn0}
	\begin{aligned}
		{}^{(2)}\Sigma_{ab}n^a &= -\mS_{ab}n^a + \nablacero_b\big(\kt(n)\big) - P^{ac}\U_{bc}\kt_a - V^c{}_b \k_{ac}n^a + \dfrac{1}{2}\ktt\nablacero_b\ntwo\\
		&\quad\, + \ntwo\left(\dfrac{1}{2}\big(\kt(n)+\ktt\big)\nablacero_b\elltwo + P^{ac}\F_{bc}\kt_a + \ktt(\r_b-\s_b)\right) 
	\end{aligned}
\end{equation}
and
\begin{equation}
	\label{symmetriesn}
	\begin{aligned}
{}^{(2)}\Sigma_{ab}n^b &= -\mS_{ab}n^b +\lie_n\kt_a - P^{bc}\U_{ac}\kt_b + \kt(n)\big(\r_a-\s_a\big)+\left(\kappa_n - \dfrac{1}{2}\ntwo n(\elltwo)\right)\kt_a  \\
&\quad\,- \left(P^{cb}\big(\r_b+\s_b\big)+\dfrac{1}{2}n(\elltwo)n^c\right) \k_{ac}+\dfrac{1}{2}\ktt\nablacero_a\ntwo  \\
&\quad\, +\ntwo\left(\kt(n)\nablacero_a\elltwo + P^{bc}\F_{ac}\kt_b+\ktt\left(\r_a-\s_a+\dfrac{1}{2}\ntwo\nablacero_a\elltwo\right)\right).
	\end{aligned}
\end{equation}
We now take on the task of obtaining the tensor $\mS_{ab}$ in terms of extended hypersurface data and the pair $(C,\bar\eta)$.
\begin{prop}
	\label{propmS}
	Let $\{\mc H,\bg,\bm\ell,\elltwo,\bY,\bZ^{(2)}\}$ be extended hypersurface data $(\Phi,\xi)$-embedded on $(\mc M,g)$ and $\eta\in\X(\mc M)$. Decomposing $\eta\st{\Phi(\mc H)}{=} C\xi + \Phi_{\star}\bar\eta$, then
\begin{equation}
	\label{Sigmaxitt}
	\begin{aligned}
		\mathbb{\Sigma}_{ab} &=  -C\Z^{(2)}_{ab} -\lie_{\bar\eta}\Y_{ab} + \left(\dfrac{1}{2}Cn(\elltwo)+ \elltwo n(C)+ (\lie_{\bar\eta}\bm\ell)(n)\right) \Y_{ab}  +\nablacero_{(a}\elltwo \nablacero_{b)}C \\
		&\quad\,+ \dfrac{1}{2}\big(\bar\eta(\elltwo)+\mathbb{k}\big)\big(\U_{ab}+\ntwo\Y_{ab}\big)   +\dfrac{1}{2}C\nablacero_a\nablacero_b\elltwo   + \elltwo\nablacero_a\nablacero_b C   + \nablacero_{(a}\lie_{\bar\eta}\ell_{b)}.
	\end{aligned}
\end{equation}	
	\begin{proof}
From Lemma \ref{lema_Marc} [Eq. \eqref{zetaSigma1}] with $\zeta=\xi$ we can write
\begin{equation*}
\mS_{ab}=e_a^{\alpha} e_b^{\beta}	\xi^{\mu}\Sigma_{\mu\alpha\beta} = e_a^{\alpha} e_b^{\beta}\nabla_{(\alpha}(\lie_{\eta}\bm\xi)_{\beta)}-\dfrac{1}{2}e_a^{\alpha} e_b^{\beta}\big(\lie_{\eta}\lie_{\xi}g\big)_{\alpha\beta},
\end{equation*} 
which after using Proposition \ref{proppullback} [Eq. \eqref{identity1}] applied to $T=\bm\xi$ in the first term gives 
\begin{equation}
	\label{inter}
	\mS_{ab} = \nablacero_{(a}(\lie_{\eta}\bm\xi)_{b)} +   \Y_{ab}n^c(\lie_{\eta}\bm\xi)_c + \K_{ab} \xi^{\mu}(\lie_{\eta}\bm\xi)_{\mu}-\dfrac{1}{2}e_a^{\alpha} e_b^{\beta}\big(\lie_{\eta}\lie_{\xi}g\big)_{\alpha\beta}.
\end{equation}
Hence to compute the first three terms it suffices to calculate $\Phi^{\star}\lie_{\eta}\bm\xi$ and $\big(\lie_{\eta}\bm\xi\big)(\xi)$. Since the the left-hand side of \eqref{inter} does not depend on the extension of $\xi$ off $\Phi(\mc H)$ we may assume without loss of generality $\nabla_{\xi}\xi=0$ to simplify the computation. Applying \eqref{liezetaT} in Proposition \ref{propliezetaT} with $\zeta=\eta$ and $T=\bm\xi$ together with
\begin{flalign}
	\label{recall}
&&	\Phi^{\star}\lie_{\xi}\bm\xi = \Phi^{\star}\left(\nabla_{\xi}\bm\xi + \dfrac{1}{2} d\big(g(\xi,\xi)\big)\right)=  \dfrac{1}{2}d\elltwo && (\nabla_{\xi}\xi=0)
\end{flalign}
gives 
\begin{equation}
	\label{inter2}
\big(\lie_{\eta}\bm\xi\big)_b = \dfrac{1}{2}C \nablacero_b\elltwo+ \elltwo \nablacero_bC + \big(\lie_{\bar\eta}\bm\ell\big)_b.
\end{equation}
In order to compute $\big(\lie_{\eta}\bm\xi\big)(\xi)$ note that 
\begin{equation}
	\label{lieetaxixi}
	\big(\lie_{\eta}\bm\xi\big)(\xi) =\big( \eta^{\mu}\nabla_{\mu}\xi_{\alpha}+\xi_{\mu}\nabla_{\alpha}\eta^{\mu}\big)\xi^{\alpha} =\dfrac{1}{2}\big(\eta^{\mu}\nabla_{\mu}\big(\xi_{\alpha}\xi^{\alpha}\big) + \xi^{\mu}\xi^{\alpha}\big(\lie_{\eta}g\big)_{\mu\alpha} \big) \st{\mc H}{=}\dfrac{1}{2} \big( \bar\eta(\elltwo)+\ktt\big),
\end{equation}
where in the last equality we used
\begin{flalign}
	\label{recall2}
&& \xi^{\mu} \nabla_{\mu} (\xi_{\alpha} \xi^{\alpha}) = 	\lie_{\xi}\big(g(\xi,\xi)\big)=2g\big(\nabla_{\xi}\xi,\xi\big)=0. && (\nabla_{\xi}\xi=0).
\end{flalign} 
For the last term in \eqref{inter} we apply again Proposition \ref{propliezetaT} [Eq. \eqref{liezetaT}] with $\zeta=\eta$ and $T=\lie_{\xi}g$. Taking into account that now $\bZ^{(2)} = \dfrac{1}{2}\Phi^{\star}\big(\lie^{(2)}_{\xi}g\big)$ (cf. \eqref{Z2embedded}) together with ${}^{(1)}T = \lie_{\xi}\bm\xi$ and \eqref{recall} it follows 
\begin{equation}
	\label{lieetaliexig}
	e_a^{\alpha} e_b^{\beta}\big(\lie_{\eta}\lie_{\xi}g\big)_{\alpha\beta} = 2C\Z^{(2)}_{ab} +\nablacero_{(a}\elltwo \nablacero_{b)}C+2\lie_{\bar\eta}\Y_{ab}.
\end{equation} 
Inserting \eqref{inter2}, \eqref{lieetaxixi} and \eqref{lieetaliexig} into \eqref{inter} and using \eqref{defK} the result is established.
	\end{proof}
\end{prop}

Having understood the contraction of $\nabla\k$ and $\Sigma$ with one transversal and two tangential directions, we proceed with the computation of the contraction of $\Sigma$ with more than one $\xi$. From \eqref{Sigmaindexdown},
\begin{align}
\Sigma_{\mu\alpha\beta}\xi^{\alpha}\xi^{\beta} &= \xi^{\alpha}\xi^{\beta} \nabla_{\alpha}\k_{\mu\beta} - \dfrac{1}{2}\xi^{\alpha}\xi^{\beta}\nabla_{\mu}\k_{\alpha\beta}\nonumber\\
&= \xi^{\beta} \big(\lie_{\xi}\k_{\mu\beta} - \k_{\mu\rho}\nabla_{\beta}\xi^{\rho} - \k_{\rho\beta}\nabla_{\mu}\xi^{\rho}\big) - \dfrac{1}{2}\nabla_{\mu}\big(\k(\xi,\xi)\big) + \k_{\alpha\beta}\xi^{\alpha}\nabla_{\mu}\xi^{\beta}\nonumber\\
&= \xi^{\beta}\lie_{\xi}\k_{\mu\beta} - \k_{\mu\rho}\nabla_{\xi}\xi^{\rho} - \dfrac{1}{2}\nabla_{\mu}\big(\k(\xi,\xi)\big),\label{23Sigma}
\end{align}
and as consequence,
\begin{equation}
	\Sigma_{\mu\alpha\beta}\xi^{\mu}\xi^{\alpha}\xi^{\beta} = \dfrac{1}{2}\xi^{\mu}\xi^{\beta}\lie_{\xi}\k_{\mu\beta} - \k_{\mu\rho}\xi^{\mu}\xi^{\beta}\nabla_{\beta}\xi^{\rho}.\label{123Sigma}
\end{equation}
Similarly from \eqref{Sigmaindexdown},
\begin{align}
\Sigma_{\mu\alpha\beta}\xi^{\mu}\xi^{\alpha}  =\Sigma_{\mu\beta\alpha}\xi^{\mu}\xi^{\alpha}= \dfrac{1}{2}\xi^{\mu}\xi^{\alpha} \nabla_{\beta}\k_{\mu\alpha}= \dfrac{1}{2}\nabla_{\beta}\big(\k(\xi,\xi)\big)  - \k_{\rho\mu}\xi^{\mu}\nabla_{\beta}\xi^{\rho} .\label{12Sigma}
\end{align}
Defining $\mf{K}^{(2)} \d \Phi^{\star}\big(\lie_{\xi}\mc{K}(\xi,\cdot)\big)$, $\ktt^{(2)}\d \big(\lie_{\xi}\k\big)(\xi,\xi)$ and $\beta,a_{\para}$ by $\nabla_{\xi}\xi\st{\mc H}{=} \beta\xi + \Phi_{\star}a_{\para}$ it follows that the pullbacks of \eqref{23Sigma}-\eqref{12Sigma} are
\begin{align}
\big({}^{(2,3)}\Sigma\big)_c &= \kt^{(2)}_c - \beta\kt_c - a_{\para}^b \k_{bc} - \dfrac{1}{2} \nablacero_c\ktt,\label{23p}\\
{}^{(1,2,3)}\Sigma  &= \dfrac{1}{2}\ktt^{(2)}-\beta\ktt  - a_{\para}^b\kt_b,\label{123p}\\
\big({}^{(1,2)}\Sigma\big)_c = \big({}^{(1,3)}\Sigma\big)_c &= \dfrac{1}{2}\nablacero_c\ktt -  \left(\r_c-\s_c+\dfrac{1}{2}\ntwo\nablacero_c\elltwo\right)\ktt  - \kt_a V^a{}_c,\label{12p}
\end{align}
where in the last one we used \eqref{nablaxi}. We shall also need the contraction of the first and third ones with $n^c$. The first is immediate, and for the second one we use \eqref{Vn}. The result is
\begin{align}
	\hspace{-0.6cm}\big({}^{(2,3)}\Sigma\big)_cn^c &= \kt^{(2)}(n) - \beta\kt(n)-a_{\para}^b\k_{bc}n^c - \dfrac{1}{2}n(\ktt),\label{23n}\\
	\hspace{-0.6cm}\big({}^{(1,2)}\Sigma\big)_cn^c = \big({}^{(1,3)}\Sigma\big)_cn^c &= \dfrac{1}{2}n(\ktt)+\left(\kappa_n-\dfrac{1}{2}\ntwo n(\elltwo)\right)\ktt - P^{ab}\big(\r+\s\big)_a\kt_b -\dfrac{1}{2}n(\elltwo)\kt(n).\label{12n}
\end{align}
Observe that, from $\nu=\Phi_{\star}n+\ntwo\xi$, these expressions imply 
\begin{align}
\hspace{-0.3cm}\Sigma(\nu,\xi,\xi) & \st{\mc H}{=} \kt^{(2)}(n) - \beta\kt(n)-a_{\para}^b\k_{bc}n^c - \dfrac{1}{2}n(\ktt) + \ntwo\left(\dfrac{1}{2}\ktt^{(2)}-\beta\ktt  - a_{\para}^b\kt_b\right),\label{23null}\\
\hspace{-0.3cm} \Sigma(\xi,\nu,\xi) = \Sigma(\xi,\xi,\nu) &\st{\mc H}{=}´\dfrac{1}{2}n(\ktt)+\left(\kappa_n-\dfrac{1}{2}\ntwo n(\elltwo)\right)\ktt - P^{ab}\big(\r+\s\big)_a\kt_b -\dfrac{1}{2}n(\elltwo)\kt(n)\nonumber\\
&\quad\, + \ntwo\left(\dfrac{1}{2}\ktt^{(2)}-\beta\ktt  - a_{\para}^b\kt_b\right).\label{12null}
\end{align}


We conclude this section by computing the tangential and transverse components of $\mc{Q}[\eta]$ on the hypersurface. The result will play a crucial role both in the non-null and in the null cases considered below.


\begin{prop}
	\label{lemaQs}
Let $\{\mc H,\bg,\bm\ell,\elltwo,\bY\}$ be hypersurface data $(\Phi,\xi)$-embedded in $(\mc M,g)$, $\eta\in\X(\mc M)$, $\k\d\lie_{\eta}g$, $\Sigma\d\lie_{\eta}\nabla$ and let $\mc{Q}_c \d e_c^{\mu}\mc{Q}_{\mu}$ and $\mf{q}\d \xi^{\mu}\mc{Q}_{\mu}$. Then,
\begin{align}
\mc{Q}_c & = - 2\mS_{ca}n^a + 2\lie_n\kt_c -2P^{ab}\U_{bc}\kt_a + \big(2\kappa_n + \tr_P\bU\big)\kt_c + 2\kt(n)\big(\r_c-\s_c + \ntwo\nablacero_c\elltwo\big)\nonumber \\
&\quad\,  + P^{ab}\nablacero_a\k_{bc} - \dfrac{1}{2}P^{ab}\nablacero_c \k_{ab} + \big(\tr_P\bY-n(\elltwo)\big)\k_{cd}n^d  -2P^{db}(\r_b+\s_b)\k_{dc} +\ktt \nablacero_c \ntwo  \nonumber \\
&\quad\, +\ntwo \Bigg(\kt^{(2)}_c - \beta\kt_c -a_{\para}^b\k_{bc} + \big(\tr_P\bY-n(\elltwo)\big)\kt_c - \dfrac{1}{2}\nablacero_c\ktt+2P^{ab}\F_{cb}\kt_a \nonumber\\
&\qquad\qquad + 2\left(\r_c-\s_c +\dfrac{1}{2} \ntwo \nablacero_c\elltwo\right)\ktt\Bigg)\label{identityQ}
\end{align}
and
\begin{equation}
	\label{q}
	\begin{aligned}
\mf{q}  &=  \tr_P\mathbb{\Sigma} + n\big(\mathbb{k}\big)+2\kappa_n\ktt -2P^{ab}(\r+\s)_a\mf{K}_b -n(\elltwo) \kt(n)  \\
&\quad\, + \ntwo \left(\dfrac{1}{2}\mathbb{k}^{(2)}-\big(n(\elltwo)+\beta\big)\mathbb{k}-a_{\para}^b\mf{K}_b\right) ,
	\end{aligned}
\end{equation}
where $\k$, $\kt$, $\ktt$, $\mS$, $\beta$, $a_{\para}$, $\mf{K}^{(2)}$ and $\ktt^{(2)}$ are as before.
\begin{proof}
	In order to prove the first identity we contract the definition $\mc{Q}_{\mu} = g^{\alpha\beta}\Sigma_{\mu\alpha\beta}$ with $e^{\mu}_c$ and insert \eqref{inversemetric} to get (recall that $\Sigma$ is symmetric in the last two indices)
	\begin{align*}
\mc{Q}_c &= \left(P^{ab}e_a^{\alpha}e_b^{\beta}  +2n^b e_b^{\beta} \xi^{\alpha} + \ntwo \xi^{\alpha}\xi^{\beta}\right) e_c^{\mu} \Sigma_{\mu\alpha\beta}=P^{ab} \Sigma_{cab} + 2\big({}^{(2)}\Sigma\big)_{cb}n^b + \ntwo \big({}^{(2,3)}\Sigma\big)_c.
	\end{align*}
Identity \eqref{identityQ} follows after inserting \eqref{Sigmattt}, \eqref{symmetriesn} and \eqref{23p}. To prove \eqref{q} we contract $\mc{Q}_{\mu} = g^{\alpha\beta}\Sigma_{\mu\alpha\beta}$ with $\xi^{\mu}$,
\begin{equation}
	\label{a0}
\mf{q}  = \left(P^{ab}e_a^{\alpha}e_b^{\beta}  + 2n^b e_b^{\beta} \xi^{\alpha} + \ntwo \xi^{\alpha}\xi^{\beta}\right) \xi^{\mu} \Sigma_{\mu\alpha\beta}= P^{ab}\mS_{ab} + 2\big({}^{(1,2)}\Sigma\big)_b n^b +\ntwo \big({}^{(1,2,3)}\Sigma\big),
\end{equation}
and substitute \eqref{12n} and \eqref{123p}.
\end{proof}
\end{prop}

\section{Non-null case}

\label{sec_nonnull}

In this section we revisit the KID problem for spacelike Cauchy data in the language of hypersurface data. The existence results are known, both in the case of Killing vectors \cite{moncrief1975spacetime,coll1977evolution,beig1997killing} and in the homothetic case \cite{garcia2019conformal}. However, the form of the Killing initial data equations that we find in this paper are more flexible as they are written in an arbitrary gauge. As we shall discuss below, this form can be useful in practical applications. We reproduce the existence argument in the present setting for completeness. First we restate the classic well-posedness theorem of the Einstein field equations \cite{Choquet,Geroch} in the language of hypersurface data. Throughout this section we will assume $\bg$ to be positive definite and $\mc{A}$ of Lorentzian signature. Recall that by ``$\lambda$-vacuum'' we mean that equation \eqref{notacion} holds.

\begin{teo}
	\label{existencespacelike}
Let $\{\mc{H},\bg,\bm\ell,\elltwo,\bY\}$ be hypersurface data of dimension $\mf n$, $\lambda\in\real$ and assume the following constraint equations (cf. \eqref{Einxi}-\eqref{Eine})
	\begin{align}
P^{ac}A_{abc}n^b + \dfrac{1}{2}P^{ac}P^{bd}B_{abcd} &= -\dfrac{\mf{n}-1}{2}\lambda\label{cons1},\\
\ntwo P^{bd}A_{bcd} - A_{bcd}n^b n^d + P^{bd}B_{abcd}n^a &=0\label{cons2},
	\end{align}
where $A$ and $B$ are as in \eqref{A}-\eqref{B}. Then there exists a unique maximal and globally hyperbolic $\lambda$-vacuum spacetime $(\mc M,g)$, embedding $\Phi:\mc{H}\hookrightarrow\mc M$ and rigging $\xi$ such that $\{\mc{H},\bg,\bm\ell,\elltwo,\bY\}$ is $(\Phi,\xi)$-embedded on $(\mc M,g)$. 
\end{teo}
\begin{rmk}
	\label{rmk}
For hypersurface data of Lorentzian signature with $\bg$ positive definite there always exists a gauge satisfying $\bm\ell=0$ and $\elltwo=-1$ (and hence $P^{ab}=\gamma^{ab}$, $n=0$, $\ntwo=-1$) \cite{Marc1}. This gauge corresponds geometrically to a choice of rigging $\xi$ that is unit and normal to $\Phi(\mc H)$. Since in this case $\nu = -\xi$ (see \eqref{nuthetaa}) the tensor $\bY$ coincides with minus the second fundamental form $\bK$ of $\Phi(\mc H)$ w.r.t the normal $\bm{\nu}$. In addition, $\nablacero$ agrees with the Levi-Civita connection of $\bg$, ${}^{(\gamma)}\nabla$. Therefore (cf. \eqref{A}-\eqref{B}) $A_{bcd}=-2{}^{(\gamma)}\nabla_{[d}\K_{c]b}$ and $B_{abcd} = R^{(\gamma)}_{abcd}-2\K_{b[c}\K_{d]a}$, so the constraint equations \eqref{cons1}-\eqref{cons2} simplify to 
\begin{align}
R^{(\gamma)} - \K_{ab}\K^{ab} + \big(\tr_{\gamma}\bK\big)^2 &= -(\mf n-2)\lambda,\label{cons3}\\
{}^{(\gamma)}\nabla_b\big(\K^b{}_a-(\tr_{\gamma}\K)\delta^b{}_a\big) & = 0,\label{cons4}
\end{align}
which are the standard constraint equations for spacelike hypersurfaces. Although \eqref{cons1}-\eqref{cons2} are certainly more complicated than \eqref{cons3}-\eqref{cons4}, they are useful because they have the key property of being fully gauge covariant, and hence one can adapt the gauge to simplify the problem at hand. For instance, if one wants to solve the KID problem with a Killing vector that is transverse to the initial hypersurface (but not necessarily normal), one can fix the gauge so that $C=1$ and $\bar\eta=0$, and hence $\bY=0$ (see \eqref{Yembedded}). Geometrically this gauge corresponds to choosing $\xi = \eta|_{\Phi(\mc H)}$.
\end{rmk}
To construct the candidate to Killing/homothetic vector the strategy is to integrate the equation one gets by setting $\mc{Q}=0$ and $\ric=\lambda g$ in \eqref{defQ}. The initial data $\eta|_{\mc{H}}$ and $\lie_{\xi}\eta|_{\mc H}$ are to be selected so that $\mc{K}\st{\mc H}{=} \mu g$ and $\lie_{\xi}\mc{K} \st{\mc H}{=} \mu \lie_{\xi} g$ with $\mu$ a constant restricted to satisfy $\lambda\mu=0$. Since \eqref{equationk} is then a homogeneous wave equation with trivial initial data it follows that $\mc{K}=\mu g$ everywhere on $(\mc M,g)$. This approach requires first of all being able to write down $\k_{ab}$ and $\big(\lie_{\xi}\k\big)_{ab}$ in terms of hypersurface data in order to ensure $\k_{ab}=\mu\gamma_{ab}$ and $\big(\lie_{\xi}\k\big)_{ab}=2\mu\Y_{ab}$ at the abstract level. The former is simply (see \eqref{Adown})
\begin{equation}
	\label{kid1}
\mu\gamma_{ab} =2C\Y_{ab}  + 2\ell_{(a}\nablacero_{b)}C +\lie_{\bar\eta}\gamma_{ab}.
\end{equation}
For the latter we first insert \eqref{Sigmaxitt} into \eqref{aux} and use \eqref{Z2} to get the (still general) identity
\begin{align}
	\hspace{-0.19cm}\Phi^{\star} \big(\lie_{\xi}\mc{K}\big)_{ab}& = \dfrac{2C}{\ntwo}\left(\Rcero_{(ab)}- {R}_{ab} -2\lie_n\Y_{ab} - \left(2\kappa_n+\tr_P\bU-\ntwo\big(n(\elltwo)-\tr_P\bY\big)\right)\Y_{ab}  \right.\nonumber\\
	&\quad\, \left.+ \nablacero_{(a}\big(\s+2\r\big)_{b)}- 2\r_a\r_b +4\r_{(a}\s_{b)} - \s_a\s_b -\big(\tr_P\bY\big) \U_{ab} + 2P^{cd}\U_{c(a}\big(2\Y+\F\big)_{b)d}\right)\nonumber\\
	&\quad\,	+2C\left(-2\r_{(a}\nablacero_{b)}\elltwo+P^{cd}\big(\Pi_{ca}\Pi_{bd}+\Pi_{ac}\Pi_{db}\big) + \dfrac{1}{4}\ntwo \nablacero_a\elltwo\nablacero_b\elltwo-\dfrac{1}{2}n(\elltwo)\Y_{ab}\right.\nonumber\\
	&\qquad\qquad \, \left.  -\dfrac{1}{2}\nablacero_a\nablacero_b\elltwo\right) -2\elltwo\big(n(C)\Y_{ab}+\nablacero_a\nablacero_b C \big)-2\nablacero_{(a}C\nablacero_{b)}\elltwo+\big(\mathbb{k}-\bar\eta(\elltwo)\big)\K_{ab}\nonumber\\	
	&\quad\, +2\lie_{\bar\eta}\Y_{ab}-2\left((\lie_{\bar\eta}\bm\ell)(n)-\kt(n)\right) \Y_{ab}  -2\nablacero_{(a}\lie_{\bar\eta}\ell_{b)}+2\nablacero_{(a}\mf{K}_{b)}.\label{liexiK2}
\end{align} 
Thus, the equation one must impose at the initial data level is the same one but replacing $\big(\lie_{\xi}\k\big)_{ab}$ by $2\mu\Y_{ab}$, $\ktt$ by $\mu\elltwo$, $\kt$ by $\mu\bm\ell$ and ${R}_{ab}$ by $\lambda\gamma_{ab}$, i.e.
\begin{align}
2\mu\Y_{ab}& = \dfrac{2C}{\ntwo}\left(\Rcero_{(ab)}- \lambda\gamma_{ab} -2\lie_n\Y_{ab} - \left(2\kappa_n+\tr_P\bU-\ntwo\big(n(\elltwo)-\tr_P\bY\big)\right)\Y_{ab}  \right.\nonumber\\
	&\quad\, \left.+ \nablacero_{(a}\big(\s+2\r\big)_{b)}- 2\r_a\r_b +4\r_{(a}\s_{b)} - \s_a\s_b -\big(\tr_P\bY\big) \U_{ab} + 2P^{cd}\U_{c(a}\big(2\Y+\F\big)_{b)d}\right)\nonumber\\
	&\quad\,	+2C\left(-2\r_{(a}\nablacero_{b)}\elltwo+P^{cd}\big(\Pi_{ca}\Pi_{bd}+\Pi_{ac}\Pi_{db}\big) + \dfrac{1}{4}\ntwo \nablacero_a\elltwo\nablacero_b\elltwo-\dfrac{1}{2}n(\elltwo)\Y_{ab} \right.\nonumber\\
	&\qquad\qquad\,\left. -\dfrac{1}{2}\nablacero_a\nablacero_b\elltwo\right) -2\elltwo\big(n(C)\Y_{ab}+\nablacero_a\nablacero_b C \big)-2\nablacero_{(a}C\nablacero_{b)}\elltwo-\big(\mu\elltwo+\bar\eta(\elltwo)\big)\K_{ab} \nonumber\\	
	&\quad\,+2\lie_{\bar\eta}\Y_{ab}-2\left((\lie_{\bar\eta}\bm\ell)(n)-\mu\right) \Y_{ab}  -2\nablacero_{(a}\lie_{\bar\eta}\ell_{b)}.\label{kid2}
\end{align} 

The hypersurface data version of the existence theorem for homothetic/Killing initial data is the following.

\begin{teo}
	\label{teo_espacial}
	Let $\{\mc{H},\bg,\bm\ell,\elltwo,\bY\}$ be hypersurface data and $\lambda\in\real$ such that the constraint equations \eqref{cons1}-\eqref{cons2} hold. Assume that $(C,\bar\eta)\in\mc{F}(\mc H)\times\X(\mc H)$ satisfy \eqref{kid1} and \eqref{kid2} with $\mu\in\real$ such that $\lambda\mu=0$. Then there exists a unique maximal, globally hyperbolic $\lambda$-vacuum spacetime $(\mc M,g)$, embedding $\Phi:\mc{H}\hookrightarrow\mc M$, rigging $\xi$ and vector field $\eta$ such that $\{\mc{H},\bg,\bm\ell,\elltwo,\bY\}$ is $(\Phi,\xi)$-embedded on $(\mc M,g)$, $\eta|_{\Phi(\mc H)} = C\xi + \Phi_{\star}\bar\eta$ and $\lie_{\eta}g=\mu g$ on $\mc M$.
	\begin{proof}
By Theorem \ref{existencespacelike} there exists a unique maximal, globally hyperbolic $\lambda$-vacuum spacetime $(\mc M,g)$, embedding $\Phi:\mc{H}\hookrightarrow\mc M$ and rigging $\xi$ such that $\{\mc{H},\bg,\bm\ell,\elltwo,\bY\}$ is $(\Phi,\xi)$-embedded on $(\mc M,g)$. Let us extend $\xi$ off $\Phi(\mc H)$ by $\nabla_{\xi}\xi=0$ and construct the (unique) vector field $\eta\in\X(\mc M)$ by integrating (cf. \eqref{defQ}) $$\square_g \eta^{\mu} + \lambda\eta^{\mu}=0$$ with initial conditions $\eta|_{\mc H} = C\xi + \Phi_{\star}\bar\eta$ and 
\begin{equation}
	\label{condini}
	\big(\lie_{\xi}\bm\eta\big)_a = \mu\ell_a - \elltwo\nablacero_aC - \nablacero_a\big(\bm\ell(\bar\eta)\big)+2\big(\Y_{ac}+\F_{ac}\big)\bar\eta^c,\qquad \big(\lie_{\xi}\bm\eta\big)(\xi) = \dfrac{1}{2}\mu\elltwo.
\end{equation} 
By the identities \eqref{liet} and \eqref{lietr} it follows that $\kt=\mu\bm\ell$ and $\ktt=\mu\elltwo$. Equation \eqref{kid1} guarantees (by identity \eqref{Adown}) that $\k_{ab}=\mu\gamma_{ab}$, and hence $\mc{K}\st{\mc H}{=} \mu g$. Therefore, condition \eqref{kid2} along with identity \eqref{liexiK2} implies $\big(\lie_{\xi}\k\big)_{ab} \st{\mc H}{=} 2\mu \Y_{ab}$. Observe also that $\mS=0$ (see Remark \ref{rmkmS}). To make sure that $\lie_{\xi}\k=\mu\lie_{\xi} g$ on the initial hypersurface $\Phi(\mc H)$ it only remains to show that $\kt_c^{(2)} =\mu \Phi^{\star}\big(\big(\lie_{\xi}g\big)(\xi,\cdot)\big)_c$ and $\ktt^{(2)} = \Phi^{\star}\big(\big(\lie_{\xi}g\big)(\xi,\xi)\big)$.\\

Inserting $\mc{K}_{ab}=\mu\gamma_{ab}$, $\kt_a=\mu\ell_a$, $\ktt=\mu\elltwo$ and $\mS=0$ into \eqref{identityQ} with $\mc{Q}_c=0$ (and $\beta=a_{\para}=0$) and using \eqref{gamman}-\eqref{Pgamma} and \eqref{nablagamma} gives
\begin{align*}
0 &=  2\mu\lie_n\ell_c +2\mu\elltwo \U_{bc}n^b + \big(2\kappa_n + \tr_P\bU \big)\mu\ell_c + 2\mu \big(1-\ntwo\elltwo\big) \big(\r_c-\s_c + \ntwo\nablacero_c\elltwo\big)\nonumber \\
&\quad\,  -2 \mu P^{ab} \U_{a(b}\ell_{c)} +\mu P^{ab}\U_{c(a}\ell_{b)}   - 2\mu\big(\delta^b_c - n^b\ell_c\big)(\r_b+\s_b)  + \mu\elltwo\nablacero_c \ntwo \nonumber \\
&\quad\, +\ntwo \left(\kt_c^{(2)} - \dfrac{1}{2}\mu\nablacero_c\elltwo+2\mu\elltwo \s_c + \left(2\r_c-2\s_c + \ntwo \nablacero_c\elltwo\right)\mu\elltwo\right)\\
&= 2\mu\lie_n\ell_c +2\mu\elltwo \U_{bc}n^b  -4\mu\s_c + \mu\elltwo\nablacero_c \ntwo +\ntwo\left(\kt^{(2)}_c + \dfrac{3}{2}\mu\nablacero_c\elltwo +2\mu\elltwo\s_c - \mu\ntwo\elltwo \nablacero_c\elltwo\right).
\end{align*}
Using now \eqref{lienell} and \eqref{Un} this equation becomes
\begin{equation}
	\ntwo \left(\kt_c^{(2)}-\dfrac{1}{2}\mu\nablacero_c\elltwo\right)=0 \qquad \Longrightarrow\qquad \kt_c^{(2)} = \dfrac{1}{2}\mu\nablacero_c\elltwo.
\end{equation}
By the equality $\Phi^{\star}\big(\big(\lie_{\xi}g\big)(\xi,\cdot)\big) = \Phi^{\star}\big(\lie_{\xi}\bm\xi\big) = \frac{1}{2}\nablacero\elltwo$ valid when $\nabla_{\xi}\xi=0$ (see \eqref{recall}) we conclude $\kt^{(2)} =\mu \Phi^{\star}\big(\big(\lie_{\xi}g\big)(\xi,\cdot)\big)$, as needed. Similarly, inserting $\mS=0$, $\ktt=\mu\elltwo$ and $\kt=\mu\bm\ell$ into \eqref{q} with $\mf{q}=0$ (and $\beta= a_{\para}=0$) gives 
\begin{equation}
	\ntwo \ktt^{(2)} = 0\qquad \Longrightarrow \qquad \ktt^{(2)} = 0.
\end{equation}
Since $\Phi^{\star}\big(\big(\lie_{\xi}g\big)(\xi,\xi)\big) = 0$ (see \eqref{recall2}) we conclude $\ktt^{(2)} = \Phi^{\star}\big(\big(\lie_{\xi}g\big)(\xi,\xi)\big)$, and therefore $\lie_{\xi}\k|_{\Phi(\mc H)} = \mu \lie_{\xi}g|_{\Phi(\mc H)}$. Equation \eqref{equationk} with $\mc Q=0$, namely 
\begin{equation*}
		\square_g \big(\mc{K}_{\beta\nu}-\mu g_{\beta\nu}\big) - 2R_{\sigma(\beta}\big(\mc{K}^{\sigma}{}_{\nu)} - \mu\delta^{\sigma}_{\nu)}\big)		+ 2R_{\sigma\beta\mu\nu}\big(\mc{K}^{\sigma\mu}-\mu g^{\sigma\mu}\big) + 2\lambda\big(\mc{K}_{\beta\nu}-\mu g_{\beta\nu}\big)  = 0,
\end{equation*}
then proves $\mc{K} = \mu g$ everywhere on $\mc M$, and hence $\eta$ is a homothety or Killing vector (depending on the value of $\mu$).
	\end{proof}
\end{teo}

The existence result for the spacelike KID problem can be generalized to include general matter \cite{coll1977evolution,racz1999existence}, in the sense that the matter field itself is invariant along the Killing vector.

\begin{rmk}
Equations \eqref{kid1}-\eqref{kid2} are the hypersurface data version of the standard KID equations. Indeed, if one chooses the gauge as in Remark \ref{rmk} these equations become
\begin{align}
0 & =-2C\K_{ab}  +2{}^{(\gamma)}\nabla_{(a}\bar\eta_{b)} -	\mu\gamma_{ab},\label{kid1b}\\
0& = C\left(-R^{(\gamma)}_{ab}+\lambda\gamma_{ab}+2\gamma^{cd}\K_{ac}\K_{bd}-(\tr_{\gamma}\bK)\K_{ab}\right)- \lie_{\bar\eta}\K_{ab}  +\dfrac{1}{2}\mu \K_{ab} +\nabla_a^{(\gamma)}\nabla_b^{(\gamma)} C,\label{kid2b}
\end{align}
 which are the standard KID equations (see \cite{moncrief1975spacetime,coll1977evolution,beig1997killing} for the Killing case and \cite{garcia2019conformal} for homo\-theties). The advantage of \eqref{kid1}-\eqref{kid2} w.r.t \eqref{kid1b}-\eqref{kid2b} is that the gauge is fully free, so it can be adapted to one's convenience. For instance, if one wants to construct a Killing/homothetic vector transverse to the initial hypersurface (but not necessarily normal) it may be easier to choose the gauge such that $C=1$ and $\bar\eta=0$, because then $\bY=2\mu\bg$ and the constraint and KID equations take a much simpler form. Geometrically this gauge corresponds to choosing $\xi = \eta|_{\Phi(\mc H)}$.
\end{rmk}

\begin{rmk}
The KID equations \eqref{kid1} and \eqref{kid2} are at the level of the initial data, in the sense that they do not involve more data that the one in Theorem \ref{existencespacelike} (apart from the pair $(C,\bar\eta)$, of course).
\end{rmk}

\section{Null hypersurfaces}

\label{sec_null}

In this section we find necessary conditions on a null hypersurface to be embeddable in a manifold $(\mc M,g)$ with a homothetic vector field $\eta$, i.e. $\k=\mu g$. We will then analyse their sufficiency in two specific cases, namely the characteristic problem and the smooth spacelike-characteristic problem. Other Cauchy problems (such as the spacelike-characteristic with corner) could be treated similarly. \\

In \cite{chrusciel2013kids} the authors perform a similar analysis for Killing vectors and derive a set of equations on an embedded null hypersurface that relate the would-be Killing vector on the hypersurface, the spacetime Levi-Civita connection and Riemann tensor, as well as the induced metric on cross-sections of $\mc H$. Then they study their sufficiency in two specific scenarios, namely the light-cone case and two null hypersurfaces intersecting in a spacelike surface. The equations presented in \cite{chrusciel2013kids} are formulated in a particular coordinate system and with a specific choice of transverse vector, while in this paper we find analogous equations in a gauge covariant and coordinate-free manner. More importantly, in Section \ref{sec_charact} we will show that the equations can be written solely in terms of the initial data for the characteristic $\lambda$-vacuum equations, thus suppressing any reference to the ambient space one wishes to construct. Consequently, we will be able to reformulate the existence theorem of homothetic/Killing vectors for the characteristic problem in a fully detached way. This puts the Killing initial data problem in the characteristic case in an equal footing to the KID problem in the standard Cauchy case, where the KID equations are written in terms of initial value quantities, with a priori no reference to the spacetime to be constructed.\\

In any of the two cases we will examine, the initial data for the wave equation in the null subset of the hypersurface consists only of the value of the field, and not any of its transverse derivatives. This makes the analysis of the KID conditions different to the spacelike case, as now well-posedness of equations \eqref{defQ} and \eqref{equationk} only requires initial data for $\eta|_{\mc H}$ and $\mc{K}|_{\mc H}$ (and not for their first transverse derivative). The first immediate consequence is that one cannot use \eqref{liet}-\eqref{lietr} to achieve $\kt=\mu\bm\ell$ and $\ktt=\mu\elltwo$ as in the non-null case. Instead, one must prove that $\kt=\mu\bm\ell$ and $\ktt=\mu\elltwo$ hold by some other means. The strategy is to find transport equations for the tensors $\kt-\mu\bm\ell$ and $\ktt-\mu\elltwo$. To do that we start by writing down \eqref{identityQ}-\eqref{q} with $\ntwo=0$,
\begin{align}
\mc{Q}_c &=-2\mS_{ac}n^a + 2\lie_n\kt_c -2P^{ab}\U_{cb}\kt_a +\big(2\kappa_n+\tr_P\bU\big)\kt_c + 2\kt(n)\big(\r-\s \big)_c \nonumber\\
&\quad\, + P^{ab}\nablacero_a\k_{bc} - \dfrac{1}{2}P^{ab}\nablacero_c\k_{ab} + \left(\tr_P\bY-n(\elltwo)\right)\k_{cd}n^d - 2P^{bd}(\r+\s)_b\k_{dc} ,\label{eqQ}\\
\mf{q} & =\tr_P\mathbb{\Sigma} + n\big(\mathbb{k}\big) + 2\kappa_n \mathbb{k} - 2P^{ab}(\r+\s)_a\mf{K}_b    - n(\elltwo)\mf{K}(n),\label{eqq}
\end{align} 
as well as the contraction of \eqref{eqQ} with $n^c$ (we shall use that $\U_{cb}n^c = 0$ repeatedly),
\begin{equation}
	\label{eqQn}
	\begin{aligned}
\mc{Q}_cn^c &= -2\mathbb{\Sigma}_{bc}n^bn^c +2\lie_n\big(\kt(n)\big) +\big(\tr_P\bU\big) \mf{K}(n)+P^{ab}n^c\nablacero_a\mc{K}_{bc}- \dfrac{1}{2}P^{ab}n^c\nablacero_c\k_{ab}\\
&\quad\, + \big(\tr_P\bY-n(\elltwo)\big)\mc{K}_{cd}n^cn^d-2P^{bd}(\s+\r)_b\k_{dc}n^c.
\end{aligned}
\end{equation}

As we show explicitly in the following lemma, the contraction of $\mS_{ab}$ with $n^b$ only depends on hypersurface data and the pair $(C,\bar\eta)$. This will be crucial in Section \ref{sec_charact} to show that the KID equations are at the level of the initial data.

\begin{lema}
	\label{mSnn}
Let $\{\mc H,\bg,\bm\ell,\elltwo,\bY\}$ be null hypersurface data $(\Phi,\xi)$-embedded in $(\mc M,g)$. Let $\eta\in\X(\mc M)$ and define $\eta\st{\mc H}{=} C\xi + \Phi_{\star}\bar\eta$. Then,
\begin{align}
\hspace{-0.3cm}\mS_{ab}n^b &= -C\left( \ric(\xi,e_a) +P^{bc}A_{bca} + P^{bc}\big(\r_b+\s_b\big)\big(\Y_{ac}+\F_{ac} \big)  - \dfrac{1}{2}\kappa_n\nablacero_a\elltwo -2\nablacero_a\big(n(\elltwo)\big)+2P^{bc}\U_{ab}\nablacero_c\elltwo\right) \nonumber \\
&\quad\, -n^b \lie_{\bar\eta}\Y_{ab} + \left(\elltwo n(C)+\big(\lie_{\bar\eta}\bm\ell\big)(n)\right)\big(\r_a-\s_a\big) + \dfrac{1}{2}n(\elltwo)\nablacero_aC + \dfrac{1}{2}n(C)\nablacero_a\elltwo\nonumber \\
&\quad\,   + \elltwo\big(\nablacero_a\big(n(C)\big) - P^{bc}\U_{ab}\nablacero_cC  \big) +\dfrac{1}{2}\lie_n\lie_{\bar\eta}\ell_a  +\dfrac{1}{2} \nablacero_a \big((\lie_{\bar\eta}\bm\ell)(n)\big)  - P^{bc}\U_{ab}\lie_{\bar\eta}\ell_c\label{mSn}
\end{align}
and
\begin{align}
	\mS_{ab}n^an^b &= -C\left( \ric(\xi,\nu) +P^{bc}A_{bca}n^a + P\big(\br+\bs,\br+\bs\big) - \dfrac{1}{2}\kappa_nn(\elltwo) -2\lie_n^{(2)}\elltwo\right)-n^an^b \lie_{\bar\eta}\Y_{ab} \nonumber \\
	&\quad\,  - \left(\elltwo n(C)+\big(\lie_{\bar\eta}\bm\ell\big)(n)\right)\kappa_n + n(\elltwo)n(C)   + \elltwo\lie_n^{(2)}C  +\lie_n\big(\big(\lie_{\bar\eta}\bm\ell\big)(n)\big)  .\label{mSnn2}
\end{align}
\begin{proof}
The starting point is \eqref{Sigmaxitt}, which after contraction with $n^b$ gives
\begin{align*}
\mS_{ab}n^b &= -C\z_a^{(2)} - n^b\lie_{\bar\eta}\Y_{ab}+ \left(\dfrac{1}{2}C n(\elltwo) + \elltwo n(C)+\big(\lie_{\bar\eta}\bm\ell\big)(n)\right)\r_a + \dfrac{1}{2}n(\elltwo)\nablacero_aC + \dfrac{1}{2}n(C)\nablacero_a\elltwo \\
&\quad\, + \dfrac{1}{2}C n^b \nablacero_a\nablacero_b \elltwo + \elltwo n^b\nablacero_a\nablacero_b C + n^b\nablacero_{(a}\lie_{\bar\eta}\ell_{b)}.
\end{align*}
We already know that $\bz^{(2)}$ is expressible in terms of hypersurface data and Ricci tensor components. Thus, \eqref{mSn} follows after inserting \eqref{z2} and using $$n^b \nablacero_a\nablacero_b f = \nablacero_a n(f) - (\nablacero_b f)\nablacero_an^b \st{\eqref{derivadan}}{=} \nablacero_a \big(n(f) \big)- P^{bc}\U_{ab}\nablacero_c f - n(f) \s_a$$ together with \eqref{nnablacerotheta} in Lemma \ref{lemmannabla} applied to $\theta_b=\lie_{\eta}\ell_b$. To obtain \eqref{mSnn2} it suffices to contract \eqref{mSn} with $n^a$ again.
\end{proof}
\end{lema}
In the following lemma we show that \eqref{eqQ}-\eqref{eqQn} possess a hierarchical structure that eventually leads to transport equations for $\kt$ and $\ktt$.
\begin{lema}
	\label{lemaagrupado}
Assume $\k$ is of the form $\k = \mu\bg$ for some $\mu\in\real$. Then,
\begin{enumerate}
	\item $\mc{Q}_cn^c = -2\mathbb{\Sigma}_{ac}n^an^c +2\lie_n\big(\kt(n)-\mu\big) +\big(\kt(n)-\mu\big)\tr_P\bU$.
	\item If, moreover, $\kt(n)=\mu$, then $$\mc{Q}_c = -2\mathbb{\Sigma}_{ac}n^a + 2\lie_n\big(\mf{K}_c-\mu\ell_c\big)  -2P^{ab}\U_{bc}\big(\mf{K}_a-\mu\ell_a\big)  +\big(2\kappa_n+\tr_P\bU\big)\big(\mf{K}_c-\mu\ell_c\big).$$
	\item If $\kt=\mu\bm\ell$, then $\mf{q} = \tr_P\mS +n\big(\mathbb{k}-\mu\elltwo\big) + 2\kappa_n \big(\mathbb{k}-\mu\elltwo\big)$.
\end{enumerate}
\begin{proof}
The three items are consequences of the identities in Lemma \ref{lemaQs}, or more specifically, of their simplification \eqref{eqQ}-\eqref{eqq} in the case $\ntwo=0$. First note that under the hypothesis $\k = \mu\bg$ together with $\lie_n\ell_c = 2\s_c$, $P^{ab}\U_{bc}\ell_a = 0$, \eqref{combina} and \eqref{Pgamma}, expression \eqref{eqQ} can be written as 
\begin{align*}
	\mc{Q}_c &= -2\mS_{ca}n^c +2 \lie_n\big(\kt_c-\mu\ell_c\big)  -2P^{ab}\U_{bc}\big(\kt_a-\mu\ell_a\big) +\big(2\kappa_n+\tr_P\bU\big)\big(\kt_c-\mu\ell_c\big)\\
	&\quad\, + 2\big(\kt(n)-\mu\big)\big(\r-\s\big)_c.
\end{align*}
The contraction of this with $n^c$ gives item 1., and the substitution $\kt(n)=\mu$ gives item 2. Finally, the last item is immediate from \eqref{eqq} after using \eqref{Pell}.
\end{proof}
\end{lema}
\begin{rmk}
Item 3. in this lemma does not require the assumption $\k = \mu\bg$, but presented in this way we emphasize the hierarchical structure that shall be used below.
\end{rmk}

When the candidate to Killing field $\eta$ is constructed by solving $\mc{Q}=0$ in the bulk spacetime, then the identities in Lemma \ref{lemaagrupado} become homogeneous transport equations on $\mc H$ for $\kt$ and $\ktt$ provided one already knew that $\k_{ab}=\mu\gamma_{ab}$, $\mathbb{\Sigma}_{bc}n^b=0$ and $\tr_P\mS=0$. So it is necessary to find conditions to guarantee these three conditions on the null hypersurface. They will be obtained as corollaries of Proposition \ref{prop_lienSigma} below, where we show that the tensor $\mS$ satisfies a transport equation along $n$. The proof of Proposition \ref{prop_lienSigma} is somewhat long, so we postpone it to Appendix \ref{app_proof}.

\begin{prop}
	\label{prop_lienSigma}
Let $\{\mc H,\bg,\bm\ell,\elltwo,\bY,\bZ^{(2)}\}$ be extended null hypersurface data $(\Phi,\xi)$-embedded in $(\mc M,g)$ and let $\eta\in\X(\mc M)$ and $\Sigma\d\lie_{\eta}\nabla$. Then,
\begin{equation}
	\label{identity}
	\begin{aligned}
2\lie_n \mS_{ab} -2\mS(n,n) \Y_{ab} - \lie_{\Xi}\gamma_{ab}- 2\ell_{(a}\nablacero_{b)}\big(\mS(n,n)\big) - 4P^{cd}\U_{c(a}\mS_{b)d}& \\
		+4\big(\r-\s\big)_{(a}\mS_{b)c}n^c + \big(2\kappa_n+\tr_P\bU\big)\mS_{ab} +\big(\tr_P\mS\big)\U_{ab}&=\mc{I}_{ab},
	\end{aligned}
\end{equation} 
where $\Xi^a \d P^{ab}\mS_{bc}n^c$,
\begin{align*}
	\mc{I}_{ab} &=\big(\lie_{\eta}R\big)_{ab} + \dfrac{1}{2}\nablacero_a\nablacero_b \big(\tr_P\k + 2\kt(n)\big) +\dfrac{1}{2} \lie_n\big(\tr_P\k + 2\kt(n)\big) \Y_{ab}  - P^{cd}\nablacero_c \Sigma_{dab} +2V^c{}_{(a} \Sigma_{d|b)c}n^d\\
	&\quad\, + \left(P^{cd}\nablacero_c\kt_d + \left(\tr_P\bY-\dfrac{1}{2}n(\elltwo)\right)\kt(n)  + \left(\tr_P\bU+\kappa_n\right)\ktt- 2P\big(\kt,\br\big) -P^{cd}V^f{}_c\k_{df}+\dfrac{1}{2}n(\ktt)\right)\U_{ab} \\
	&\quad\, +\lie_n\left(\nablacero_{(a}\kt_{b)} + \ktt \U_{ab}\right) - \k(n,n)\Z^{(2)}_{ab} -\dfrac{1}{2} \nablacero_{(a}\elltwo \nablacero_{b)}\big(\k(n,n)\big) -\lie_X\Y_{ab} - \big(\chi+\lie_n(\kt(n))\big)\Y_{ab} \\
	&\quad\,- \ell_{(a}\nablacero_{b)}\big(\chi+2\lie_n(\kt(n))\big) -\dfrac{1}{2}\lie_W\gamma_{ab}+2P^{cd}(\r+\s)_c\Sigma_{dab}-\big(\tr_P\bY-n(\elltwo)\big)\Sigma_{cab}n^c -2P^{cd}\Y_{c(a|}\Sigma_{de|b)}n^e \\
	&\quad\,-2\s_{(a}\nablacero_{b)}\big(\kt(n)\big) -2\big(\r-\s\big)_{(a}V^d{}_{b)}\k_{cd}n^c -2P^{cd}\U_{c(a}\left( \nablacero_{b)}\kt_{d} + \kt(n)\Y_{b)d} + \ktt \U_{b)d}  - V^c{}_{b)} \k_{cd}\right),
\end{align*}
and $X$, $\chi$ and $W$ are defined by $X^a\d P^{ab}\k_{bc}n^c$, $\chi\d -2\br(X)  - \dfrac{1}{2} n(\elltwo)\k(n,n)$, and
\begin{align*}
	W^a&\d P^{ab} \left(\lie_n\kt_b +\nablacero_b \big(\kt(n)\big) -2\kt(n)\s_b -2P^{dc}\U_{bd}\kt_c  - V^c{}_b \k_{dc}n^d - \big(\Y_{cb}+\F_{cb}\big)X^c \right) \\
	&\quad\, +\dfrac{1}{2}\big(n(\ktt)-X(\elltwo)-n(\elltwo) \kt(n)\big)n^a.
\end{align*}
\end{prop}

In the following corollary we show that when $\k|_{\Phi(\mc H)}$ is pure trace, identity \eqref{identity} admits again a hierarchical structure and can be rewritten as transport equations for $\tr_P\k$, all components of $\mS_{ab}n^b$ except $\mS(n,n)$ and $\tr_P\mS$. Before proving this we recall that given a 2-covariant, symmetric tensor field $T$, it can be decomposed uniquely as \cite{MarcAbstract} 
\begin{equation}
	\label{decomposition}
	T_{ab} = \dfrac{\tr_P T}{\mf n-1}\gamma_{ab} + 2\ell_{(a}T_{b)c}n^c + T(n,n)\left(\dfrac{\elltwo}{n-1}\gamma_{ab}-\ell_a\ell_b\right) + \wh{T}_{ab},
\end{equation} 
where $\wh{T}$ is a symmetric tensor that satisfies $P^{ab}\wh{T}_{ab}=0$ and $\wh{T}_{ab}n^a=0$. It follows that when $T = \mu \gamma_{ab}$ then $\wh{T}_{ab}=0$ and $T_{ab} n^b=0$. For the first item in the next corollary we only need to impose these two last conditions on $\k_{ab}$ while for items 2. and 3. we also need restrictions on the transverse components of $\k_{\mu\nu}$. The conditions needed here are well adapted to the ones appearing in Lemma \ref{lemaagrupado}. This is what allows one to close the argument in Theorem \ref{teorema_null} below.
\begin{cor}
\label{cor_junto}
Assume $\k_{ab}n^b=0$, $\wh{\k}=0$ and $R_{\alpha\beta}=\lambda g_{\alpha\beta}$ in a neighbourhood\footnote{In fact it suffices that this equality holds up to first derivatives on $\Phi(\mc H)$.} of $\Phi(\mc H)$ and let $\mu\in\real$ satisfying $\lambda\mu=0$. Then, 
\begin{enumerate}
\item $\left(\lie_n^{(2)}+\left(\dfrac{2\tr_P\bU}{\mf{n-1}}-\kappa_n\right)\lie_n\right)\tr_P\k = 2\big(\tr_P\bU \big)\mS(n,n).$
\item If, moreover, $\k_{ab}=\mu\gamma_{ab}$ and $\kt(n)=\mu$, one has $$\lie_n\big(\mS_{ab}n^b\big)-\nablacero_a\big(\mS(n,n)\big)+\big(\tr_P\bU\big) \mS_{ab}n^b = 0.$$
\item If $\k_{ab}=\mu\gamma_{ab}$ and in addition $\kt=\mu\bm\ell$ and $\mS(n,\cdot)=0$, then 
\begin{equation*}
	\begin{aligned}
		2\lie_n\big(\tr_P\mS\big) +  2\big(\kappa_n+\tr_P\bU\big)\tr_P\mS&= \big((\ktt-\mu\elltwo)\tr_P\bU + n\big(\ktt-\mu\elltwo\big) (\ktt-\mu\elltwo)\kappa_n\big) \\
		&\quad\, +\big(\ktt-\mu\elltwo\big) \lie_n(\tr_P\bU) .
	\end{aligned}
\end{equation*}
\end{enumerate}
\textbf{Remark:} Contracting the expression in item 2. with $n^a$ gives $(\tr_P\bU)\mS(n,n)=0$, which implies $\mS(n,n)=0$ at the points where $\tr_P\bU\neq 0$. The same conclusion follows from item 1. provided $\k_{ab}=\mu\gamma_{ab}$. However, if $\tr_P\bU$ vanishes on open sets of $\mc H$ one cannot use these equations to obtain the value of $\mS(n,n)$. This is why we do not assume $\k_{ab}=\mu\gamma_{ab}$ from the beginning. 
\begin{proof}
We start by noting that the quantities $X^a$ and $\chi$ in Proposition \ref{prop_lienSigma} vanish identically and that $\k_{ab}=\frac{1}{\mf n-1}(\tr_P\k)\gamma_{ab}$. Inserting this into \eqref{Sigmattt} yields 
\begin{align}
\hspace{-0.5cm}\Sigma_{abc} &= \dfrac{1}{2(\mf n-1)}\big(2\gamma_{a(c}\nablacero_{b)}(\tr_P\k)-\gamma_{bc}\nablacero_a(\tr_P\k)\big)+\dfrac{\tr_P\k}{2(\mf n-1)}\big(2\nablacero_{(b}\gamma_{c)a}-\nablacero_a\gamma_{bc}\big)+\kt_a\U_{bc}\nonumber\\
&\st{\eqref{nablagamma}}{=}\dfrac{1}{2(\mf n-1)}\left(2\gamma_{a(c}\nablacero_{b)}(\tr_P\k)´-\gamma_{bc}\nablacero_a(\tr_P\k)\right) + \left(\kt_a-\dfrac{\tr_P\k}{\mf n-1}\ell_a\right)\U_{bc}.\label{Sigmafacil}
\end{align}
The following contractions are immediate
\begin{equation}
	\label{sigmafacil2}
	\Sigma_{abc}n^b =\dfrac{n(\tr_P\k)}{2(\mf n-1)}\gamma_{ac},\qquad \Sigma_{abc}n^an^b =0,\qquad \Sigma_{abc}n^b n^c = 0,
\end{equation}
and hence
\begin{equation}
	\label{facil}
	P^{cd}n^an^b \nablacero_c\Sigma_{dab}  = -2 P^{cd} \Sigma_{dab} n^a \nablacero_c n^b \st{\eqref{derivadan}}{=} -2P^{cd}P^{bf}\U_{fc}\Sigma_{dab}n^a= -\dfrac{\tr_P\bU}{\mf n-1} n\big(\tr_P\mc{K}\big) ,
\end{equation}
where in the last equality we used $P^{cd}P^{bf}\gamma_{db} = P^{cf}+\elltwo n^c n^f$, which is a consequence of \eqref{Pell}-\eqref{Pgamma}. Many terms in the contraction of $\mc{I}_{ab}$ with $n^a n^b$ vanish identically. Using $\big(\lie_{\eta}R\big)_{ab} =\lambda\k_{ab}$, $n^a\nablacero_a n^b=0$ and inserting \eqref{facil} then gives
\begin{equation}
	\label{Inn}
	\mc{I}(n,n) = \dfrac{1}{2}\left(\lie_n^{(2)}+\left(\dfrac{2\tr_P\bU}{\mf{n-1}}-\kappa_n\right)\lie_n\right)\tr_P\k.
\end{equation}
The contraction of the LHS of \eqref{identity} with $n^a n^b$ simplifies to $(\tr_P\bU)\mS(n,n)$, and hence item 1. of the corollary follows.\\

To prove the second item we observe that now $\tr_P\k = (\mf n-1)\mu$, so \eqref{Sigmafacil} becomes $\Sigma_{abc} = \big(\kt_a-\mu\ell_a\big)\U_{bc}$ and hence $\Sigma_{abc}$ contracted with $n$ in any of its indices vanishes. Moreover, the vector $W^a$ (see Prop. \ref{prop_lienSigma}) simplifies to $$W^a = P^{ab}\left(\lie_n\big(\kt_b-\mu\ell_b\big)-2P^{cd}\U_{bc}\big(\kt_d-\mu\ell_d\big)\right)+\dfrac{1}{2} n(\ktt-\mu\elltwo) n^a,$$ where we used $2\bs=\lie_n\bm\ell$ (see \eqref{lienell}) and $P^{cd}\U_{bc}\ell_d=0$. These two relations will be used throughout the rest of the proof without further notice. Then, since the term $\lie_{\eta}R_{ab}$ in \eqref{identity} is zero because $\lambda\mu=0$,
\begin{align*}
\mc{I}_{ab} &=  -P^{cd} \nablacero_c\big((\kt_d-\mu\ell_d)\U_{ab}\big) +\left(P^{cd}\nablacero_c\kt_d + \left(\tr_P\bY-\dfrac{1}{2}n(\elltwo)\right)\mu + \big(\tr_P\bU+\kappa_n\big)\ktt\right.\nonumber\\
&\quad\,\left. -2P(\kt,\br) - \mu\big(\delta^c_f - n^c\ell_f\big)V^f{}_c + \dfrac{1}{2}n(\ktt)\right)\U_{ab} +\lie_n\big(\nablacero_{(a}\kt_{b)} + \ktt\U_{ab} \big)-\dfrac{1}{2}\lie_W\gamma_{ab} \nonumber\\
&\quad\, +2P^{cd}(\r+\s)_c \big(\kt_d-\mu\ell_d\big)\U_{ab}-2P^{cd}\U_{c(a}\left(\nablacero_{b)}\kt_d + \mu\Y_{b)d} + \ktt\U_{b)d} - \mu V^c{}_{b)}\gamma_{cd}\right),
\end{align*}
which after using $V^c{}_c = \tr_P\bY+\frac{1}{2}n(\elltwo)$, $n^c\ell_f V^f{}_c=\kappa_n\elltwo + \dfrac{1}{2}n(\elltwo)$ (see \eqref{Vn}) and $V^c{}_b\gamma_{cd}=\Y_{bd}+\F_{bd}-(\r-\s)_b\ell_d$ (see \eqref{gammaV}) simplifies to
\begin{align}
	\mc{I}_{ab} &=  -P^{cd} \nablacero_c\big((\kt_d-\mu\ell_d)\U_{ab}\big) +\left(P^{cd}\nablacero_c\kt_d + \big(\tr_P\bU\big)\ktt +\kappa_n\big(\ktt+\mu\elltwo\big)-2P(\kt,\br) \right.\nonumber\\
	&\quad\,\left.  + \dfrac{1}{2}n\big(\ktt-\mu\elltwo\big)\right)\U_{ab} +\lie_n\big(\nablacero_{(a}\kt_{b)} + \ktt\U_{ab} \big)-\dfrac{1}{2}\lie_W\gamma_{ab} +2P^{cd}(\r+\s)_c \big(\kt_d-\mu\ell_d\big)\U_{ab}\nonumber\\
	&\quad\, -2P^{cd}\U_{c(a}\left(\nablacero_{b)}\kt_d - \mu\F_{b)d} + \ktt\U_{b)d} \right).\label{Ifacil}
\end{align}
The contraction of $\mc{I}_{ab}$ with $n^b$ then gives, after using \eqref{derivadan} and $\U_{ab}n^b=0$ in the first equality and identities \eqref{contra2} and \eqref{nnablacerotheta} in Lemma \ref{lemmannabla} applied to $\theta_c=\kt_c$ (with $\ntwo=0$) in the second,
\begin{align}
	\mc{I}_{ab}n^b &= P^{cd}P^{bf}\U_{fc} \U_{ab}(\kt_d-\mu\ell_d) + \lie_n\big(n^b\nablacero_{(a}\kt_{b)}\big)  - \dfrac{1}{2}n^b \lie_W\gamma_{ab} - P^{cd}\U_{ca}\big(n^b\nablacero_b\kt_d - \mu\s_{d}\big)\nonumber\\
	&=2P^{cd}P^{bf}\U_{fc}\U_{ab} (\kt_d-\mu\ell_d)+ \dfrac{1}{2}\lie_n \left(\lie_n(\kt_a-\mu\ell_a)-2P^{cd}\U_{ac}(\kt_d-\mu\ell_d)\right)- \dfrac{1}{2}n^b \lie_W\gamma_{ab}  \nonumber \\
	&\quad\, - P^{cd}\U_{ca}\lie_n\big(\kt_d-\mu\ell_d\big).\label{In}
\end{align}
From Lemma \ref{lemagammalie} and the fact that $\lie_n\big(\kt_c-\mu\ell_c\big)-2P^{fd}\U_{cf}\big(\kt_d-\mu\ell_d\big)$ contracted with $n^c$ is identically zero, the term $n^b \lie_W \gamma_{ab}$ is
\begin{equation}
	\label{lienW}
	n^b \lie_W \gamma_{ab}=\big(\delta^c_a\lie_n -2P^{bc}\U_{ab}\big)\left(\lie_n\big(\kt_c-\mu\ell_c\big)-2P^{fd}\U_{cf}\big(\kt_d-\mu\ell_d\big)\right).
\end{equation}
This implies that all the terms in the RHS of \eqref{In} cancel each other and we arrive at $\mc{I}_{ab}n^b = 0$. Contracting the LHS of \eqref{identity} with $n^b$ and using Lemma \ref{lemagammalie} [Eq. \eqref{nliegamma}] applied to $X^a = P^{ab}\mS_{bc}n^c$ gives $\lie_n\big(\mS_{ab}n^b\big)-\nablacero_a\big(\mS(n,n)\big)+\big(\tr_P\bU\big) \mS_{ab}n^b$, so item 2. follows.\\


 Finally, to prove item 3. we note that now $\Sigma_{abc} = 0$ and $W^a = \dfrac{1}{2} n(\ktt-\mu\elltwo)n^a$, so $\lie_W\gamma_{ab} = n(\ktt- \mu\elltwo)\U_{ab}$. In addition, $\kt=\mu\bm\ell$ implies $\nablacero_a\kt_b \st{\eqref{nablaell}}{=} \mu\big(\F_{ab}-\elltwo\U_{ab}\big)$, so \eqref{Ifacil} becomes
 \begin{align*}
 \mc{I}_{ab} &=\left(\big(\ktt-\mu\elltwo\big)\tr_P\bU +\big(\ktt-\mu\elltwo\big)\kappa_n  + \dfrac{1}{2}n\big(\ktt-\mu\elltwo\big)\right)\U_{ab} +\lie_n\big(\big(\ktt-\mu\elltwo\big)\U_{ab} \big)\\
 &\quad\, -\dfrac{1}{2}n(\ktt- \mu\elltwo)\U_{ab}  -2\big(\ktt -\mu\elltwo\big)P^{cd}\U_{c(a}\U_{b)d}.
 \end{align*}
Taking the trace with respect to $P^{ab}$ and using \eqref{lietrPY} in Lemma \ref{lemalieP} with $T_{ab}=\mc{I}_{ab}$,
\begin{equation*}
	\begin{aligned}
		P^{ab}\mc{I}_{ab} &=  \left((\ktt-\mu\elltwo)\tr_P\bU +(\ktt-\mu\elltwo)\kappa_n+ n\big(\ktt-\mu\elltwo\big) \right)\tr_P\bU+\big(\ktt-\mu\elltwo\big) \lie_n(\tr_P\bU).
	\end{aligned}
\end{equation*}
Contracting the LHS of \eqref{identity} with $P^{ab}$ and using again identity \eqref{lietrPY} now with $T_{ab}=\mS_{ab}$ gives $2\lie_n\big(\tr_P\mS\big) +  2\big(\kappa_n+\tr_P\bU\big)\tr_P\mS$, which proves item 3. of the corollary.
\end{proof}
\end{cor}


We can summarize all the results of this section in the following theorem.
\begin{teo}
	\label{teorema_null}
Let $\{\mc H,\bg,\bm\ell,\elltwo,\bY\}$ be null hypersurface data $(\Phi,\xi)$-embedded in $(\mc M,g)$ with $R_{\alpha\beta}=\lambda g_{\alpha\beta}$ in a neighbourhood\footnote{As before, this condition is only needed up to first derivatives on $\Phi(\mc H)$.} of $\Phi(\mc H)$ and let $\eta$ be a vector field on $\mc M$ satisfying $\mc{Q}\st{\mc H}{=}0$, $\mS(n,n)=0$, $\k_{ab}n^b = 0$ and $\wh{\k}=0$. Assume that $\mc H$ admits a cross-section and let $\mu$ be a constant satisfying $\lambda\mu=0$. Then,
\begin{enumerate}
	\item If $P^{ab}\k_{ab} = (\mf n-1)\mu$ and $\lie_n\big(P^{ab}\k_{ab}\big)=0$ on a cross-section of $\mc H$, then $\mc{K}_{ab}=\mu\gamma_{ab}$ everywhere on $\mc H$. Moreover, if $\mf{K}(n)=\mu$ on a cross-section, then $\mf{K}(n)=\mu$ everywhere.
	\item If in addition to item 1. one has $\mf{K}_c=\mu\ell_c$ and $\mS_{bc}n^b=0$ on a cross-section of $\mc H$, then $\mf{K}_c=\mu\ell_c$ and $\mS_{bc}n^b=0$ everywhere on $\mc H$.
	\item Moreover, if in addition to items 1. and 2. one has $\mathbb{k}=\mu\elltwo$ and $\tr_P\mS=0$ on a cross-section of $\mc H$, then $\mathbb{k}=\mu\elltwo$ and $\tr_P\mS=0$ everywhere on $\mc H$.
\end{enumerate}
\begin{proof}
From the first equation in Corollary \ref{cor_junto}, $\tr_P\k$ satisfies a second order homogeneous transport equation along $n$. So, if $P^{ab}\k_{ab} = (\mf n-1)\mu$ and $\lie_n\big(P^{ab}\k_{ab}\big)=0$ on a cross-section of $\mc H$ it follows that $P^{ab}\k_{ab} = (\mf n-1)\mu$ everywhere on $\mc H$, and hence by decomposition \eqref{decomposition} we conclude $\k=\mu\bg$ everywhere. As a consequence of item 1. in Lemma \ref{lemaagrupado}, if $\mf{K}(n)=\mu$ on a cross-section, then $\mf{K}(n)=\mu$ everywhere. This proves item 1. of the theorem. The second item follows at once from item 2. in Lemma \ref{lemaagrupado} and item 2. in Corollary \ref{cor_junto}, because $\mf{K}_c-\mu\ell_c$ and $\mS_{bc}n^b$ satisfy a homogeneous system of transport equations and they vanish initially. Similarly, item 3. follows from the third item in Lemma \ref{lemaagrupado} and item 3. in Corollary \ref{cor_junto}, since $\mathbb{k}-\mu\elltwo$ and $\tr_P\mS$ also satisfy a homogeneous system and also vanish initially.
\end{proof}
\end{teo}

\section{Review of double null data}
\label{sec_DND}

In this section we review the main aspects of \cite{Mio1,Mio2} where we proved that the characteristic Cauchy problem for the Einstein equations can be formulated in a completely detached way via the notion of double null data. We start by recalling that given two null metric hypersurface data $\mc{D}=\{\mc H,\bg,\bm\ell,\elltwo\}$ and $\ul{\mc D}=\{\wh{\mc H},\wh{\bg},\wh{\bm\ell},\wh{\ell}{}^{(2)}\}$ with non-degenerate and isometric boundaries $\partial\mc{H}$ and $\partial\wh{\mc H}$ one can construct a so-called $\partial$-isometry $\Psi:T\mc H\oplus \mc{F}(\partial\mc H)\to T\mc{\ul H}\oplus\mc{F}(\partial\mc{\ul H})$ (see Definition 4.1 of \cite{Mio2}) that is uniquely defined from the assumed isometry $\phi:\partial\mc H\to \partial\mc{\ul H}$ and a nowhere vanishing function $\sigma\in \mc{F}^{\star}(\partial\ul{\mc H})$ by the condition $\sigma=\mc{\ul A}\big((\ul n,0),\Psi(n,0)\big)$.\\

The $\partial$-isometry $\Psi$ is a basic prerequisite for two null hypersurface data to be glueable along their boundaries, but on top of that there are several additional \textit{compatibility conditions} constructed from the data and the $\partial$-isometry that need to be fulfilled on the would-be intersection surface $\mc S$. Geometrically they imply that the torsion one-form, the two null second fundamental forms and the (pullback of the) ambient Ricci tensor on $\mc S$ are well-defined quantities. They key point is that these conditions can be written solely in terms of abstract data. Whenever these conditions hold, we say that the tuple $\{\mc D,\mc{\ul D},\sigma,\phi\}$ is double null data (see Definition 4.6 of \cite{Mio2}). \\

As proven in \cite{Mio2}, the compatibility conditions form a complete set in the sense that they are everything one needs to embed $\{\mc D,\mc{\ul D},\sigma,\phi\}$ in some ambient spacetime $(\mc M,g)$. As shown in Theorem 7.15 of \cite{Mio1}, in order for $(\mc M,g)$ to be $\lambda$-vacuum one also needs to assume the so-called \textit{constraint equations}, as we review next.
\begin{teo}
	\label{main}
	Let $\{\mc D,\mc{\ul D},\sigma,\phi\}$ be double null data satisfying the abstract constraint equations 
	\begin{equation}
		\label{constraintsL}
		\bm{\mc R} = \lambda\bg \quad \text{and}\quad \bm{\mc{\ul{R}}}=\lambda\ul\bg,
	\end{equation}
	where $\bm{\mc R}$ is the constraint tensor as defined in \eqref{constraint}, $\ul{\bm{\mc R}}$ is the corresponding tensor for the underlined data and $\lambda\in\real$. Then there exists a development $(\mc M,g)$ of $\{\mc D,\mc{\ul D},\sigma,\phi\}$ (possibly shrinking the data if necessary), solution of the $\lambda$-vacuum Einstein equations. Moreover, for any two such developments $(\mc M,g)$ and $(\mc{\wh M},\wh g)$, there exist neighbourhoods of $\mc H\cup\mc{\ul H}$, $\mc U\subseteq\mc M$ and $\wh{\mc U}\subseteq\mc{\wh M}$, and a diffeomorphism $\phi: \mc U\to \wh{\mc U}$ such that $\phi^{\star}\wh g=g$.
\end{teo}

We conclude this short review by recalling some notation from \cite{Mio1,Mio2}.
\begin{nota}
Let $\{\mc H,\bg,\bm\ell,\elltwo\}$ be null metric hypersurface data with non-degenerate boun\-dary $\partial\mc H$. Let $\{e_A\}$ be a basis on $\partial\mc H$ and define the metric $h$ by $h_{AB}\d \bg(e_A,e_B)$ (we denote its inverse by $h^{AB}$), as well as $\ell_A\d \bm\ell(e_A)$, $\ell^A\d h^{AB}\ell_B$ and $\ell^{(2)}_{\sharp}\d h^{AB}\ell_A\ell_B$. From relations \eqref{Pell}-\eqref{Pgamma} one can easily check that the tensor $P^{ab}$ on $\partial\mc H$ can be decomposed as \cite{Mio1} 
	\begin{equation}
		\label{decoP}
		P \st{\partial\mc H}{=} h^{AB}e_A \otimes e_B -2\ell^A e_A\otimes_s n -\big(\elltwo-\ell_{\sharp}^{(2)}\big)n\otimes_s n.
	\end{equation}
\end{nota}

As proven in \cite{Mio2}, given double null data $\{\mc D,\mc{\ul D},\sigma,\phi\}$ embedded in $(\mc M,g)$ with embeddings $\Phi$ and $\ul\Phi$, the vectors $\nu$ and $\ul\nu$ on $\mc S$ are related with $\ul\xi$ and $\xi$ by means of
\begin{equation}
	\label{nuandxi}
\nu \st{\mc S}{=} \sigma(\ul\xi + \ul\Phi_{\star}\ul t),\qquad \ul\nu \st{\mc S}{=} \sigma(\xi + \Phi_{\star}t),
\end{equation}
where the vectors $t$ and $\ul t$ are defined by
\begin{equation}
	\label{ts}
t\d -\frac{1}{2}(\elltwo-\ell_{\sharp}^{(2)}) n - \ell^Ae_A,\qquad 	\ul t\d \frac{1}{2}(\ul{\ell}^{(2)}-\ul{\ell}_{\sharp}^{(2)}) \ul n - \ul\ell^A \ul e_A.
\end{equation}
We denote by $\{\ul e_A\}$ the basis on $\puH$ induced by $\{e_A\}$ on $\pH$, i.e. $\ul{e}_A\d \phi_{\star} e_A$. Observe that $\Phi_{\star} e_A = \ul\Phi_{\star}\ul e_A$.

\section{Characteristic homothetic KID problem}
\label{sec_charact}

In this section we present the homothetic (including Killing) KID problem for two characte\-ristic hypersurfaces in the language of hypersurface data. In particular, we want to show that the initial data for the KID problem is at the same level as the characteristic data, meaning that all the extra restrictions one needs to impose can be written solely in terms of double null data (i.e. at the abstract level). The existence problem for Killing vectors in the characteristic case has already been addressed in \cite{chrusciel2013kids}. In that work, the authors assume the characteristic data to be embedded in a vacuum spacetime and find necessary and sufficient conditions that guarantee the existence of an ambient Killing vector. As already mentioned in Section \ref{sec_null}, these conditions are written in a particular coordinate system and involve the ambient Riemann tensor, so in principle it is not clear whether they are at the level of the initial data or not. In this work we generalize the KID equations to the homothetic case and we show that, both in the homothetic and Killing cases, they can be written in terms of abstract data at the hypersurfaces. \\

Let $\{\mc D,\mc{\ul D},\sigma,\phi\}$ be double null data satisfying the constraint equations $\R=\lambda\bg$ and $\ul{\bm{\mc R}}=\lambda\ul{\bg}$ and assume $\mc H$ and $\mc{\ul H}$ admit cross-sections. Consider two pairs $(\bar\eta,C)\in\X(\mc H)\times\mc{F}(\mc H)$ and $(\ul{\bar\eta},\ul C)\in\X(\mc{\ul H})\times\mc{F}(\ul{\mc H})$ (generally speaking we denote with an underline all quantities on $\ul{\mc D}$) which on the respective boundaries $\partial\mc H$ and $\partial\mc{\ul H}$ are related by $\Psi$, namely $\Psi\big((\bar\eta,C)\big) \st{\puH}{=} (\ul{\bar\eta},\ul C)$. By Theorem \ref{main} there exists a $\lambda$-vacuum spacetime $(\mc M,g)$ where $\{\mc D,\mc{\ul D},\sigma,\phi\}$ is embedded (with respective embeddings $\Phi$ and $\ul\Phi$), and the abstract condition $\Psi\big((\bar\eta,C)\big) \st{\puH}{=} (\ul{\bar\eta},\ul C)$ gua\-rantees that the vector field $\eta|_{\Phi(\mc H)} \d C\xi + \Phi_{\star}\bar\eta$ agrees with $\ul{\eta}|_{\ul\Phi(\ul{\mc H})} \d \ul C\ul\xi + \ul\Phi_{\star}\ul{\bar\eta}$ on $\mc S$. Then, from the well-posedness of the wave equation with data on two intersecting characteristic hypersurfaces (see e.g. \cite{Rendall,cabet}) one can construct a vector field $\eta$ on $(\mc M,g)$ by solving $\square\eta = -\lambda\eta$ (i.e. $\mc{Q}[\eta]=0$) with initial data $\eta|_{\Phi(\mc H)} = C\xi + \Phi_{\star}\bar\eta$ and $\eta|_{\ul\Phi(\mc{\ul H})} = \ul C\ul\xi + \ul\Phi_{\star}\ul{\bar\eta}$. The equation obtained from \eqref{equationk} after setting $\mc{Q}[\eta]=0$ is homogeneous, so $\eta$ will be a homothety in $(\mc M,g)$ (i.e. satisfies $\k[\eta]=\mu g$ on $\mc M$ for some $\mu\in \real$) if and only if one can guarantee that $\k[\eta]= \mu g$ along $\Phi(\mc H)\cup \ul\Phi(\mc{\ul H})$. Thus, one must look for sufficient conditions on $\{\mc D,\mc{\ul D},\sigma,\phi\}$ (i.e. at the abstract level) that a posteriori imply $\k|_{\Phi(\mc H)}=\mu g |_{\Phi(\mc H)}$ and $\k|_{\ul\Phi(\mc{\ul H})}=\mu g |_{\ul\Phi(\mc{\ul H})}$ once the data is embedded. \\

In Theorem \ref{teorema_null} we have found sufficient conditions (they are also necessary) that the tensors $\k_{ab}$, $\kt_a$, $\ktt$ and $\mS$ must satisfy on $\mc H\cup\mc{\ul H}$ and also on $\mc S$ in order to have $\k=\mu g$ at $\Phi(\mc H)\cup\ul\Phi(\mc{\ul H})$. On the full $\mc H$ the conditions are $$\mbox{1. } \mS(n,n)=0,\qquad \mbox{2. } \k_{ab}n^b=0, \qquad \mbox{3. } \wh{\k}=0,$$
and similarly on $\mc{\ul H}$, while on $\mc S$ the conditions are

\begin{minipage}{0.3\textwidth}
\noindent
\begin{enumerate}
	\item[4.] $P^{ab}\k_{ab} = (\mf n-1)\mu$,
	\item[5.] $\ul P^{ab}\k_{ab}=(\mf n-1)\mu$,
	\item[6.] $\lie_n\big(P^{ab}\k_{ab}\big)=0$,
	\item[7.] $\lie_{\ul n}\big(\ul{P}^{ab}\ul{\k}_{ab}\big)=0$,
\end{enumerate}
\end{minipage}
\begin{minipage}{0.25\textwidth}
\noindent
\begin{enumerate}
\item[8.] $\kt(n)=\mu$,
\item[9.] $\ul\kt(\ul n)=\mu$,
\item[10.] $ \kt= \mu\bm\ell$,
\item[11.] $\ul{\kt}=\mu\bm\ell$,
\end{enumerate}
\end{minipage}
\begin{minipage}{0.25\textwidth}
	\noindent
	\begin{enumerate}
\item[12.] $\mS_{ab}n^b = 0$,
\item[13.] $\ul{\mS}_{ab}\ul{n}^b=0$,
\item[14.] $\ktt=\mu\elltwo$,
\item[15.] $\ul{\ktt}=\mu\ul{\ell}{}^{(2)}$,
	\end{enumerate}
\end{minipage}
\begin{minipage}{0.25\textwidth}
	\noindent
	\begin{enumerate}
\item[16.] $\tr_P\mS=0$,
\item[17.] $\tr_{\ul P}\ul{\mS}=0$.
	\end{enumerate}
\end{minipage}
\vspace{0.2cm}

These conditions come from the analysis on $\Phi(\mc H)$ and $\ul\Phi(\ul{\mc H})$ as separate null hypersurfaces. However, they are glued to each other across their boundaries. This makes some of these conditions redundant. Our task now is to identify a complete subset of sufficient conditions and to prove that they can be written solely in terms of double null data. It is clear that conditions 1., 2., 3. and their underlined versions must necessarily be included in the set of sufficient conditions. However, as already mentioned many of the conditions on $\mc S$ are redundant because they are related to information coming from the other hypersurface. To get some intuition let us for the moment simplify the problem and assume a choice of riggings satisfying $\nu\st{\mc S} = \sigma \ul{\xi}$ and $\ul{\nu}\st{\mc S}{=} \sigma\xi$, which at the abstract level means $\elltwo\st{\pH}{=}0$, $\ell_{A}\st{\pH}{=}0$, $\ul{\ell}{}^{(2)}\st{\puH}{=}0$ and $\ul\ell_{A}\st{\puH}{=}0$. With this choice of gauge $P^{ab}\st{\mc S}{=} h^{AB}e_A^a e_B^b \st{\mc S}{=}\ul{P}^{ab}$ (see \eqref{decoP}), so it is clear that condition 4. implies condition 5. Similarly, condition 8. implies condition 9., because $\mu\st{\mc S}{=}\k(\xi,\nu)\st{\mc S}{=}\k(\sigma^{-1}\ul{\nu},\sigma\ul{\xi})\st{\mc S}{=}\k(\ul\nu,\ul\xi)$, and conditions 10., 11., 14. and 15. are automatically fulfilled provided 2. and 8. hold. We discuss e.g. 15. with similar arguments applying to the other three cases. Condition 2. implies $\k(\nu,\nu)=0$ and since $\k(\nu,\nu) \st{\mc S}{=} \k(\sigma^{-1}\ul\xi,\sigma^{-1}\ul\xi)$ one gets $\ul{\ktt}\st{\puH}{=}0$, and by our choice of gauge $\ul{\ell}{}^{(2)}\st{\puH}{=}0$, so 15. follows. From \eqref{S2} one can also guess that conditions 12. and 13. are not independent, so most likely one only needs to impose one of them. That 16. and 17. are also redundant is less intuitive but we shall prove this below in an arbitrary gauge. In summary, the suggestion is that the complete subset of necessary conditions are 1., 2. and 3. on $\mc H$, their underlined versions on $\mc{\ul H}$, and 4., 6., 7., 8. and 12. on $\mc S$. We now take on the task of recasting these conditions in terms of double null data in an arbitrary gauge and of proving that they are indeed sufficient to guarantee $\k=\mu g$ on $\Phi(\mc H)\cup\ul\Phi(\mc{\ul H})$. Let us start with the former. That 2., 3., 4., 6. and 7. are writeable in terms of double null data is clear from \eqref{Adown}. In the following lemma we show that 1., 8. and 12. can also be written in terms of the data.

\begin{lema}
Let $\{\mc D,\mc{\ul D},\sigma,\phi\}$ be double null data embedded in a $\lambda$-vacuum spacetime $(\mc M,g)$ and let $\eta\in\X(\mc M)$. Then,
\begin{align}
\mS(n,n) &\st{\H}{=} -C\left( \lambda +P^{bc}A_{bca}n^a + P\big(\br+\bs,\br+\bs\big) - \dfrac{1}{2}\kappa_nn(\elltwo) -2\lie_n^{(2)}\elltwo\right)-n^an^b \lie_{\bar\eta}\Y_{ab} \nonumber \\
&\quad\,  - \left(\elltwo n(C)+\big(\lie_{\bar\eta}\bm\ell\big)(n)\right)\kappa_n + n(\elltwo)n(C)   + \elltwo\lie_n^{(2)}C  +\lie_n\big(\big(\lie_{\bar\eta}\bm\ell\big)(n)\big)  ,\label{mSnn4}
\end{align}
with analogous expressions on $\mc{\ul H}$, and 
\begin{align}
\hspace{-0.45cm}\kt(n) &\st{\mc S}{=} \elltwo n(C) +\dfrac{1}{2}C n(\elltwo) -(\br-\bs)(\bar\eta) + \bm\ell\big(\lie_n\bar\eta\big) + \ul{\ell}^{(2)} \ul{n}(\ul C) +\dfrac{1}{2}\ul{C}\ul{n}(\ul\ell{}^{(2)}) - (\ul{\br}-\ul{\bs})(\bar{\ul\eta}) +\ul{\bm\ell}\big(\lie_{\ul n}\bar{\ul\eta}\big)\nonumber\\
\hspace{-0.45cm}	& \quad\, -t^a \big(\nablacero_aC + C(\r-\s)_a -\U_{ac}\bar\eta^c \big) +\sigma^{-1}\ul{t}^b\big( \ul{n}(\ul C)\ul\ell_b +\ul{C}(\ul\r+\ul\s)_b + \ul{\gamma}_{bc}\lie_{\ul n}\ul{\bar\eta}^c +\ul\U_{ab}\ul{\bar\eta}^a\big),\label{kn0}
\end{align}
\vspace{-0.7cm}
\begin{align}
\hspace{-0.45cm}\mS_{ab}n^b &\st{\mc S}{=} -C\left( \lambda\ell_a +P^{bc}A_{bca} + P^{bc}\big(\r_b+\s_b\big)\big(\Y_{ac}+\F_{ac} \big)  - \dfrac{1}{2}\kappa_n\nablacero_a\elltwo -2\nablacero_a\big(n(\elltwo)\big)+2P^{bc}\U_{ab}\nablacero_c\elltwo\right) \nonumber \\
\hspace{-0.45cm}&\quad\, -n^b \lie_{\bar\eta}\Y_{ab} + \left(\elltwo n(C)+\big(\lie_{\bar\eta}\bm\ell\big)(n)\right)\big(\r_a-\s_a\big) + \dfrac{1}{2}n(\elltwo)\nablacero_aC + \dfrac{1}{2}n(C)\nablacero_a\elltwo\nonumber \\
\hspace{-0.45cm}&\quad\,   + \elltwo\big(\nablacero_a\big(n(C)\big) - P^{bc}\U_{ab}\nablacero_cC  \big) +\dfrac{1}{2}\lie_n\lie_{\bar\eta}\ell_a  +\dfrac{1}{2} \nablacero_a \big((\lie_{\bar\eta}\bm\ell)(n)\big)  - P^{bc}\U_{ab}\lie_{\bar\eta}\ell_c.\label{mSn1}
\end{align}
\begin{proof}
Relations \eqref{mSnn4} and \eqref{mSn1} follow immediately from \eqref{mSnn2} and \eqref{mSn} after inserting $R_{\alpha\beta}=\lambda g_{\alpha \beta}$. To prove \eqref{kn0} we use the fact recalled at the end of Section \ref{sec_DND}, namely that $\nu \st{\mc S}{=} \sigma(\ul\xi + \ul\Phi_{\star}\ul t)$ and $\ul\nu \st{\mc S}{=} \sigma(\xi + \Phi_{\star}t)$ with $t$ and $\ul t$ given by \eqref{ts}, from where it follows
\begin{align}
\kt(n) &= \nu^{\alpha}\xi^{\beta}\nabla_{\alpha}\eta_{\beta} + \nu^{\alpha}\xi^{\beta}\nabla_{\beta}\eta_{\alpha}\nonumber\\
& \st{\mc S}{=} \nu^{\alpha}\xi^{\beta}\nabla_{\alpha}\eta_{\beta} + \ul\xi^{\alpha}\ul\nu^{\beta}\nabla_{\beta}\eta_{\alpha} -\nu^{\alpha}t^{\beta}\nabla_{\beta}\eta_{\alpha} +\sigma^{-1} \ul t^{\alpha}\ul\nu^{\beta}\nabla_{\beta}\eta_{\alpha}\nonumber\\
&\st{\mc S}{=} n^b\big({}^{(2)}\nabla\bm\eta\big)_b + \ul{n}^b\big({}^{(2)}\nabla\bm\eta\big)_b - t^a n^b \big(\nabla\bm\eta\big)_{ab} + \sigma^{-1}\ul{n}^a \ul{t}^b \big(\nabla\bm\eta\big)_{ab}  .\label{k(n)}
\end{align}
So we need to compute the tensors $\big(\nabla\bm\eta\big)_{ab}$ and $\big({}^{(2)}\nabla\bm\eta\big)_b$ and their contractions with $n$. For $\big(\nabla\bm\eta\big)_{ab}$ we apply equation \eqref{identity1} in Proposition \ref{proppullback} for $T=\bm\eta$ to get
\begin{align*}
	\big(\nabla\bm\eta\big)_{ab} &= \nablacero_a\bm\eta_b +\Y_{ab}\bm\eta_c n^c + \U_{ab}\bm\eta(\xi)\nonumber\\
	&\st{\eqref{etatangxi}}{=}\nablacero_a\big(C\ell_b + \gamma_{bc}\bar\eta^c\big) +C\Y_{ab} +\big(C\elltwo+\bm\ell(\bar\eta)\big) \U_{ab}\nonumber\\
	&= \ell_b \nablacero_a C +C\big(\Y_{ab}+\F_{ab}\big)-\ell_b\U_{ac}\bar\eta^c + \gamma_{bc}\nablacero_a\bar\eta^c,
\end{align*}
where in the last line we used \eqref{nablagamma}-\eqref{nablaell}. The contractions with $n^a$ and $n^b$ follow after using $n^a\nablacero_a\bar\eta^c = \lie_n\bar\eta^c + \bar\eta^a\nablacero_a n^c$ and \eqref{derivadan},
\begin{align}
	n^a \big(\nabla\bm\eta\big)_{ab} &= n(C)\ell_b +C(\r+\s)_b +\gamma_{bc} \lie_n\bar\eta^c + \U_{ab}\bar\eta^a,\label{nablaetan1}\\
	n^b \big(\nabla\bm\eta\big)_{ab} &= \nablacero_aC + C(\r-\s)_a -\U_{ac}\bar\eta^c .\label{nablaetan2}
\end{align}
To compute $\big({}^{(2)}\nabla\bm\eta\big)_b$ we use equation \eqref{identity3} in Proposition \ref{proppullback} for $T=\bm\eta$ and $j=1$ together with \eqref{etatangxi} to get
\begin{align*}
	\big({}^{(2)}\nabla\bm\eta\big)_b &= \nablacero_b \big(C\elltwo+\bm\ell(\bar\eta)\big)-\big(C\elltwo+\bm\ell(\bar\eta)\big)\big(\r-\s\big)_b -V^c{}_b\big(C\ell_c+\gamma_{cd}\bar\eta^d\big)\nonumber\\
	&=\elltwo \nablacero_b C + \dfrac{1}{2}C\nablacero_b\elltwo + \nablacero_b\big(\bm\ell(\bar\eta)\big) -\big(\Y_{bd}+\F_{bd}\big)\bar\eta^d,
\end{align*}
where in the second equality we used \eqref{ellV}-\eqref{gammaV}. Contracting this with $n^b$ and using $\lie_n\big(\bm\ell(\bar\eta)\big) \st{\eqref{lienell}}{=} 2\bs(\bar\eta) + \ell_a\lie_n\bar\eta^a$ we arrive at
\begin{equation}
	\label{2nablaetan}
	n^b \big({}^{(2)}\nabla\bm\eta\big)_b = \elltwo n(C) +\dfrac{1}{2}C n(\elltwo) -(\br-\bs)(\bar\eta) + \bm\ell\big(\lie_n\bar\eta\big).
\end{equation}
Inserting \eqref{nablaetan1}-\eqref{2nablaetan} into \eqref{k(n)} yields the result.
\end{proof}
\end{lema}

\begin{rmk}
	Note that in a gauge where $\elltwo\st{\pH}{=}0$, $\ell_{A}\st{\pH}{=}0$, $\ul{\ell}{}^{(2)}\st{\puH}{=}0$ and ${\ul\ell}_{A}\st{\puH}{=}0$ (and hence $\nu \st{\mc S}{=} \sigma\ul\xi $ and $\ul\nu \st{\mc S}{=} \sigma\xi$) the expression for $\kt(n)$ at $\mc S$ simplifies drastically. We keep the gauge arbitrary so that Theorem \ref{teo_char} below is written in a gauge-covariant form.
\end{rmk}

From this lemma together with expression \eqref{Adown} it follows that conditions 1., 2., 3. on $\mc H$ can be written in terms of the data as
\begin{align}
	0 &\st{\mc H}{=} -C\left( \lambda +P^{bc}A_{bca}n^a + P\big(\br+\bs,\br+\bs\big) - \dfrac{1}{2}\kappa_nn(\elltwo) -2\lie_n^{(2)}\elltwo\right)-n^an^b \lie_{\bar\eta}\Y_{ab} \nonumber \\
	&\quad\,  - \left(\elltwo n(C)+\big(\lie_{\bar\eta}\bm\ell\big)(n)\right)\kappa_n + n(\elltwo)n(C)   + \elltwo\lie_n^{(2)}C  +\lie_n\big(\big(\lie_{\bar\eta}\bm\ell\big)(n)\big)  ,\label{mSnn3}\\
	0&\st{\mc H}{=}  2C\r_a + \nablacero_a C+ n(C)\ell_a + n^b\lie_{\bar\eta}\gamma_{ab},\label{Kn=0}\\
	0& \st{\mc H}{=}  2C\Y_{ab}  + 2\ell_{(a}\nablacero_{b)}C +\lie_{\bar\eta}\gamma_{ab} - \dfrac{1}{\mf n-1}\left(2C\tr_P\bY - 2\elltwo n(C) + P^{ab}\lie_{\bar\eta}\gamma_{ab}\right)\gamma_{ab},\label{whK=0}
\end{align}
where in the last one we used $\wh{\k} = \k - \frac{\tr_p\k}{\mf n-1}\bg$ because $\k_{ab}n^b=0$. On $\uH$ the equations are analogous, while on $\pH$ and $\puH$ conditions 4., 6., 7., 8. and 12. can be written as 
\begin{align}
(\mf n-1)\mu& \st{\pH}{=}  2C\tr_P\bY - 2\elltwo n(C) + P^{ab}\lie_{\bar\eta}\gamma_{ab} ,\label{conds1}\\
0&\st{\pH}{=}\lie_n \left(C\tr_P\bY - \elltwo n(C) + \dfrac{1}{2}P^{ab}\lie_{\bar\eta}\gamma_{ab}\right) ,\\
0&\st{\puH}{=}\lie_{\ul n} \left(\ul C\tr_P\ul\bY - \ul\ell{}^{(2)} \ul n(\ul C) + \dfrac{1}{2}\ul P^{ab}\lie_{\bar{\ul\eta}}\ul\gamma_{ab}\right) ,\label{conds2}
\end{align}
\vspace{-0.8cm}
\begin{align}
\mu &\st{\pH}{=} \elltwo n(C) +\dfrac{1}{2}C n(\elltwo) -(\br-\bs)(\bar\eta) + \bm\ell\big(\lie_n\bar\eta\big) -t^a \big(\nablacero_aC + C(\r-\s)_a -\U_{ac}\bar\eta^c \big)\nonumber\\
	& \quad\, +\phi^{\star}\left(\ul{\ell}^{(2)} \ul{n}(\ul C) +\dfrac{1}{2}\ul{C}\ul{n}(\ul\ell{}^{(2)}) - (\ul{\br}-\ul{\bs})(\bar{\ul\eta}) +\ul{\bm\ell}\big(\lie_{\ul n}\bar{\ul\eta}\big)\right)\nonumber\\
	&\quad\, +\phi^{\star}\left(\sigma^{-1}\ul{t}^b\big( \ul{n}(\ul C)\ul\ell_b +\ul{C}(\ul\r+\ul\s)_b + \ul{\gamma}_{bc}\lie_{\ul n}\ul{\bar\eta}^c +\ul\U_{ab}\ul{\bar\eta}^a\big)\right),\label{kn}\\
0 &\st{\pH}{=} -C\left( \lambda\ell_a +P^{bc}A_{bca} + P^{bc}\big(\r_b+\s_b\big)\Pi_{ac}  - \dfrac{1}{2}\kappa_n\nablacero_a\elltwo -2\nablacero_a\big(n(\elltwo)\big)+2P^{bc}\U_{ab}\nablacero_c\elltwo\right) \nonumber \\
	&\quad\, -n^b \lie_{\bar\eta}\Y_{ab} + \left(\elltwo n(C)+\big(\lie_{\bar\eta}\bm\ell\big)(n)\right)\big(\r_a-\s_a\big) + \dfrac{1}{2}n(\elltwo)\nablacero_aC + \dfrac{1}{2}n(C)\nablacero_a\elltwo\nonumber \\
	&\quad\,   + \elltwo\big(\nablacero_a\big(n(C)\big) - P^{bc}\U_{ab}\nablacero_cC  \big) +\dfrac{1}{2}\lie_n\lie_{\bar\eta}\ell_a  +\dfrac{1}{2} \nablacero_a \big((\lie_{\bar\eta}\bm\ell)(n)\big)  - P^{bc}\U_{ab}\lie_{\bar\eta}\ell_c.\label{mSn2}
\end{align}

Next we prove that whenever the data is embedded, these conditions imply the rest on $\mc S$, and thus that $\k$ is pure trace on ${\Phi(\mc H)\cup\ul\Phi(\mc{\ul H})}$.

\begin{lema}
	Let $\{\mc D,\mc{\ul D},\sigma,\phi\}$ be double null data embedded in a $\lambda$-vacuum spacetime and satisfying \eqref{mSnn3}-\eqref{mSn2}. Then, conditions 5., 9., 10., 11., 13., 14., 15., 16. and 17. hold on $\mc S$, and as a consequence, $\k_{\alpha\beta}=\mu g_{\alpha\beta}$ on ${\Phi(\mc H)\cup\ul\Phi(\mc{\ul H})}$.
\begin{proof}
Since the data satisfies \eqref{mSnn3}-\eqref{mSn2} the embedded conditions 1., 2., 3., 4., 6., 7., 8. and 12. are true. Let us prove that they imply conditions 5., 9., 10., 11., 13., 14., 15., 16. and 17. on $\mc S$. We start with condition 5. From $P^{ab}\k_{ab}\st{\mc S}{=} (\mf n-1)\mu$, $\k_{ab}n^b=0$ and using the decomposition \eqref{decoP} it follows $h^{AB}\k_{AB}= (\mf n-1)\mu$ on $\mc S$, and hence condition $\ul P^{ab}\ul{\k}_{ab}=(\mf n-1)\mu$ is automatically fulfilled on $\mc S$. To prove 9. note that conditions 1., 2., 3., 4. and 6. imply, by item 1. in Theorem \ref{teorema_null}, that $\k_{ab}=\mu\gamma_{ab}$ on $\mc H$ (and analogously $\ul{\k}_{ab}=\mu\ul{\gamma}_{ab}$ on $\ul{\mc H}$). Then, using \eqref{nuandxi} $\kt(n)\st{\mc S}{=}\k(\nu,\xi)\st{\mc S}{=} \k(\ul\xi+\Phi_{\star}\ul t,\ul\nu)- \k(\nu,\Phi_{\star}t)\st{\mc S}{=} \k(\ul\xi,\ul\nu)\st{\mc S}{=}\ul{\kt}(\ul n)$. So item 1. in Theorem \ref{teorema_null} gives $\kt(n)=\mu$ (resp. $\ul{\kt}(\ul n)=\mu$) everywhere on $\mc H$ (resp. on $\ul{\mc H}$). Note also that $\kt_a=\mu\ell_a$ on $\mc S$ because $\kt(n)\st{\mc S}{=}\mu = \mu\bm\ell(n)$ and $$\kt(e_A) = \k(\xi,\Phi_{\star} e_A) \st{\mc S}{=} \k\big(\sigma^{-1}\ul\nu,\ul\Phi_{\star}\ul e_A \big) - \k\big(\Phi_{\star}t,\Phi_{\star}e_A \big) \st{\mc S}{=} \mu h_{BA}\ell^B \st{\mc S}{=}\mu\ell_A.$$ Similarly $\ul{\kt}_a=\mu\ul\ell_a$ on $\mc S$, which proves items 10. and 11. To show 13. note that from the symmetry \eqref{S2}, \eqref{nuandxi} and \eqref{2Kn} it follows 
\begin{align*}
	\ul{\mS}(\ul n,\ul e_A) =	\Sigma(\ul\xi,\ul\nu,\ul\Phi_{\star} \ul e_A) &\st{\mc S}{=} -\Sigma(\ul\nu,\ul\xi,\ul\Phi_{\star}\ul e_A) + \big(\nabla\k\big)(\ul\Phi_{\star}\ul e_A,\ul\xi,\ul\nu)\\
	&\st{\mc S}{=} - \Sigma(\xi,\nu,\Phi_{\star}e_A)  + \Sigma(\Phi_{\star}t,\nu,\Phi_{\star}e_A)+ \Sigma(\ul\nu,\ul\Phi_{\star}\ul t, \ul\Phi_{\star}\ul e_A) + \big({}^{(2)}\nabla\ul\k\big)_{ab} \ul e_A^a \ul n^b\\
	& \st{\mc S}{=} - \mS(n,e_A)  + \Sigma(t,n,e_A)+ \Sigma(\ul n,\ul t,\ul e_A) - \ul P^{bc}\ul\U_{ca}\ul\kt_b \ul e_A^a \st{\mc S}{=}0,
\end{align*}
where in the third equality we used $\ul{P}^{bc}\ul\U_{ca}\ul\kt_b  \st{\puH}{=} \mu\ul{P}^{bc}\ul\U_{ca}\ul\ell_b  \st{\puH}{=} 0$, $\mS(n,e_A)\st{\pH}{=}0$ (because of 12.) and that $\ul\Sigma(\ul n,\ul t,\ul e_A)\st{\puH}{=}0$ and $\Sigma(t,n,e_A)\st{\pH}{=}0$ (see Rmk. \ref{rmk1}). By item 2. of Theorem \ref{teorema_null} we have that $\kt=\mu\bm\ell$ and $\mS(n,\cdot)=0$ on $\H$, and $\ul\kt=\mu\ul{\bm\ell}$ and $\ul\mS(\ul n,\cdot)=0$ on $\uH$.\\

From $\k(\ul\nu,\ul\nu)\st{\mc S}{=} 0$ and \eqref{nuandxi} it follows $$0=\sigma^{-2}\k(\ul\nu,\ul\nu) \st{\mc S}{=} \k(\xi,\xi)+2\k(\xi,\Phi_{\star}t)+\k(\Phi_{\star}t,\Phi_{\star}t) \st{\mc S}{=}\ktt +2\kt(t)+\k(t,t),$$ which after inserting $\kt=\mu\bm\ell$, $\k_{ab}=\mu\gamma_{ab}$ and \eqref{ts} gives $$0\st{\pH}{=} \ktt + 2\mu\bm\ell(t) + \mu\bg(t,t)\st{\pH}{=}\ktt + \mu\big(-\big(\elltwo - \ell_{\sharp}^{(2)}\big) - 2\ell_{\sharp}^{(2)} + \ell_{\sharp}^{(2)}\big) \st{\pH}{=} \ktt - \mu\elltwo \qquad \Longrightarrow\qquad \ktt\st{\pH}{=}\mu\elltwo.$$ Similarly $\ul{\ktt}\st{\puH}{=}\mu\ul{\ell}{}^{(2)}$. Let us conclude by proving conditions 16. and 17. (we just prove 16. since 17. is analogous). Using again that $\kt=\mu\bm\ell$ and $\k_{ab}=\mu\gamma_{ab}$ it follows 
\begin{align*}
	\big({}^{(2)}\nabla\k\big)_{ab}& \st{\eqref{2K}}{=} \mu\big(\F_{ab} - \elltwo \U_{ab}\big) + \mu\Y_{ab} + \ktt\U_{ab} - \mu\big(\r-\s\big)_a\ell_b - \mu V^c{}_a\gamma_{bc}\nonumber\\
	&\hspace{0.0cm} \st{\eqref{gammaV}}{=} \mu\big(\F_{ab}+\Y_{ab}\big) + \big(\ktt-\mu\elltwo\big)\U_{ab} -\mu\big(\r-\s\big)_a\ell_b-\mu\big(\F_{ab}+\Y_{ab}\big) + \mu\big(\r-\s\big)_a\ell_b\nonumber\\
	&\hspace{0.1cm}= \big(\ktt-\mu\elltwo\big)\U_{ab},
\end{align*} 
which in particular implies $\big({}^{(2)}\nabla\k\big)_{ab}\st{\pH}{=} 0$. Now, from \eqref{S2}, \eqref{nuandxi} and using $\Sigma_{abc}\st{\pH}{=} 0$ and $\ul\Sigma_{abc}\st{\puH}{=} 0$ (see Remark \ref{rmk1}) we get
\begin{align*}
	\Sigma(\xi,\Phi_{\star}e_A,\Phi_{\star}e_B)&\st{\mc S}{=} - \Sigma(\Phi_{\star}e_A,\xi,\Phi_{\star}e_B) + \big(\nabla\k\big)(\Phi_{\star} e_B,\xi,\Phi_{\star}e_A) \\
	&\st{\mc S}{=} - \Sigma(\ul\Phi_{\star} \ul e_A,\sigma^{-1}\ul\nu,  \ul\Phi_{\star} \ul e_B) + \Sigma(\Phi_{\star}e_A,\Phi_{\star}t,\Phi_{\star}e_B) + \big({}^{(2)}\nabla\k\big)_{ab} e_A^a e_B^b\st{\mc S}{=} 0.
\end{align*}
This, combined with the decomposition \eqref{decoP} and $\mS(n,\cdot)\st{\pH}{=}0$ implies $\tr_P\mS\st{\pH}{=}0$, and hence by item 3. in Theorem \ref{teorema_null} it follows that $\ktt = \mu\elltwo$ and $\tr_P\mS=0$ everywhere on $\mc H$ (and similarly $\ul{\ktt}=\mu\ul{\ell}{}^{(2)}$ and $\tr_{\ul P}\ul\mS=0$ on $\ul{\mc H}$).
\end{proof}
\end{lema}

Thus, conditions \eqref{mSnn3}-\eqref{mSn2} are everything one needs to impose to guarantee $\k_{\alpha\beta}=\mu g_{\alpha\beta}$ on ${\Phi(\mc H)\cup\ul\Phi(\mc{\ul H})}$, and hence $\lie_{\eta}g =\mu g$ on $\mc M$. This concludes the proof of one of our main results of the paper.
\begin{teo}
	\label{teo_char}
	Let $\{\mc D,\mc{\ul D},\sigma,\phi\}$ be double null data satisfying the constraint equations \eqref{constraintsL}. Let $(\bar\eta,C)\in\X(\mc H)\times\mc{F}(\mc H)$ and $(\ul{\bar\eta},\ul{C})\in\X(\mc{\ul H})\times\mc{F}(\mc{\ul H})$ satisfy $\Psi(\bar\eta,C) \st{\puH}{=} (\ul{\bar\eta},\ul{C})$. Assume conditions \eqref{mSnn3}-\eqref{whK=0} on $\mc H$, the analogous ones in $\mc{\ul H}$, and conditions \eqref{conds1}-\eqref{mSn2} on $\pH$ with $\mu\in\real$ satisfying $\lambda\mu=0$. Then, there exists a $\lambda$-vacuum development $(\mc M,g)$ of $\{\mc{D},\ul{\mc D},\sigma,\phi\}$ and a vector field $\eta$ on $\mc M$ satisfying $\lie_{\eta}g = \mu g$ such that $\eta|_{\Phi(\mc H)} = C\xi + \bar\eta$ and $\eta|_{\Phi(\mc{\ul H})} = C\xi + \bar\eta$.
\end{teo}

\begin{rmk}
In our opinion this theorem is relevant for at least two reasons. Firstly because the homothetic KID equations are written solely in terms of the initial data for the characteristic problem (and hence detached from the spacetime). This allows one to decide whether an initial data set will give rise to an ambient spacetime admitting a homothetic/Killing vector before solving the Einstein equations. The second reason is that both the constraints and the KID equations are fully gauge and diffeomorphism covariant, which makes it possible to choose the gauge and the coordinates at will.
\end{rmk}

To the best of our knowledge the characteristic KID problem has not yet been analysed in the presence of matter fields. The generalization of Theorem \ref{teo_char} to more general matter models is not immediate and would require additional work. One possible strategy would be to compute the derivative of the matter field along the candidate to Killing vector, and use the matter equations to eliminate the transverse derivatives. This would lead to necessary conditions that could then be imposed at the initial data level and (hopefully) establish that they imply the existence of a Killing vector leaving the matter fields invariant.

\section{Smooth Spacelike-Null homothetic KID problem}
\label{sec_sp-ch}

In this section we illustrate another application of the identities found in the previous sections, namely the homothetic KID problem in the context of smooth spacelike-characteristic initial data. Given a smooth achronal hypersurface $\mc H$ with normal $\nu$ embedded in a Lorentzian manifold $(\mc M,g)$, the equation $\square_g F + L_1F=f$, where $L_1$ is a a first order differential operator with smooth coefficients and $f$ is also smooth, admits a unique solution on the domain of dependence of $\H$ in $\mc M$ with initial data given by $(F|_{\mc H},\nu(F)|_{\mc H})$ at the points where the hypersurface is spacelike and $F|_{\mc H}$ where it is null \cite{hormander1990remark}. As far as we are aware, such a result is not yet known for arbitrary tensor fields, but it is believed to be true. Hence, one can also expect that the Cauchy problem for the Einstein equations with initial data posed on a smooth spacelike-null hypersurface is also well-posed. Addressing these issues is beyond the scope of this paper, so we simply state the following conjectures.

\begin{con}
	\label{conj0}
Let $\mc H$ be a smooth achronal hypersurface embedded in a Lorentzian manifold $(\mc M,g)$ with normal $\nu$. Let $L_1$ be a first order differential operator with smooth coefficients and $f$, $T_0$ and $T_1$ smooth $(p,q)$ tensor fields. Then there exists a unique $(p,q)$ smooth tensor field $T$ satisfying $\square_g T + L_1 T = f$ on the domain of dependence of $\mc H$ in $\mc M$ with initial conditions $\big(T,\lie_{\nu}T\big)|_{\mc H}=(T_0,T_1)$ at the points where $\H$ is spacelike and $T|_{\mc H}=T_0$ where it is null.
\end{con}

\begin{con}
	\label{conj}
Let $\{\mc H,\bg,\bm\ell,\elltwo,\bY\}$ be hypersurface data satisfying $\ntwo\le 0$ such that the tensor $\mc A$ defined in \eqref{def_A} is Lorentzian. Assume the constraint equations \eqref{cons1}-\eqref{cons2} are fulfilled on $\mc H_S\d \{\ntwo < 0\}$ and $\R=\lambda\bg$ on $\mc H_0\d \{\ntwo = 0\}$. Then there exists a $\lambda$-vacuum spacetime $(\mc M,g)$, embedding $\Phi:\mc H\hookrightarrow\mc M$ and rigging $\xi$ such that $\{\mc H,\bg,\bm\ell,\elltwo,\bY\}$ (possibly shrinking the data in the null part if necessary) is $(\Phi,\xi)$-embedded.
\end{con}


Assuming these results to be true we can determine the necessary and sufficient conditions on $\{\mc H,\bg,\bm\ell,\elltwo,\bY\}$ for the existence of a homothetic vector field $\eta$ on $(\mc M,g)$. From the results of the previous sections it seems reasonable to impose $\k_{ab}=\mu\gamma_{ab}$ and $\big(\lie_{\xi}\k\big)_{ab} = 2\mu\Y_{ab}$ on $\mc H_S$ and also $\k_{ab}n^b = 0$, $\wh{\k} =0$ and $\mS(n,n)=0$ on $\mc H_0$. As we show next, the remaining conditions needed on $\pH_S$ for the argument to work are automatically fulfilled.

\begin{teo}
	\label{teo_smooth}
Let $\{\mc H,\bg,\bm\ell,\elltwo,\bY\}$ be hypersurface data with the properties listed in Conjecture \ref{conj} and assume further that the boundary of $\mc H_S$ is a smooth cross-section of $\mc H_0$. Suppose the constraint equations \eqref{cons1}-\eqref{cons2} are satisfied on $\mc H_S$ and $\R=\lambda\bg$ on $\mc H_0$. Assume that $(C,\bar\eta)\in\mc{F}(\mc H)\times\X(\mc H)$ satisfy \eqref{kid1} and \eqref{kid2} on $\mc H_S$ and \eqref{mSnn3}-\eqref{whK=0} on $\mc H_0$ with $\mu\in\real$ such that $\lambda\mu=0$. Then, assuming the validity of Conjectures \ref{conj0} and \ref{conj}, there exists a $\lambda$-vacuum spacetime $(\mc M,g)$, vector field $\eta\in\X(\mc M)$, embedding $\Phi:\mc{H}\hookrightarrow\mc M$ and rigging $\xi$ such that $\{\mc{H},\bg,\bm\ell,\elltwo,\bY\}$ (possible shrunk if necessary) is $(\Phi,\xi)$-embedded on $(\mc M,g)$, $\lie_{\eta}g=\mu g$ and $\eta|_{\Phi(\mc H)}=C\xi+\Phi_{\star}\bar\eta$. 
	\begin{proof}
By Conjecture \ref{conj} there exists a $\lambda$-vacuum spacetime $(\mc M,g)$, embedding $\Phi:\mc{H}\hookrightarrow\mc M$ and rigging $\xi$ such that $\{\mc{H},\bg,\bm\ell,\elltwo,\bY\}$ is $(\Phi,\xi)$-embedded on $(\mc M,g)$. Let us extend $\xi$ to a neighbourhood of $\Phi(\mc H)$ by $\nabla_{\xi}\xi=0$ and construct a vector field $\eta\in\X(\mc M)$ by solving $\square_g\eta = -\lambda\eta$ (i.e. $\mc{Q}[\eta]=0$, see \eqref{defQ}) with initial data $\eta = C\xi + \Phi_{\star}\bar\eta$ on $\Phi(\H)$ and \eqref{condini} on $\Phi(\mc H_S)$. Under this condition, \eqref{equationk} becomes a homogeneous PDE. By the assumed validity of Conjecture \ref{conj0} we only need to show that $\k_{\alpha\beta}=\mu g_{\alpha\beta}$ on $\Phi(\mc H)$ and $\lie_{\xi}\k_{\alpha\beta} = \mu\lie_{\xi}g_{\alpha\beta}$ on $\Phi\big(\mc H_S\big)$. By the same argument as the one employed in Theorem \ref{teo_espacial} it follows $\k_{\alpha\beta}=\mu g_{\alpha\beta}$ and $\lie_{\xi}\k_{\alpha\beta} = \mu\lie_{\xi}g_{\alpha\beta}$ on $\Phi\big(\mc H_S\big)$. So, we only need to show that $\k_{\alpha\beta}=\mu g_{\alpha\beta}$ on the null region. If $\H_0$ has empty interior it follows by continuity. Otherwise, by Theorem \ref{teorema_null} it suffices to prove that (i) $P^{ab}\k_{ab}=(\mf n-1)\mu$, (ii) $\lie_n\big(P^{ab}\k_{ab}\big)=0$, (iii) $\kt_a=\mu\ell_a$, (iv) $\ktt=\mu\elltwo$, (v) $\mS(n,\cdot)=0$ and (vi) $\tr_P\mS=0$ on $\partial\mc H_S$. Since $\k_{ab}=\mu\gamma_{ab}$, $\kt_a=\mu\ell_a$ and $\ktt=\mu\elltwo$ on $\mc H_S$ and the hypersurface data is smooth it is in particular continuous and hence the three conditions also hold on the boundary $\partial\mc H_S$. So items (i), (iii) and (iv) are automatically fulfilled. Item (ii) follows after noting that again by continuity the vector $n$ cannot vanish in a neighbourhood of $\partial\mc H_S$, and since $$\lie_n\big(P^{ab}\k_{ab}\big)\st{\H_S}{=} \mu \lie_n\big(P^{ab}\gamma_{ab}\big)\st{\H_S}{=} \mu \lie_n \big(\ntwo\elltwo\big) \st{\H_S}{=} \mu \ntwo n(\elltwo)+\mu\elltwo n(\ntwo),$$ it follows that $\lie_n\big(P^{ab}\k_{ab}\big)\st{\pH_S}{=}0$, since $\ntwo\st{\pH_S}{=}0$ and $n(\ntwo)\st{\pH_S}{=} 0$ (because $\ntwo$ is identically zero on $\H_0$). To prove items (v) and (vi) observe that since $\k_{\alpha\beta}=\mu g_{\alpha\beta}$ and $\lie_{\xi}\k_{\alpha\beta} = \mu\lie_{\xi}g_{\alpha\beta}$ on $\Phi\big(\mc H_S\big)$ then the full tensor $\Sigma$ vanishes on $\Phi(\mc H_S)$, so by continuity $\Sigma$ also vanishes on $\Phi(\mc S)$, which in particular implies $\mS(n,\cdot)=0$ and $\tr_P\mS=0$ on $\partial\mc H_S$. Applying Theorem \ref{teorema_null} it follows $\k_{\alpha\beta}=\mu g_{\alpha\beta}$ on $\Phi(\mc H_0)$, which proves the theorem.
	\end{proof}
\end{teo}

\section{Conclusions and future work}

In this paper, we have developed a set of general identities that relate the deformation tensor $\k\d\lie_{\eta}g$ of an arbitrary vector field $\eta$ to the tensor $\Sigma\d\lie_{\eta}\nabla$ on an abstract hypersurface $\mc H$ of arbitrary causal character. These identities have allowed us to derive explicit necessary conditions that an ambient vector field must satisfy on $\mc H$ in order to be a homothety. In the spacelike case, we recovered the well-known homothetic KID equations \cite{coll1977evolution,moncrief1975spacetime,beig1997killing,garcia2019conformal} within the hypersurface data formalism. For null hypersurfaces, we identified two distinct types of conditions: those that must hold on the entire hypersurface and those that only need to be sa\-tisfied on a cross-section of $\mc H$. We formulated both sets of conditions in a fully detached manner.\\

Applying this framework to the characteristic Cauchy problem, we extended the results in \cite{chrusciel2013kids} by expressing the KID equations in a fully gauge-covariant form, placing them at the same level as the initial data for the characteristic problem (i.e. detached from the spacetime one wishes to construct). This puts the characteristic KID problem on equal footing as the spacelike one. Furthermore, we showed the versatility of the formalism by applying it to the smooth spacelike-characteristic setting. \\

A natural continuation of this work would be generalizing the results to conformal Killing vectors, i.e. vectors satisfying $\lie_{\eta}g = 2\psi g$. In the spacelike case this has been studied in \cite{garcia2019conformal}, where the authors show that two extra equations (coming from the fact that $\psi$ is no longer determined a priori) need to be imposed at the initial data level. In the characteristic case the only result we know is \cite{paetz2014kids}, where the problem of existence of Killing vectors in an \mbox{asymptotically} flat vacuum spacetime is analysed in the conformally compactified spacetime, including null infinity. Dealing with the case of general conformal Killing characteristic data with our detached framework is considerably more involved, but should be doable exploiting some of the ideas in \cite{paetz2014kids}.\\

The identities established in this work have several notable implications. Firstly, they are formulated in terms of abstract hypersurface data, making them applicable in a broad range of contexts. Additionally, they are fully gauge and diffeomorphism covariant and are valid for hypersurfaces of any causal character. Beyond their immediate applications to the problems studied here, these identities provide a systematic approach for analysing homothetic Killing initial data in other well-posed Cauchy problems, showing that the conditions for the existence of a homothetic Killing field can be imposed at the level of the initial data. A particularly straightforward application is the spacelike-characteristic problem with corners \cite{chrusciel2015characteristic,czimek2022spacelike}, where one can check that the required conditions reduce to \eqref{kid1} and \eqref{kid2} in the spacelike region and \eqref{mSnn3}-\eqref{whK=0} in the null region, without any additional constraints on the intersection surface. This result reinforces the conclusion that the homothetic KID problem is fully characterized at the level of abstract initial data. Another physically interesting application of our results is to determine the Killing algebra of a stationary non-degenerate black hole spacetime in terms of the data at the bifurcation surface. The idea is that equations \eqref{mSnn3}-\eqref{whK=0} can be used to obtain transport equations for $C$ and $\bar\eta$ along $n$, which can then be solved from initial values at the bifurcation surface. Since the whole system of KID equations is overdetermined, the initial data cannot be given freely. Studying the compatibility conditions should give a handle on the determination of all Killing vectors depending on the free black hole data on the bifurcation surface.\\

Moreover, the identities are valid for any vector $\eta$ (not necessarily a homothety) and hold independently of any specific field equations, making them applicable to a wider class of initial value problems beyond the vacuum Einstein equations $\ric=\lambda g$. A particularly promising avenue for future research is the study of the asymptotic characteristic (homothetic) KID conditions within the conformal Einstein field equations. Another interesting application of the formalism is the construction of candidates to homotheties from data at a cross-section of $\mc H$. We intend to explore these two topics in subsequent works.

\section*{Acknowledgements}
This work has been supported by Projects PID2021-122938NB-I00 (Spanish Ministerio de Ciencia e Innovación and FEDER ``A way of making Europe''). M. Mars acknowledges financial support under projects SA097P24 (JCyL) and RED2022-134301-T funded by MCIN/AEI/10.13039/ 501100011033. G. Sánchez-Pérez also acknowledges support of the PhD. grant FPU20/03751 from Spanish Ministerio de Universidades.

\begin{appendices}
	
	\section{Useful identities}
	\label{app}
In the following proposition we extend a result from \cite{Mio3} valid for null hypersurface data to the context of general hypersurface data. We follow the same notation as in Appendix A of \cite{Mio3}. In particular, given a $(0,p)$ tensor field $T_{\alpha_1\cdots\alpha_p}$ on $\mc M$ we use the standard notation $T_{a_1\cdots a_p}$ to denote the pullback of $T$ to $\mc H$. Moreover, we introduce the notation ${}^{(i)}T_{\alpha_1\cdots \alpha_{p-1}}$ for $\xi^{\mu}T_{\alpha_1\cdots \alpha_{i-1}\mu\alpha_i\cdots \alpha_{p-1}}$ and ${}^{(i)}T_{a_1\cdots a_{p-1}}$ for the pullback of ${}^{(i)}T_{\alpha_1\cdots \alpha_{p-1}}$ to $\mc H$. In addition, we use ${}^{(i,j)}T_{a_1\cdots a_{p-2}}$ for the pullback to $\mc H$ of the tensor obtained by first the contraction of $T$ with $\xi$ in the j-th slot and then in the i-th slot of the resulting $(0,p-1)$ tensor, i.e. ${}^{(i,j)}T={}^{(i)}\big({}^{(j)}T\big)$. This notation requires care with the order of indices. For instance, note that ${}^{(1,1)}T={}^{(1,2)}T$ and ${}^{(1,3)}T={}^{(2,1)}T\neq {}^{(1,2)}T$.
	\begin{prop}
		\label{proppullback}
		Let $(\mc M,g)$ be a semi-Riemannian manifold and $\Phi:\mc H\hookrightarrow\mc M$ a smooth embedded hypersurface with rigging $\xi$. Let $T$ be a $(0,p)$-tensor on $\mc M$ and $f\in\mc{F}(\mc M)$. Then, for any $j\ge 1$
		\begin{align}
\hspace{-0.4cm}	\left(\nabla T\right)_{b a_1\cdots a_p} &= \nablacero_b T_{a_1\cdots a_p} + \sum_{i=1}^p \Y_{ba_i}T_{a_1\cdots a_{i-1} c a_{i+1}\cdots a_p} n^c \nonumber\\
&\quad\, + \sum_{i=1}^p \big(\U_{ba_i}+\ntwo \Y_{b a_i}\big)\big({}^{(i)} T\big)_{a_1\cdots a_{i-1} a_{i+1}\cdots  a_p},\hfill\label{identity1}\\
\hspace{-0.4cm}			\left({}^{(1)}\nabla T\right)_{a_1\cdots a_p} &= (\lie_{\xi}T)_{a_1\cdots a_p} - \sum_{i=1}^p \left(\r_{a_i}-\s_{a_i}+\dfrac{1}{2}\ntwo\nablacero_{a_i}\elltwo\right)\big({}^{(i)} T\big)_{a_1\cdots a_{i-1} a_{i+1}\cdots a_p}\nonumber\\
\hspace{-0.4cm}			& - \sum_{i=1}^p V^c{}_{a_i} T_{a_1\cdots a_{i-1} c a_{i+1}\cdots a_p},\hfill\label{identity2}
\end{align}
\begin{align}
	\left({}^{(j+1)}\nabla T\right)_{ba_1\cdots a_{p-1}} &= \nablacero_{b} \big({}^{(j)} T\big)_{a_1\cdots a_{p-1}} + \sum_{i=1}^{p-1} \Y_{b a_i} \big({}^{(j)} T\big)_{a_1\cdots a_{i-1} c a_{i+1}\cdots a_{p-1}}n^c\hfill \nonumber\\
		&+ \sum_{i=1}^{p-1} \big(\U_{b a_i}+\ntwo\Y_{b a_i}\big) \big({}^{(i,j)}T\big)_{a_1\cdots a_{i-1} a_{i+1}\cdots a_{p-1}} \nonumber\\
		&  - \left(\r_{b}-\s_{b}+\dfrac{1}{2}\ntwo\nablacero_{b}\elltwo\right) \big({}^{(j)} T\big)_{a_1\cdots a_{p-1}} - V^c{}_{b}T_{a_1\cdots a_{j-1}c a_{j}\cdots a_{p-1}},\hfill\label{identity3}\\
	\left(\hess f\right)_{ab} &= \nablacero_a\nablacero_b f + n(f) \Y_{ab} + \xi(f) \big(\U_{ab}+\ntwo\Y_{ab}\big),\label{hess}
		\end{align}
		where $V^{b}{}_{a} \d P^{bc}\big(\Y_{a c}+\F_{a c}\big) + \dfrac{1}{2}(d\elltwo)_{a}n^b$.
		\begin{proof}
			The proof of \eqref{identity1}-\eqref{identity3} is analogous to the one presented in \cite{Mio3} for the null case, with the only difference that this time $\ntwo\neq 0$. In order to prove \eqref{hess} we contract the (ambient) Hessian $\nabla_{\alpha}\nabla_{\beta}f$ with $e_a^{\alpha} e_b^{\beta}$ and use \eqref{connections},
			\begin{align*}
				e_a^{\alpha} e_b^{\beta}\nabla_{\alpha}\nabla_{\beta}f & = \nabla_{e_a}\nabla_{e_b} f -\big(\nabla_{e_a} e_b \big)(f)\\
				&= \nablacero_{e_a}\nablacero_{e_b} f - \big(\nablacero_{e_a} e_b\big) (f)  + n(f)\Y_{ab}+ \xi(f)\big(\U_{ab}+\ntwo\Y_{ab}\big)\\
				&= \nablacero_a\nablacero_b f + n(f) \Y_{ab} + \xi(f) \big(\U_{ab}+\ntwo\Y_{ab}\big).
			\end{align*}
		\end{proof}
	\end{prop}
	In addition to the previous proposition we shall also need to relate contractions of $\nabla T$ with two or more $\xi$ with either Lie derivatives or with pullbacked quantities in $\mc H$.
	\begin{prop}
		\label{proppullback3}
		Let $(\mc M,g)$ be a semi-Riemannian manifold, $\Phi:\mc H\hookrightarrow\mc M$ a smooth embedded hypersurface with rigging $\xi$ and let $T$ be a $(0,p)$-tensor on $\mc M$. Then,
\begin{equation}
	\label{a2id1}
\begin{aligned}
	e_b^{\mu}\xi^{\alpha_1}\cdots \xi^{\alpha_p}\nabla_{\mu}T_{\alpha_1\cdots \alpha_p} &\st{\mc H}{=}\nablacero_b\big(T(\xi,\cdots,\xi)\big) - p T(\xi,\cdots,\xi) \left(\r_b-\s_b+\frac{1}{2}\ntwo \nablacero_b\elltwo\right)\\
&\quad\, -\sum_{i=1}^p V^{c}{}_b \big({}^{(1,...,i-1,i+1,...p)}T\big)_{c}, 
	\end{aligned}
\end{equation}
\begin{equation}
	\label{a2id2}
		\xi^{\mu}\xi^{\alpha_1}\cdots \xi^{\alpha_p}\nabla_{\mu}T_{\alpha_1\cdots \alpha_p} \st{\mc H}{=} \big(\lie_{\xi}T\big)(\xi,\cdots,\xi) - p\beta T(\xi,\cdots,\xi) - \sum_{i=1}^p a_{\para}^c \big({}^{(1,...,i-1,i+1,...,p)}T\big)_{c},
\end{equation}
\begin{equation}
	\label{a2id3}
	\begin{aligned}
\hspace{-0.4cm}	\xi^{\mu}\xi^{\nu} e^{\alpha_1}_{a_1}\cdots e^{\alpha_{p-1}}_{a_{p-1}}\nabla_{\mu}T_{\nu\alpha_1\cdots \alpha_{p-1}} &\st{\mc H}{=} \big(\lie_{\xi}\big({}^{(1)}T\big)\big)_{a_1\cdots a_{p-1}}- \beta \, \big({}^{(1)}T\big)_{a_1\cdots a_{p-1}} - a_{\para}^b T_{ba_1\cdots  a_{p-1}}\\
&\quad\, - \sum_{i=1}^{p-1}V^b{}_{a_i} \big({}^{(1)}T\big)_{a_1\cdots a_{i-1} b a_{i+1}\cdots a_{p-1}}\\
&\quad\, -\sum_{i=1}^{p-1}\left(\r_{a_i}-\s_{a_i}+\dfrac{1}{2}\ntwo\nablacero_{a_i}\elltwo\right)\big({}^{(i,1)}T\big)_{a_1\cdots a_{i-1} a_{i+1}\cdots  a_{p-1}},
	\end{aligned}
\end{equation}
where $\nabla_{\xi}\xi \st{\mc H}{=}\beta\xi + \Phi_{\star} a_{\para}$.
		\begin{proof}
The first identity follows at once from the expression of $\nabla_{e_c}\xi$ in \eqref{nablaxi}, and the second and third identities are a consequence of the relation between $\lie_{\xi}$ and $\nabla$, namely $$\xi^{\mu}\nabla_{\mu} T_{\alpha_1\cdots\alpha_p} = \lie_{\xi}T_{\alpha_1\cdots\alpha_p} - \sum_{i=1}^p T_{\alpha_1\cdots\alpha_{i-1}\mu\alpha_{i+1}\cdots\alpha_p}\nabla_{\alpha_i}\xi^{\mu}.$$
		\end{proof}
	\end{prop}

We conclude the appendix by computing the pullback of the Lie derivative of an arbitrary covariant tensor field along any vector.

\begin{prop}
	\label{propliezetaT}
Let $(\mc M,g)$ be a semi-Riemannian manifold, $\Phi:\mc H\hookrightarrow\mc M$ a smooth embedded hypersurface with rigging $\xi$, $\zeta\in\X(\mc M)$ and $T$ a $(0,p)$ tensor field on $\mc M$. Define $\zeta_{t}\in\mc F(\mc H)$ and $\zeta_{\para}\in\X(\mc H)$ by $\zeta\st{\mc H}{=}\zeta_t\xi + \Phi_{\star}\zeta_{\para}$. Then,
\begin{equation}
	\label{liezetaT}
\Phi^{\star}\big(\lie_{\zeta}T\big)_{a_1\cdots a_p} = \zeta_t\big(\lie_{\xi}T\big)_{a_1\cdots a_p} + \sum_{i=1}^p\big( \nablacero_{a_i}\zeta_t \big)\big({}^{(i)}T\big)_{a_1\cdots a_{i-1} a_{i+1}\cdots a_p} + \big(\lie_{\zeta_{\para}}T\big)_{a_1\cdots a_p}.
\end{equation}
\begin{proof}
Inserting $\zeta\st{\mc H}{=}\zeta_t\xi + \Phi_{\star}\zeta_{\para}$, $$\lie_{\zeta}T_{\alpha_1\cdots \alpha_p} \st{\mc H}{=} \lie_{\zeta_t\xi + \Phi_{\star}\zeta_{\para}}T_{\alpha_1\cdots \alpha_p}  \st{\mc H}{=} \zeta_t\big(\lie_{\xi}T\big)_{\alpha_1\cdots \alpha_p} + \sum_{i=1}^p \big(\nabla_{\alpha_i}\zeta_t\big) T_{\alpha_1\cdots \alpha_{i-1}\mu\alpha_{i+1}\cdots \alpha_p}\xi^{\mu} + \big(\lie_{\Phi_{\star}\zeta_{\para}}T\big)_{\alpha_1\cdots \alpha_p}.$$ Taking the pullback to $\mc H$ and using the property $\Phi^{\star}\lie_{\Phi_{\star}\zeta_{\para}}T = \lie_{\zeta_{\para}}\Phi^{\star}T$ the result follows.
\end{proof}
\end{prop}

\section{Proof of Proposition \ref{prop_lienSigma}}
	\label{app_proof}
	
Before proving Proposition \ref{prop_lienSigma} we start with two preliminary results.

	\begin{lema}
		\label{lema_vector}
		Let $\{\mc H,\bg,\bm\ell,\elltwo\}$ be null metric hypersurface data $(\Phi,\xi)$-embedded in $(\mc M,g)$ and $Z\in\X(\mc M)$. Then, $$Z^{\alpha}\st{\Phi(\mc H)}{=}\xi^{\alpha} n^a Z_a+ P^{ab}e_a^{\alpha}Z_b +\nu^{\alpha}\bm{Z}(\xi) ,$$ where $\bm Z=g(Z,\cdot)$ and $Z_a \d \Phi^{\star}(\bm Z)_a$.
		\begin{proof}
			Using \eqref{inversemetric} with $\ntwo =0$ it follows $$Z^{\alpha} = g^{\alpha\beta}Z_{\beta} = \left(P^{ab}e_a^{\alpha}e_b^{\beta} + \xi^{\alpha}\nu^{\beta} + \xi^{\beta}\nu^{\alpha}\right)Z_{\beta} = P^{ab}e_a^{\alpha}Z_b + \xi^{\alpha} n^a Z_a+\nu^{\alpha} \bm{Z}(\xi) .$$
		\end{proof}
	\end{lema}

	\begin{lema}
		\label{lemainter}
		Let $\{\mc H,\bg,\bm\ell,\elltwo,\bY,\bZ^{(2)}\}$ be extended null metric hypersurface data $(\Phi,\xi)$-embedded in $(\mc M,g)$, $\eta\in\X(\mc M)$ and $\k\d\lie_{\eta}g$. Assume $\xi$ and $\nu$ are extended off $\Phi(\mc H)$ by $\nabla_{\xi}\xi=0$ and $\lie_{\xi}\nu=0$. Then,
		\begin{align}
			\big(\lie_{k^{(\nu)}}\lie_{\xi}g\big)_{ab} & = 2\k(n,n)\Z^{(2)}_{ab}+  \nablacero_{(a}\elltwo \nablacero_{b)}\big(\k(n,n)\big)+ 2\lie_X\Y_{ab}\nonumber\\
			&\quad\,+2\kt(n)\lie_n\Y_{ab} + 4\r_{(a}\nablacero_{b)}(\kt(n)),\label{lie1}\\
			\big(\lie_{[\xi,k^{(\nu)}]}g\big)_{ab} &= 2\big(\chi+2\lie_n(\kt(n))-2\mS(n,n)\big) \Y_{ab}+ 2\ell_{(a} \nablacero_{b)}\big(\chi+2\lie_n(\kt(n))-2\mS(n,n)\big)\nonumber\\
			&\quad\,+\lie_W\gamma_{ab} - 2\lie_{\Xi}\gamma_{ab} +2\Sigma(\nu,\xi,\xi)\U_{ab}, \label{lie2}
		\end{align}
		where $\Xi$, $X$, $\chi$ and $W$ are as in Proposition \ref{prop_lienSigma}.\\
		
\textbf{Remark:} The two terms involving Lie derivative of $\gamma_{ab}$ in \eqref{lie2} could be grouped together, but our interest is to keep the terms depending on $\mS_{ab}$ separated from the rest.
		\begin{proof}
To show \eqref{lie1} we use Lemma \ref{lema_vector} with $Z=k^{(\nu)}$ together with $g\big(k^{(\nu)},e_a\big)=\k_{ab}n^b$ and $g\big(k^{(\nu)},\xi\big)=\kt(n)$ to write 
			\begin{equation}
				\label{auxMarc}
				k^{(\nu)} = \k(n,n)\xi+\left(P^{ab}\k_{bc}n^c + \kt(n)n^a\right)e_a= \k(n,n)\xi + X + \kt(n) n,
			\end{equation}
			and hence by applying Proposition \ref{propliezetaT} to $T=\lie_{\xi}g$ with $\zeta_t=\k(n,n)$ and $\zeta_{\para}=X+\kt(n)n$ and using \eqref{Z2embedded}, \eqref{recall} and \eqref{Yembedded} in the respective terms,
			\begin{align*}
				e_a^{\alpha}e_b^{\beta}\lie_{k^{(\nu)}} \lie_{\xi} g_{\alpha\beta} &=2\k(n,n)\Z^{(2)}_{ab}+ \nablacero_{(a}\elltwo \nablacero_{b)}\big(\k(n,n)\big)+ 2\lie_X\Y_{ab} + 2\lie_{\kt(n)n}\Y_{ab}\\
				&=2\k(n,n)\Z^{(2)}_{ab}+ \nablacero_{(a}\elltwo \nablacero_{b)}\big(\k(n,n)\big)+ 2\lie_X\Y_{ab} +2\kt(n)\lie_n\Y_{ab} + 4\r_{(a}\nablacero_{b)}(\kt(n)).
			\end{align*}
			In order to prove \eqref{lie2} we need the commutator $[\xi,k^{(\nu)}]$. It is easier to compute the one-form $[\xi,k^{(\nu)}]_{\alpha}$ and then use Lemma \ref{lema_vector},
			\begin{align}
				[\xi,k^{(\nu)}]_{\alpha} &= \xi^{\mu}\nabla_{\mu}\big(\k_{\alpha\beta}\nu^{\beta}\big)- \k^{\mu}{}_{\beta}\nu^{\beta} \nabla_{\mu}\xi_{\alpha} \nonumber\\
				& = \xi^{\mu}\nu^{\beta}\nabla_{\mu}\k_{\alpha\beta} + \k_{\alpha\beta} \xi^{\mu}\nabla_{\mu}\nu^{\beta} - \k^{\mu}{}_{\beta}\nu^{\beta} \nabla_{\mu}\xi_{\alpha}\nonumber\\
				&= \xi^{\mu}\nu^{\beta}\nabla_{\mu}\k_{\alpha\beta} + \k_{\alpha\beta} \nu^{\mu}\nabla_{\mu}\xi^{\beta} - \big(X^b+\kt(n)n^b\big) e_b^{\mu} \nabla_{\mu}\xi_{\alpha},\label{commudown}
			\end{align}
			where in the third line we used $\lie_{\xi}\nu=0$, inserted the expression for $\k^{\mu}{}_{\beta}\nu^{\beta} = (k^{(\nu)})^{\mu}$ in \eqref{auxMarc} and used $\nabla_{\xi}\xi=0$. Contracting \eqref{commudown} with $\xi^{\alpha}$ and using the symmetries \eqref{S2}-\eqref{S3},
			\begin{align*}
				[\xi,k^{(\nu)}]_{\alpha}\xi^{\alpha} &= \Sigma(\nu,\xi,\xi) + \Sigma(\xi,\xi,\nu) +\k_{\alpha\beta}\xi^{\alpha}\nu^{\mu}\nabla_{\mu}\xi^{\beta} - \big(X^b+\kt(n)n^b\big) e_b^{\mu} \xi^{\alpha}\nabla_{\mu}\xi_{\alpha}\\
				&\st{\Phi(\mc H)}{=} \Sigma(\nu,\xi,\xi) +\dfrac{1}{2}\big(n(\ktt)-X(\elltwo)-\kt(n) n(\elltwo)\big),
			\end{align*}
			where in the second line we used \eqref{12Sigma} contracted with $\nu$. Note that the scalar $\Sigma(\nu,\xi,\xi)$ on $\Phi(\mc H)$ cannot be written in terms of extended hypersurface data and $\k|_{\Phi(\mc H)}$ (see \eqref{23null} and observe the appearance of $\kt^{(2)}$). Now, the contraction of \eqref{commudown} with $e_a^{\alpha}$ gives, after using \eqref{nablanuxiconV},
			\begin{align*}
				[\xi,k^{(\nu)}]_{\alpha}e_a^{\alpha}&= \big({}^{(1)}\nabla\k\big)_{ab}n^b + \k_{\alpha\beta} e_a^{\alpha} \left(-\kappa_n \xi^{\beta} + V^c{}_bn^be_c^{\beta}\right) - \big(X^b+\kt(n)n^b\big) e_b^{\mu} e_a^{\alpha} \nabla_{\mu}\xi_{\alpha}\\
				&= \lie_n\kt_a +\nablacero_a \big(\kt(n)\big) -2\kt(n)\s_a -2P^{bc}\U_{ac}\kt_b  - V^c{}_a \k_{bc}n^b  - \big(\Y_{ba}+\F_{ba}\big)X^b -2\mS_{ab}n^b,
			\end{align*}
where in the second line we inserted \eqref{1Kn} and used $e_b^{\mu} e_a^{\alpha} \nabla_{\mu}\xi_{\alpha} =\frac{1}{2}e_b^{\mu} e_a^{\alpha}\big(\lie_{\xi}g_{\mu\alpha}+(d\bm\xi)_{\mu\alpha}\big)= \Y_{ba}+\F_{ba}$. Using Lemma \ref{lema_vector} applied to $Z=[\xi,k^{(\nu)}]$ and simplifying one then has $$[\xi,k^{(\nu)}]^{\alpha} = \big(\chi+2\lie_n(\kt(n))-2\mS(n,n)\big)\xi^{\alpha}+\big(W^a -2P^{ab}\mS_{bc}n^c + \Sigma(\nu,\xi,\xi)n^a\big) e_a^{\alpha} .$$ Identity \eqref{lie2} follows from this after applying Proposition \ref{propliezetaT} to $T=g$ and $\zeta=W$.
		\end{proof}
	\end{lema}
	
We are ready to show that the tensor $\mS$ satisfies a transport equation along $n$. 

\begin{proof}[Proof of Proposition \ref{prop_lienSigma}]
The proof is based on identity \eqref{id2}, which after pulled back to $\mc{H}$ will be shown to provide the transport equation \eqref{identity} with a right hand side that depends only on extended hypersurface data and $\k_{ab}$, $\kt_a$ and $\ktt$ but no transverse derivatives of $\k_{\alpha\beta}$ on the hypersurface. We start by rewriting the divergence $\nabla_{\mu}\Sigma^{\mu}{}_{\alpha\beta}$ at points on $\mc H$ using \eqref{inversemetric} for $g^{\mu\rho}$ and replacing covariant derivatives by Lie derivatives,
\begin{align*}
	\nabla_{\mu}\Sigma^{\mu}{}_{\alpha\beta} & = g^{\mu\rho}\nabla_{\mu}\Sigma_{\rho\alpha\beta}\\
	&=\big(P^{cd}e_c^{\mu}e_d^{\rho}+\xi^{\mu}\nu^{\rho}+\xi^{\rho}\nu^{\mu}\big)\nabla_{\mu}\Sigma_{\rho\alpha\beta}\\
	&= P^{cd}e_c^{\mu}e_d^{\rho}\nabla_{\mu}\Sigma_{\rho\alpha\beta} +\nu^{\rho}\big(\lie_{\xi}\Sigma_{\rho\alpha\beta} - \Sigma_{\sigma\alpha\beta}\nabla_{\rho}\xi^{\sigma} - 2\Sigma_{\rho\sigma(\alpha}\nabla_{\beta)}\xi^{\sigma}\big) \\
	&\quad\, + \xi^{\rho}\big(\lie_{\nu}\Sigma_{\rho\alpha\beta} - \Sigma_{\sigma\alpha\beta}\nabla_{\rho}\nu^{\sigma} - 2\Sigma_{\rho\sigma(\alpha}\nabla_{\beta)}\nu^{\sigma}\big)\\
	&=  P^{cd}e_c^{\mu}e_d^{\rho}\nabla_{\mu}\Sigma_{\rho\alpha\beta} +\nu^{\rho}\lie_{\xi}\Sigma_{\rho\alpha\beta} + \xi^{\rho}\lie_{\nu}\Sigma_{\rho\alpha\beta}\\
	&\quad\, -\Sigma_{\sigma\alpha\beta}\big(\nu^{\rho}\nabla_{\rho}\xi^{\sigma}+\xi^{\rho}\nabla_{\rho}\nu^{\sigma}\big)-2\Sigma_{\rho\sigma(\alpha}\big(\nu^{\rho}\nabla_{\beta)}\xi^{\sigma}+\xi^{\rho}\nabla_{\beta)}\nu^{\sigma}\big).
\end{align*}
Since the result is independent of the extension of $\nu$ and $\xi$ off $\Phi(\mc H)$ we assume without loss of generality $\nabla_{\xi}\xi=0$ and $\lie_{\xi}\nu=0$, so the pullback of the divergence of $\Sigma$ on $\mc H$ takes the form
\begin{equation}
	\label{LHS}
	\begin{aligned}
		e_a^{\alpha}e_b^{\beta}	\nabla_{\mu}\Sigma^{\mu}{}_{\alpha\beta} & = P^{cd}e_c^{\mu}e_d^{\rho}e_a^{\alpha}e_b^{\beta}\nabla_{\mu}\Sigma_{\rho\alpha\beta} +e_a^{\alpha}e_b^{\beta}\lie_{\xi}\big(\nu^{\rho}\Sigma_{\rho\alpha\beta}\big) + e_a^{\alpha}e_b^{\beta}\lie_{\nu}\big(\xi^{\rho}\Sigma_{\rho\alpha\beta}\big)\\
		&\quad\, -2e_a^{\alpha}e_b^{\beta}\Sigma_{\sigma\alpha\beta}\nu^{\rho}\nabla_{\rho}\xi^{\sigma}-2e_a^{\alpha}e_b^{\beta}\Sigma_{\rho\sigma(\alpha}\big(\nu^{\rho}\nabla_{\beta)}\xi^{\sigma}+\xi^{\rho}\nabla_{\beta)}\nu^{\sigma}\big).
	\end{aligned}
\end{equation}
We now proceed with the explicit computation of each term. The first term in \eqref{LHS} can be evaluated using \eqref{identity1} applied to $T=\Sigma$ and \eqref{symmetries}, namely
\begin{align}
	\hspace{-0.5cm} P^{cd}e_c^{\mu}e_d^{\rho}e_a^{\alpha}e_b^{\beta}\nabla_{\mu}\Sigma_{\rho\alpha\beta} &= P^{cd}\nablacero_c\Sigma_{dab} + (\tr_P\bY)\Sigma_{eab}n^e +2P^{cd}\Y_{c(a|}\Sigma_{de|b)}n^e + (\tr_P\bU) \mS_{ab}+ 2P^{cd}\U_{c(a|}\big({}^{(2)}\Sigma\big)_{d|b)}\nonumber\\
	&= P^{cd}\nablacero_c\Sigma_{dab} + (\tr_P\bY)\Sigma_{eab}n^e +2P^{cd}\Y_{c(a|}\Sigma_{de|b)}n^e  + (\tr_P\bU) \mS_{ab}- 2P^{cd}\U_{c(a}\mS_{b)d}\nonumber\\
	&\quad\,  + 2P^{cd}\U_{c(a}\big(\nablacero_{b)}\kt_d + \kt(n)\Y_{b)d} + \ktt\U_{b)d}- \big(\r-\s\big)_{b)}\kt_d - V^e{}_{b)}\k_{de}\big).\label{1st}
\end{align} 
The third term in \eqref{LHS} is directly
\begin{equation}
	\label{3rd}
	e_a^{\alpha}e_b^{\beta} \lie_{\nu}\big(\xi^{\rho}\Sigma_{\rho\alpha\beta}\big) = \lie_n\mS_{ab}.
\end{equation}
For the three terms in the second line of \eqref{LHS} we use \eqref{nablanuxiconV}, namely $$\nu^{\rho}\nabla_{\rho}\xi^{\sigma} = -\kappa_n \xi^{\sigma} + \left(P^{cd}\big(\r+\s\big)_c+\dfrac{1}{2}n(\elltwo) n^d\right) e_d^{\sigma}$$ and $$\nabla_{e_a}\nu^{\sigma} = \big(P^{cd}\U_{bd}-\big(\r-\s\big)_b n^c\big)e_c^{\sigma},$$ which is a direct consequence of \eqref{connections} and \eqref{derivadan}. Hence,
\begin{align}
	-2e^{\alpha}_a e_b^{\beta}\Sigma_{\sigma\alpha\beta}\nu^{\rho}\nabla_{\rho}\xi^{\sigma} &= 2\kappa_n\mS_{ab} -2P^{cd}\big(\r+\s)_{c}\Sigma_{dab} -n(\elltwo) \Sigma_{cab}n^c,\label{4th}\\
	-2e^{\alpha}_a e_b^{\beta}\Sigma_{\rho\sigma(\alpha}\xi^{\rho}\nabla_{\beta)}\nu^{\sigma}  &= -2\mS_{c(a}\big(P^{cd}\U_{b)d}-\big(\r-\s\big)_{b)}n^c\big),\label{5th}\\
	-2e^{\alpha}_a e_b^{\beta}\Sigma_{\rho\sigma(\alpha}\nu^{\rho}\nabla_{\beta)}\xi^{\sigma} & \st{\eqref{nablaxi1}}{=} -2n^c\big({}^{(2)}\Sigma\big)_{c(a}\big(\r-\s\big)_{b)} - 2n^cP^{de}\Sigma_{cd(a}\big(\Y+\F\big)_{b)e}  - n^c n^d \Sigma_{cd(a} \nablacero_{b)}\elltwo\nonumber\\
	&\st{\eqref{symmetriesn0}}{=} 2n^c\mS_{c(a}\big(\r-\s\big)_{b)} - 2n^cP^{de}\Sigma_{cd(a}\big(\Y+\F\big)_{b)e} - n^c n^d \Sigma_{cd(a} \nablacero_{b)}\elltwo\nonumber\\
	&\quad\, -2\big(\r-\s\big)_{(a}\Big(\nablacero_{b)}\big(\kt(n)\big) - P^{cd}\U_{b)d}\kt_c - V^{d}{}_{b)}\k_{cd}n^c\Big).\label{6th}
\end{align} 
In order to compute the remaining term of \eqref{LHS}, namely $e_a^{\alpha}e_b^{\beta}\lie_{\xi}\big(\nu^{\rho}\Sigma_{\rho\alpha\beta}\big)$, we make use of the second identity in Lemma \ref{lema_Marc} applied to $\zeta=\nu$, 
	\begin{equation}
		\label{nuliexisigma}
		\begin{aligned}
			\lie_{\xi}\left(\nu^{\rho}\Sigma_{\rho\alpha\beta}\right)= \dfrac{1}{2}\lie_{\xi}\left(\lie_{k^{(\nu)}}g_{\alpha\beta}-\lie_{\nu}\k_{\alpha\beta}\right)=\dfrac{1}{2}\lie_{k^{(\nu)}} \lie_{\xi} g_{\alpha\beta} -\dfrac{1}{2}\lie_{\nu}\lie_{\xi}\k_{\alpha\beta} + \dfrac{1}{2}\lie_{[\xi,k^{(\nu)}]}g_{\alpha\beta}.
		\end{aligned}
	\end{equation}
	The pullbacks of the first and third terms have been already computed in Lemma \ref{lemainter}. The pullback of the second term is, by \eqref{aux}, 
	\begin{align*}
		-\dfrac{1}{2}e_a^{\alpha}e_b^{\beta}\lie_{\nu}\lie_{\xi}\k_{\alpha\beta}&=	\lie_n \mS_{ab} - \lie_n\left(\nablacero_{(a}\kt_{b)} + \kt(n)\Y_{ab} + \ktt \U_{ab}\right).
	\end{align*}
	Combining the three,
	\begin{align}
	e_a^{\alpha}e_b^{\beta}\lie_{\xi}\big(\nu^{\rho}\Sigma_{\rho\alpha\beta}\big) &= \k(n,n)\Z^{(2)}_{ab}+ \dfrac{1}{2} \nablacero_{(a}\elltwo \nablacero_{b)}\big(\k(n,n)\big)+ \lie_X\Y_{ab} + 2\r_{(a}\nablacero_{b)}(\kt(n)) + \lie_n \mS_{ab} \nonumber\\
		&\quad\, - \lie_n\left(\nablacero_{(a}\kt_{b)}  + \ktt \U_{ab}\right) + \big(\chi+\lie_n(\kt(n))-2\mS(n,n)\big) \Y_{ab}\label{2nd}\\
		&\quad\, + \ell_{(a} \nablacero_{b)}\big(\chi+2\lie_n(\kt(n))-2\mS(n,n)\big)+ \dfrac{1}{2}\lie_W\gamma_{ab}- \lie_{\Xi}\gamma_{ab} +\Sigma(\nu,\xi,\xi)\U_{ab} .\nonumber
	\end{align}
	With this, all the terms in \eqref{LHS} have been evaluated. Inserting the corresponding expressions \eqref{1st}, \eqref{3rd}, \eqref{4th}, \eqref{5th}, \eqref{6th} and \eqref{2nd} into \eqref{LHS} the pullback of the divergence term is, finally,
	\begin{align*}
		e_a^{\alpha}e_b^{\beta}	\nabla_{\mu}\Sigma^{\mu}{}_{\alpha\beta} & = P^{cd}\nablacero_c\Sigma_{dab} + \big(\tr_P\bY-n(\elltwo)\big)\Sigma_{cab}n^c +2P^{cd}\Y_{c(a|}\Sigma_{de|b)}n^e  + \big(\tr_P\bU+2\kappa_n\big) \mS_{ab}\nonumber\\
		&\quad\,- 4P^{cd}\U_{c(a}\mS_{b)d}  + 2P^{cd}\U_{c(a}\big(\nablacero_{b)}\kt_d + \kt(n)\Y_{b)d} + \ktt\U_{b)d} - V^e{}_{b)}\k_{de}\big) +2\lie_n \mS_{ab} \\
		&\quad\, - \lie_n\left(\nablacero_{(a}\kt_{b)}  + \ktt \U_{ab}\right)+ \k(n,n)\Z^{(2)}_{ab}+ \dfrac{1}{2} \nablacero_{(a}\elltwo \nablacero_{b)}\big(\k(n,n)\big) + \lie_X\Y_{ab} + \dfrac{1}{2}\lie_W\gamma_{ab}\\
		&\quad\,  -\lie_{\Xi}\gamma_{ab} - 2\ell_{(a}\nablacero_{b)}\big(\mS(n,n)\big) +\Sigma(\nu,\xi,\xi)\U_{ab} + \big(\chi+\lie_n(\kt(n))-2\mS(n,n)\big) \Y_{ab}\nonumber\\
		&\quad\, + \ell_{(a} \nablacero_{b)}\big(\chi+2\lie_n(\kt(n))\big)  -2P^{cd}\big(\r+\s)_{c}\Sigma_{dab} +4\big(\r-\s\big)_{(a}\mS_{b)c}n^c  - 2n^cP^{de}\Sigma_{cd(a}\big(\Y+\F\big)_{b)e}\nonumber\\
		&\quad\,  - n^c n^d \Sigma_{cd(a} \nablacero_{b)}\elltwo +2\s_{(a}\nablacero_{b)}\kt(n) +2 \big(\r-\s\big)_{(a}V^{d}{}_{b)}\k_{cd}n^c.
	\end{align*}

	Now we consider the pullback of the RHS of \eqref{id2}. Using Proposition \ref{proppullback} [Eq. \eqref{hess}] applied to $f=\tr_g\k$ and taking into account that $\tr_g \k\st{\mc H}{=} \tr_P\k+2\kt(n)$ (cf. \eqref{inversemetric}),
	\begin{align*}
		e_a^{\alpha} e_b^{\beta}\left(\lie_{\eta}R_{\alpha\beta} + \dfrac{1}{2}\nabla_{\alpha}\nabla_{\beta} \tr_g\k\right) & = \big(\lie_{\eta}R\big)_{ab} + \dfrac{1}{2}\nablacero_a\nablacero_b \big(\tr_P\k + 2\kt(n)\big) +\dfrac{1}{2} n\big(\tr_P\k + 2\kt(n)\big) \Y_{ab}  \\
		&\quad\, + \dfrac{1}{2}{\xi}\big(\tr_g\k\big)\U_{ab}.
	\end{align*}
Let us compute ${\xi}\big(\tr_g\k\big)$. From the decomposition of $g^{\alpha\beta}$ in \eqref{inversemetric} together with ${}^{(1,2)}\nabla\k  = {}^{(1,3)}\Sigma +{}^{(2,3)}\Sigma$ (see \eqref{S2}), $$\lie_{\xi}\big(\tr_g\k\big)  = g^{\alpha\beta}\xi^{\mu}\nabla_{\mu}\k_{\alpha\beta}=P^{ab}\big({}^{(1)}\nabla\k\big)_{ab} + 2\big({}^{(1,2)}\nabla\k\big)_a n^a = P^{ab}\big({}^{(1)}\nabla\k\big)_{ab} + 2\Sigma(\xi,\nu,\xi)+2\Sigma(\nu,\xi,\xi).$$ Inserting \eqref{1K} and \eqref{12null} (with $\ntwo=\beta=a_{\para}=0$), but keeping the term $\Sigma(\nu,\xi,\xi)$ as it depends on transverse derivatives of $\k_{\alpha\beta}$ (because of the term $\kt^{(2)}(n)$ in \eqref{23null}) gives
	\begin{align*}
{\xi}\big(\tr_g\k\big)&=2P^{ab}\nablacero_a\kt_b  +\big(2\tr_P\bY-n(\elltwo)\big)\kt(n) +2\big(\tr_P\bU+\kappa_n\big)\ktt \\
		&\quad\, - 4P^{ab}\r_a\kt_b -2P^{ab}V^c{}_a\k_{bc}- 2\tr_P\mS   +n(\ktt)+ 2\Sigma(\nu,\xi,\xi),
	\end{align*}
	and thus, the pullback of the RHS of \eqref{id2} is
	\begin{align*}
		\mbox{RHS}_{ab} & = \big(\lie_{\eta}R\big)_{ab} + \dfrac{1}{2}\nablacero_a\nablacero_b \big(\tr_P\k + 2\kt(n)\big) +\dfrac{1}{2} n\big(\tr_P\k + 2\kt(n)\big) \Y_{ab}  \\
		&\quad\, +\Bigg(P^{cd}\nablacero_c\kt_d +\left(\tr_P\bY-\dfrac{1}{2}n(\elltwo)\right)\kt(n) +\big(\tr_P\bU+\kappa_n\big)\ktt - 2P^{cd}\r_c\kt_d \\
		&\qquad\qquad\, -P^{cd}V^f{}_c\k_{df}- \tr_P\mS+\dfrac{1}{2}n(\ktt) + \Sigma(\nu,\xi,\xi)\Bigg)\U_{ab}.
	\end{align*}
	Identity \eqref{identity} follows by equating the LHS and the RHS and rearranging terms. Observe that the only term that is not expressible in terms of extended hypersurface data and $\k|_{\Phi(\mc H)}$ (namely $\Sigma(\nu,\xi,\xi)$) cancels out.
\end{proof}
	
\section{Useful contractions}
\label{app_conn}

In this appendix we compute several contractions that are needed in the main text of the paper. We start with all possible contractions of the tensor $\nablacero\bm\theta$ with $n$ for any one-form $\bm\theta$.

\begin{lema}
	\label{lemmannabla}
Let $\{\mc H,\bg,\bm\ell,\elltwo\}$ be metric hypersurface data and $\bm\theta$ any one-form. Then,
\begin{align}
\hspace{-0.15cm}n^b\nablacero_a\theta_b &=\nablacero_a\big(\bm\theta(n)\big)  -\bm\theta(n)\s_a- P^{bc}\U_{ac}\theta_b + \ntwo \big(\bm\theta(n)\nablacero_a\elltwo + P^{bc}\F_{ac}\theta_b\big),\label{contra1}\\
\hspace{-0.15cm}n^b\nablacero_b \theta_a &=\lie_n\theta_a -\bm\theta(n)\s_a- P^{bc}\U_{ac}\theta_b + \ntwo \big(\bm\theta(n)\nablacero_a\elltwo + P^{bc}\F_{ac}\theta_b\big),\label{contra2}\\
\hspace{-0.15cm}2n^b \nablacero_{(a}\theta_{b)}& = \lie_n\theta_a + \nablacero_a\big(\bm\theta(n)\big) - 2\big(\bm\theta(n)\s_a + P^{bc}\U_{ac}\theta_b\big) + 2\ntwo\big(\bm\theta(n) \nablacero_a\elltwo + P^{bc}\F_{ac}\theta_b\big),\label{nnablacerotheta}\\
\hspace{-0.15cm}n^an^b \nablacero_{a}\theta_{b} &= \lie_n\big(\bm\theta(n)\big)  + \ntwo\bm\theta(n) n(\elltwo) + P^{bc}\theta_b \left(2\ntwo \s_c -\dfrac{1}{2}(\ntwo)^2 \nablacero_c\elltwo - \dfrac{1}{2} \nablacero_c\ntwo\right).\label{nnablacerothetan}
\end{align}
\begin{proof}
From \eqref{derivadan} it follows $$n^b\nablacero_a\theta_b = \nablacero_a\big(\bm\theta(n)\big) - \theta_b\nablacero_a n^b = \nablacero_a\big(\bm\theta(n)\big)  -\bm\theta(n)\s_a- P^{bc}\U_{ac}\theta_b + \ntwo \big(\bm\theta(n)\nablacero_a\elltwo + P^{bc}\F_{ac}\theta_b\big),$$ and $$ n^b\nablacero_b \theta_a = \lie_n\theta_a - \theta_b \nablacero_a n^b = \lie_n\theta_a -\bm\theta(n)\s_a- P^{bc}\U_{ac}\theta_b + \ntwo \big(\bm\theta(n)\nablacero_a\elltwo + P^{bc}\F_{ac}\theta_b\big),$$ which are \eqref{contra1}-\eqref{contra2}. Adding them, \eqref{nnablacerotheta} follows. Another contraction with $n$ then gives $$n^an^b \nablacero_{a}\theta_{b} = \lie_n\big(\bm\theta(n)\big) -P^{bc}\U_{ac}\theta_b n^a + \ntwo\big(\bm\theta(n) n(\elltwo)+P^{bc}\s_c\theta_b\big),$$ which becomes \eqref{nnablacerothetan} after using \eqref{Un}.
\end{proof}
\end{lema}

The following two lemmas are only needed in the null case, so from now on we restrict to null metric hypersurface data. In the first one we recall a result from \cite{Mio3} where the commutator $[P^{ab},\lie_n]$ acting on a 2-covariant, symmetric tensor field $T_{ab}$ is computed. 
\begin{lema}
	\label{lemalieP}
	Let $\{\mc H,\bg,\bm\ell,\elltwo\}$ be null metric hypersurface data. Then, 
			\begin{equation}
				\label{lienP}
				\lie_n P^{ab} = -2 \left(P^{ac}n^b + P^{bc}n^a\right)2\s_c -2P^{ac}P^{bd}\U_{cd} - n^an^bn\big(\elltwo\big) .
			\end{equation}
		As a consequence, for any (0,2) symmetric tensor field $T_{ab}$ it holds
		\begin{equation}
			\label{lietrPY}
			P^{ab}\lie_n \T_{ab} = \lie_n(\tr_P\bT) + 4P(\bm t,\bs)+2P^{ac}P^{bd}\U_{cd}T_{ab} + n(\elltwo)\bm t(n),
		\end{equation}	
		where $\bm t\d \bT(n,\cdot)$. 
\end{lema}

To conclude the appendix we compute the contraction of $\lie_X\gamma_{ab}$ with $n$ when $X$ is either of the form $X^b=P^{bc}\theta_c$ or $X=fn$. 

\begin{lema}
	\label{lemagammalie}
Let $\{\mc H,\bg,\bm\ell,\elltwo\}$ be null metric hypersurface data, $\bm\theta$ any one-form and $f$ any smooth function. Then,
\begin{align}
n^b\lie_{P(\bm\theta,\cdot)}\gamma_{ab} &= \lie_n\theta_a - \lie_n\big(\bm\theta(n)\big) \ell_a -2 \big(\bm\theta(n)\s_a + P^{bc}\U_{ab}\theta_c\big), \label{nliegamma}\\
n^b\lie_{fn}\gamma_{ab} &=0.\label{nliegamma2}
\end{align}
\begin{proof}
We start by noting that $\bg(n,\cdot)=0$ immediately implies for any vector field $X$ 
\begin{align*}
	n^b\lie_X\gamma_{ab} = -\gamma_{ab}\lie_X n^b= \gamma_{ab}\lie_n X^b.
\end{align*}
When $X$ is of the form $X^b=P^{bc}\theta_c$ this becomes
\begin{align*}
n^b\lie_{P(\bm\theta,\cdot)}\gamma_{ab} & = \gamma_{ab}\theta_c\lie_nP^{bc} + \gamma_{ab}P^{bc}\lie_n\theta_c\\
\hspace{-0.1cm}&\st{\eqref{Pgamma}}{=} \gamma_{ab}\theta_c\lie_nP^{bc} + \lie_n\theta_a - \lie_n\big(\theta(n)\big)\ell_a\\
\hspace{-0.2cm}&\st{\eqref{lienP}}{=}-2\gamma_{ab}P^{bd}\big( n^c\theta_c \s_d + P^{cf}\U_{df}\theta_c\big) + \lie_n\theta_a - \lie_n\big(\theta(n)\big)\ell_a,
\end{align*}
which after using again $\gamma_{ab}P^{bd} = \delta^d_a-n^d\ell_a$, $\bs(n)=0$ and $\bU(n,\cdot)=0$ gives \eqref{nliegamma}. When $X=fn$ then $n^b \lie_{fn}\gamma_{ab} =\gamma_{ab} \lie_n (fn^b) = 0$, which is \eqref{nliegamma2}.
\end{proof}
\end{lema}

\end{appendices}

\begingroup
\let\itshape\upshape

\renewcommand{\bibname}{References}
\bibliographystyle{acm}
\bibliography{biblio} 

\end{document}